%% file: main.tex
\documentclass[a4paper,10pt,leqno]{amsart}
\usepackage{amsmath}
\usepackage{amsthm}
\usepackage{amssymb}
\usepackage{stmaryrd}
\usepackage{url}
\usepackage{here}
\usepackage[all]{xy}
\usepackage{comment}
\usepackage{mathtools}
\usepackage{ascmac}
\usepackage{ulem}
\usepackage{framed}
\usepackage{graphicx}
\usepackage{nicematrix}
\usepackage{lmodern}
\usepackage{caption}
\usepackage{subcaption}
\usepackage{footnote}
\usepackage{amscd}
\usepackage{mathdots}
\usepackage{titlesec}
\usepackage{titletoc}
\mathtoolsset{showonlyrefs}

\usepackage{soul}

\usepackage{xcolor}
\definecolor{webgreen}{rgb}{0,.5,0}
\definecolor{webbrown}{rgb}{.6,0,0}
\definecolor{RoyalBlue}{cmyk}{1, 0.50, 0, 0}
\usepackage[colorlinks=true, breaklinks=true, urlcolor=webbrown, linkcolor=RoyalBlue, citecolor=webgreen, backref=page]{hyperref}

\titlecontents{section} % <section type>
[0pt]                 % <left indent>
{\bfseries}            % <above code>
{\thecontentslabel\quad} % <numbered entry format>
{}                    % <numberless entry format>
{\titlerule*[0.5em]{.}\contentspage} % <filler page format>

\titlecontents{subsection} % <section type>
[5pt]                 % <left indent>
{}            % <above code>
{\thecontentslabel\quad} % <numbered entry format>
{}                    % <numberless entry format>
{\titlerule*[0.5em]{.}\contentspage} % <filler page format>

\titlecontents{subsubsection} % <section type>
[10pt]                 % <left indent>
{}            % <above code>
{\thecontentslabel\quad} % <numbered entry format>
{}                    % <numberless entry format>
{\titlerule*[0.5em]{.}\contentspage} % <filler page format>

\titlecontents{paragraph} % <section type>
[15pt]                 % <left indent>
{}            % <above code>
{\thecontentslabel\quad} % <numbered entry format>
{}                    % <numberless entry format>
{\titlerule*[0.5em]{.}\contentspage} % <filler page format>

\titlespacing*{\section}{0pt}{5.5ex plus 1ex minus .2ex}{4.3ex plus .2ex}
\numberwithin{equation}{subsection}

\titleformat{\section}{\centering\normalfont\large\scshape}{\thesection}{1em}{}

\titleformat{\subsection}{\normalfont\normalsize\bfseries}{\thesubsection}{1em}{}

\titleformat{\subsubsection}{\normalfont\normalsize\bfseries}{\thesubsubsection}{1em}{}

\titleformat{\paragraph}{\normalfont\normalsize\it}{\theparagraph}{1em}{}

\theoremstyle{plain}
\newtheorem{thm}{Theorem}[section]
\newtheorem{lem}[thm]{Lemma}
\newtheorem{prop}[thm]{Proposition}
\newtheorem{cor}{Corollary}
\newtheorem{RHP}{Riemann-Hilbert Problem}

\theoremstyle{definition}

\theoremstyle{remark}
\newtheorem*{rem}{Remark}
\newtheorem*{note}{Note}

\newcommand{\N}{\mathbb N}
\newcommand{\C}{\mathbb C}   
\newcommand{\R}{\mathbb R}

\newcommand{\Z}{\mathbb Z}

\DeclareMathOperator{\res}{res}

\DeclareMathOperator{\diag}{diag}

\newcommand{\bp}{\begin{pmatrix}} 
\newcommand{\ep}{\end{pmatrix}} 
\renewcommand{\O}{\mathcal{O}}
\renewcommand{\Re}{\operatorname{Re}}
\renewcommand{\Im}{\operatorname{Im}}

\begin{document}

\title[]{Connection formulae for the radial Toda equations I}

\author[Guest]{Martin A. Guest}
\address{(M.A. Guest) Department of Mathematics, Faculty of Science and Engineering, Waseda University, 3-4-1 Okubo, Shinjuku, Tokyo 169-8555 JAPAN}
\email{martin@waseda.jp}

\author[Its]{Alexander R. Its}
\address{(A.R. Its) Department of Mathematical Sciences, 
Indiana University Indianapolis, 402 N. Blackford St., Indianapolis, IN 46202 USA}
\email{aits@iu.edu}

\author[Kosmakov]{Maksim Kosmakov}
\address{(M. Kosmakov) Department of Mathematical Sciences, University of Cincinnati, P.O. Box 210025, Cincinnati, OH 45221, USA}
\email{kosmakmm@ucmail.uc.edu}

\author[Miyahara]{Kenta Miyahara$^\ast$}
\thanks{$^\ast$ Corresponding author.}
\address{(K. Miyahara) Department of Mathematical Sciences, 
Indiana University Indianapolis, 402 N. Blackford St., Indianapolis, IN 46202 USA}
\email{kemiya@iu.edu}

\author[Odoi]{Ryosuke Odoi}
\address{(R. Odoi) Department of Pure and Applied Mathematics, Faculty of Science and Engineering, Waseda University, 3-4-1 Okubo, Shinjuku, Tokyo 169-8555 JAPAN}
\email{ryosuke.odoi@moegi.waseda.jp}

\begin{abstract}
    This paper is the first in a forthcoming series of works where the authors study the global asymptotic behavior of the radial solutions of the 2D periodic Toda equation of type $A_n$. The principal issue is the connection formulae between the asymptotic parameters describing the behavior of the general solution at zero and infinity. To reach this goal we are using a fusion of the Iwasawa factorization in the loop group theory and the Riemann-Hilbert nonlinear steepest descent method of Deift and Zhou which is applicable to 2D Toda in view of its Lax integrability. A principal technical challenge is the extension of the nonlinear steepest descent analysis to Riemann-Hilbert problems of matrix rank greater than $2$. In this paper, we meet this challenge for the case $n=2$ (the rank $3$ case) and it already captures the principal features of the general $n$ case.
\end{abstract}

\date{\today}

\let\ds\displaystyle

\maketitle

\setcounter{tocdepth}{3}
\tableofcontents

%%%%%%%%%%%%%%%%%%%%%%%%%%%%%%%%%%%%%%%%%%%%%%%%%%%%%%%%%%%%%%%%%%%%%%%%%%%%%
\setlength{\parskip}{6pt}

\input{section1.tex}

\input{section2.tex}

\input{section3.tex}

\input{section4.tex}

\input{Appendix.tex}

%%%%%%%%%%%%%%%%%%%%%%%%%%%%%%%%%%%%%%%%%%%%%%%%%%%%%%%%%%%%%%%%%%%%%%%%%%%%%

\bibliographystyle{alpha}
\bibliography{reference.bib}

\end{document}

%% file: section1.tex
\section{Introduction and Main Results} \label{intro}

This paper is the first in a forthcoming series of works where the authors study the global asymptotic behavior of the radial solutions of the 2D periodic Toda equation of type $A_n$, $n \in \N$ with $\epsilon$ sign:
\begin{equation}\label{ost}
2(w_i)_{t\bar{t}}=\epsilon(e^{2(w_{i+1}-w_{i})} - e^{2(w_{i}-w_{i-1})}), \quad \epsilon = \pm 1, \quad i \in \Z,
\end{equation}
where $w_i: \C \setminus \{0\} \to \R$ and subject to the conditions 
\[
\begin{cases}
\ \ w_i = w_{i+n+1} \quad (\text{periodicity}),\\
\ \  w_i=w_i(\vert t\vert) \quad (\text{radial condition}),\\
\ \  w_i +  w_{n-i}=0\quad (\text{anti-symmetry}).
\end{cases}
\]

Since Mikhailov-Olshanetsky-Perelomov \cite{Mik, MOP} proposed the two dimensional generalization of the classical Toda lattice equation on the affine root systems, the 2D Toda equations have been an important integrable system with many aspects.
Among various boundary conditions and underlying root systems, our focus is on the 2D periodic Toda equation of the affine root system $A_n$, which has been studied extensively from both mathematical and physical perspectives.
When $\epsilon = 1$ in \eqref{ost}, for example, it can be interpreted as the equation for primitive harmonic maps taking values in a compact flag manifold (see \cite{BurPed, BolPedWood}). Such maps are closely related to harmonic maps into symmetric spaces.
These, in turn, have numerous geometrical interpretations, such as surfaces in $\R^3$ of constant mean curvature (see \cite{Dorf}) or special Lagrangian cones in $\C^3$ (see \cite{McIntosh03, Joyce}). In contrast, when $\epsilon = -1$ in \eqref{ost}, it is an example of the tt* (topological—anti-topological fusion) equations introduced by Cecotti and Vafa \cite{CV1, CV2, CV3} to describe certain deformations of super-symmetric quantum field theories. In their context, imposing radial symmetry and anti-symmetry for the solution of \eqref{ost} is natural.
Due to its Toda structure, \eqref{ost} with $\epsilon = -1$ is known as the \textit{tt*-Toda equation}.

The project is a continuation of the works \cite{GuLi12, GuLi14, GIL1, GIL2, GIL3, GIL4} devoted to the case $\epsilon = -1$.
%This case was specifically singled out in the seminal articles \cite{CV1, CV2, CV3} of Cecotti and Vafa and it is known as the tt*-Toda equations. 
Of special importance are the {\it global solutions} of the tt*-Toda equation first predicted by Cecotti and Vafa, i.e., the solutions which are smooth for all $0<|t| < \infty$.
A comprehensive analysis of these solutions was performed in \cite{GuLi12} - \cite{GIL4}.
Earlier, a class of the solutions of the tt*-Toda equation was introduced and their behavior at $t=0$ and $t=\infty$ had been evaluated by Tracy and Widom in \cite{TW}. As it follows from the comparison of the results of \cite{TW} and \cite{GuLi12} - \cite{GIL4}, the Tracy-Widom solutions turned out to be exactly the Cecotti-Vafa global solutions of \eqref{ost}.

The ultimate goal of the current project is to study the general radial solution of \eqref{ost} with $\epsilon = -1$. The principal issue is the {\it connection formulae} between the asymptotic parameters describing the behavior of the solution at $t=0$ and $t=\infty$. In \cite{GuLi12} - \cite{GIL4} the problem was solved for the global solutions. The technique used was a fusion of the Iwasawa factorization in the loop group theory and the Riemann-Hilbert nonlinear steepest descent method of Deift and Zhou \cite{DZ} which is applicable to \eqref{ost} in view of its Lax-integrability (see more details in the next section).

The extension of the methods of \cite{GuLi12} - \cite{GIL4} to the general families of solutions is not trivial. The main problem is that the generic solution of  \eqref{ost} with $\epsilon =-1$ is singular in the neighborhood of both $t=0$ and $t=\infty$. This fact creates very serious technical challenges for the Riemann-Hilbert approach. Hence the idea is to try to obtain first the connection formulae for the general solution of \eqref{ost} with $\epsilon = 1$. The obvious reason is that for this sign of $\epsilon$, every solution is smooth for all $0<|t| < \infty$ (again, the details are given in the next section). 
 
In this paper, we shall consider the case $\epsilon = 1$ and  $n=2$. The last restriction is not of principal importance. Indeed, the associated Riemann-Hilbert problem, as we will see, is of matrix rank 3. That is, we are already beyond the usual and well-developed rank 2 Riemann-Hilbert setting, and the $n=2$ case actually captures the principal features of the general $n$ case. At the same time, we think that it makes sense to first present all the details assuming  $n=2$.  Indeed, the transition to $n >2$, which we will do in a subsequent paper of the series, will be then much easier to follow.  
This is the same strategy as the one used in \cite{GuLi12} - \cite{GIL4}: first to consider  $n=3$  \cite{GuLi12} - \cite{GIL3} and then extend the analysis to general $n$ \cite{GIL4}.

Introduce the variable,
$$
x := |t| = \left( t\overline{t} \,   \right)^{\frac{1}{2}}.
$$
Then, equation \eqref{ost} becomes
\begin{equation}\label{ost1}
(w_i)_{xx} + \frac{1}{x} (w_i)_x = \epsilon (2e^{2(w_{i+1}-w_{i})} -2 e^{2(w_{i}-w_{i-1})}), \quad i\in \Z.
\end{equation}
When $n=1$, the periodicity condition implies that there are only two functions involved: $w_0$ and $w_1$.
Moreover, the anti-symmetry yields $w_1 = - w_0$, and hence the system (\ref{ost1}) reduces to
a single equation on $w_0$,
$$
(w_0)_{xx} + \frac{1}{x} (w_0)_x = \epsilon (2e^{-4w_{0}} -2 e^{4w_{0}})
$$
or
\begin{equation}\label{ost2}
(w_0)_{xx} + \frac{1}{x} (w_0)_x = -4\epsilon \sinh (4w_0).
\end{equation}
This is a particular case of the third Painlev\'e equation whose connection problem for a special global solution, in the case  $\epsilon =-1$, was first solved by McCoy, Tracy, and Wu in \cite{MTW} and was used there to describe in detail the transition regime in the Ising model. A complete solution of the connection problem for the general solution of \eqref{ost2} was obtained in the series of works \cite{Novok2, Novok3, Kitaev2, Novok07boutrouxforSGP3} (see the monograph \cite{FIKN} for more on the history of the question).

In the case $n=2$, which is the subject of this paper, one has
$$
w_1 = 0, \quad w_2= -w_0,
$$
and hence the system \eqref{ost1} again becomes a single equation for $w_0$,
\begin{align}
    (w_0)_{xx} + \frac{1}{x} (w_0)_x = \epsilon (2e^{-2w_{0}} - 2e^{4w_{0}}). \label{radial 2D periodic toda with epsilon}
\end{align}
This equation is known as the radial version of the Bullough-Dodd equation or Tzitzeica equation (see \cite{LofMcI} and references therein), and it is yet another particular case of the Painlev\'e III equation. 
The connection formulae for this equation have been studied by A. Kitaev in \cite{Kitaev}. (At the end of this introduction, we will say more about Kitaev's results.)
When $\epsilon = 1$ in \eqref{radial 2D periodic toda with epsilon}, we obtain
\begin{equation}\label{negative tt*-Toda with x when n=2}
(w_0)_{xx} + \frac{1}{x} (w_0)_x = 2e^{-2w_{0}} -2 e^{4w_{0}},
\end{equation}
which we call the \textit{radial Toda equation} (with $n=2$) in our context.

Our main result is the following theorem:
\begin{thm}\label{main theorem} 
For every $\gamma \in (-1/2, 1)$ and every $\rho \in \R$, there exists
a unique real-valued, smooth for all $x>0$,  solution of equation (\ref{negative tt*-Toda with x when n=2}) 
such that
\begin{equation}\label{zero}
w_0(x) = \gamma \ln x + \rho + o(1), \quad x \rightarrow 0.
\end{equation}
The large $x$ behavior of this solution is described by the asymptotic formula,
\begin{equation}\label{infty}
w_0(x) = \frac{\sigma}{\sqrt{x}} \cos\left(2\sqrt{3}x  +\frac{2}{\sqrt{3}}\sigma^2\ln x + \psi \right) + O\left(\frac{1}{x}\right),
\quad x \rightarrow \infty.
\end{equation}

The {\it connection formulae}, i.e. the expression of $\sigma$ and $\psi$ in terms of
$\gamma$ and $\rho$, are given by the equations
\begin{align}
\sigma^2 &= \frac{\sqrt{3}}{2} X,\quad \sigma >0, \label{con1} \\
\psi &= \frac{2 \sigma^2}{\sqrt{3}} \ln (24 \sqrt{3}) +\frac{3\pi}{4} + \alpha - \arg \Gamma\left(\frac{2i}{\sqrt{3}}\sigma^2\right) \label{con2}
\end{align}
where
\begin{align}\label{X}
\begin{aligned}
    X &= \dfrac{1}{2\pi}\ln \Bigg[ \dfrac{1}{8 \cos \frac{\pi (1 - \gamma)}{3} \sin^2 \frac{\pi (1 - \gamma)}{3}} \Bigg( q^{\R} + \frac{1}{q^{\R}} \Bigg) + \dfrac{1}{4\sin^2\frac{\pi (1 - \gamma)}{3} } \Bigg]
\end{aligned}
\end{align}
and
\begin{align}\label{alpha}
\begin{aligned}
    \alpha &= \arg\Bigg[ \Big( q^{\R} + \frac{1}{q^{\R}} \Big) \cos\frac{2\pi (1 - \gamma)}{3} + 2\cos \frac{\pi (1 - \gamma)}{3}\\
    &\qquad \qquad \qquad \qquad\qquad \qquad + i \Big( q^{\R} - \frac{1}{q^{\R}} \Big) \sin \frac{2\pi (1 - \gamma)}{3} \Bigg]
\end{aligned}
\end{align}
with 
\begin{align*}
    q^{\R} = \frac{2 (\gamma - 1)^2 }{e^{-2\rho}} 3^{2 (\gamma - 1)} \frac{\Gamma \left( \frac{\gamma - 1}{3} \right) \Gamma \left( \frac{2 \gamma - 2}{3} \right) }{ \Gamma \left( \frac{2 - 2 \gamma }{3} \right) \Gamma \left( \frac{1 - \gamma}{3} \right) }.
\end{align*}

Moreover, every solution of \eqref{negative tt*-Toda with x when n=2} can be characterized by the asymptotic behavior at $x=0$
given by \eqref{zero} for some $\gamma \in (-1/2, 1)$ and $\rho \in \R$ (and, consequently,
behaving at infinity according to \eqref{infty} for some real $\sigma$ and $\psi$).
\end{thm}

As already mentioned above, the 2D periodic Toda equation is Lax-integrable due to Mikhailov \cite{Mik}. 
In particular, this means that its radial version  \eqref{ost1} describes the isomonodromic deformations of a certain $(n+1)\times(n+1)$ system of linear differential equations.
In the case $n=2$,  which we study in this paper,  this is a $3\times 3$ linear system \eqref{Lax pair with x} presented in the next section. The monodromy data of this linear system constitute the first integrals of the nonlinear equation (\ref{negative tt*-Toda with x when n=2}). The connection problem is solved as soon as the asymptotic behaviors of $w_0(x)$ at $x = \infty$ and $x = 0$ are explicitly described in terms of these integrals.
We solve the first problem by obtaining the explicit formulae for the parameters $X$ and $\alpha$ in terms of the monodromy data of the linear system \eqref{Lax pair with x}, in section \ref{section 3}, using the extension of the Deift-Zhou steepest descent method to the associated  $3\times 3$ matrix Riemann-Hilbert problem. The second problem is the derivation of the explicit formulae for the parameters $\gamma$ and $\rho$ in terms of the {\it same} monodromy data that is done in section \ref{section 4} with the help of Iwasawa Factorization. 
Combining these yields the {\it direct} connection formulae between the asymptotic parameters $(\sigma, \psi)$ at infinity and the asymptotic parameters $(\gamma, \rho)$ at zero which are given in Theorem \ref{main theorem}. Section \ref{section 4} essentially borrows the technique already developed in \cite{GIL3}.  Section \ref{section 3} represents our main technical development - a higher-rank version of the Deift-Zhou nonlinear steepest descent method which we will use in the following paper of this series where the case of general $n$ will be studied.

The fact that we have accounted for all solutions of \eqref{negative tt*-Toda with x when n=2}, i.e. the proof of the last statement of Theorem \ref{main theorem}, follows from the above computations --- this is explained at the end of section \ref{section 4}.

The ``first half'' of Theorem \ref{main theorem}, i.e. the asymptotic formulae \eqref{infty}, \eqref{con1}, \eqref{con2} at $x =\infty$  with the parameters $X$ and $\alpha$ given in terms of the monodromy data of the linear system \eqref{Lax pair with x} had already been obtained in \cite{Kitaev}.
Also, in \cite{Kitaev} a complete connection formula for the one-parameter family of the global solution in the case $\epsilon = -1$ was presented. 
The important difference between our approach and the one by \cite{Kitaev} is that the latter is based on the WKB asymptotic solution of the associated direct monodromy problem. 
Unlike \cite{Kitaev}, our paper gives an alternative derivation and proof of the connection formulae based on the asymptotic solution of the inverse monodromy problem via the Deift-Zhou nonlinear steepest descent method. As already stressed, a key methodological point of our work is that in order to prove Theorem \ref{main theorem}, we need to develop an extension of the Deift-Zhou method in the case of a Riemann-Hilbert problem whose matrix rank is higher than 2.

In the smooth case, with $\epsilon =1$, a complete description of the connection formulae had been obtained in \cite{KitV1}. In fact, in this work a much broader class of the complex and generally singular solutions of (\ref{negative tt*-Toda with x when n=2}) is studied\footnote{The reader can find a detailed summary of the results of  \cite{Kitaev} and \cite{KitV1} in Appendices B, C, and D of the recent paper \cite{KV23}. 
In this paper, the authors have corrected some typos and small arithmetic mistakes contained in their earlier works, and they have made considerable efforts to check numerically and present in a more transparent and simplified way the asymptotic formulae of \cite{Kitaev} and \cite{KV23}.}.
The authors of \cite{KitV1} use an alternative, $2\times2$ Lax pair for (\ref{negative tt*-Toda with x when n=2}) and they apply again the WKB-based technique of solution of the associated direct monodromy problem. The $2\times2$ Lax pair used in \cite{KitV1} has no analog for the general $n$ case. Hence the methods of \cite{KitV1} can not be immediately extended to the radial Toda equation for general $n$, while our $3\times 3$ based approach is easily, at least in principle, generalized to the $n \times n$ case. 
It should also be mentioned that matching our Theorem \ref{main theorem} with the formulae of \cite{KitV1} is an outstanding issue since in \cite{KitV1} a different monodromy parametrization is used, and the authors do not extract from their formulae the exact equations directly relating the asymptotic parameters at $0$ and $\infty$. Surely, this, though necessarily somewhat cumbersome, would be possible to carry out if needed.

Our paper is, of course, not the first one where the nonlinear steepest descent method has been extended to the higher-rank Riemann-Hilbert setting. In particular, close to our Riemann-Hilbert problem  (but not coinciding with it), some $3 \times 3$ Riemann-Hilbert problems have been analyzed in the papers \cite{BS}, \cite{BLS}, \cite{CL}, and \cite{CLW}.

%The Riemann-Hilbert problems considered in these works are different from the one which is analyzed in our paper; indeed, their solutions are supposed to have only one irregular singular point, while the solution to our problem must have two.

Acknowledgments: 
The authors would like to thank all the referees and editors for their careful reading and for their numerous comments and suggestions, which greatly improved the quality of this paper.
The first author was partially supported by JSPS grant 18H03668. 
The second author was partially supported by NSF grant DMS:1955265, by RSF grant No. 22-11-00070, and by a Visiting Wolfson Research Fellowship from the Royal Society. 
The fourth author was partially supported by the ITO Foundation for International Education Exchange. 
The authors thank Andrei Prokhorov for valuable comments concerning the asymptotic analysis of the $Y$-RH problem discussed in section \ref{section Y problem}.

Some of the results of this work concerning the large $x$ asymptotics
have already been presented in the Ph.D. thesis \cite{Kosmakov} of the third author. 
Some figures from that thesis are reproduced here, with the permission of the author.
%Some figures used in this paper already appeared in this thesis, but we have permission from him to reuse those figures here.

%% file: section2.tex
\section{Direct Monodromy Problem}

In this section, we will show that every solution of the radial Toda equation \eqref{negative tt*-Toda with x when n=2} has the corresponding monodromy data. To this end, we give the Lax pair representation of \eqref{negative tt*-Toda with x when n=2} and then describe all the monodromic properties.

Before getting into this point, we shall prove that \eqref{negative tt*-Toda with x when n=2} is equivalent to a certain Painlev\'e III equation. Indeed, if we introduce the change of variables 
\begin{align*}
\begin{dcases}
   x = \frac{3}{4}s^{\frac{2}{3}}\\
   \Tilde{w} = s^{\frac{1}{3}}e^{-2w_0},
\end{dcases}
\end{align*}
then one can readily verify that \eqref{negative tt*-Toda with x when n=2} is transformed into
\begin{align}
    \Tilde{w}_{ss} = \frac{\Tilde{w}_s^2}{\Tilde{w}} - \frac{\Tilde{w}_s}{s} - \frac{\Tilde{w}^2}{s} + \frac{1}{\Tilde{w}}, \label{Painleve III}
\end{align}
which is a special case of the Painlev\'e III ($D_7$) equation, and vice versa.

\begin{prop} \label{smoothness prop}
Every solution $w_0(x)$ of \eqref{negative tt*-Toda with x when n=2} is smooth on the positive real line.
\end{prop}
\begin{proof}
    From the Painlev\'e property, a solution of the equation \eqref{Painleve III} could have only a pole on $\R_{>0}$, i.e., $\Tilde{w}(s)$ is of the form,
    \begin{align*}
        \Tilde{w}(s) = \frac{c}{(s - s_0)^n} + \O\left( \frac{1}{(s - s_0)^{n-1}} \right)
    \end{align*}
    where $s_0 \in \R_{>0}$. This implies that $w_0 (x)$ might have a singularity only of the form,
    \begin{align}
        w_0(x) = a \ln (x - x_0) + b + \O \left( x - x_0 \right) \label{w0 possible form}
    \end{align}
    where $x_0 \in \R_{>0}$, but one can show that the ansatz \eqref{w0 possible form} is not compatible with equation \eqref{negative tt*-Toda with x when n=2} for real $a$, $x_0$, and $b$.
\end{proof}

%%%%%%%%%%%%%%%%%%%%%%%%%

\subsection{Lax Pair}

The radial Toda equation \eqref{negative tt*-Toda with x when n=2} admits a Lax pair representation:
\begin{align}
\begin{dcases}
    \Psi_{\zeta} = \left(- \frac{1}{\zeta^2} W -\frac{x}{\zeta} w_x - x^2 W^{T} \right) \Psi\\
    \Psi_{x} = \left( - w_x - 2 x \zeta W^T \right) \Psi,
\end{dcases} \label{Lax pair with x}
\end{align}
where $\Psi := \Psi (x, \zeta)$ is a $3 \times 3$ matrix with $x \in \R_{>0}$ and $\zeta \in \C^* = \C \setminus \{0\}$, and
\begin{align}
    w = \diag (w_0(x), 0, -w_0(x)), \quad W = \begin{pmatrix}
        0 & e^{-w_0(x)} & 0\\
        0 & 0 & e^{-w_0(x)}\\
        e^{2w_0(x)} & 0 & 0
    \end{pmatrix}. \label{w and W}
\end{align}
The compatibility condition of \eqref{Lax pair with x} gives  \eqref{negative tt*-Toda with x when n=2}.

If we take any solution $w_0(x)$ of the radial Toda equation \eqref{negative tt*-Toda with x when n=2} and the corresponding solution $\Psi(x, \zeta)$ of \eqref{Lax pair with x} 
and put
\begin{align}
    \tilde{\Psi}(t,\bar{t}, \lambda) := \Psi(|t|, \lambda/t), \quad t \in \C, \quad \lambda \in \C^*,
\end{align}
then $\tilde{\Psi}$ will satisfy Mikhailov's Lax pair (see \cite{Mik}) for the 2D periodic Toda equation \eqref{ost} of type $A_2$ with $\epsilon = 1$:
\begin{align}
\begin{dcases}
    \Psi_{t} = \left( w_t + \frac{1}{\lambda} W \right) \Psi\\
    \Psi_{\bar{t}}= \left( -w_{\bar{t}} - \lambda W^{T} \right) \Psi
\end{dcases} \label{Lax pair}    
\end{align}
where $w_0(x)$ in \eqref{w and W} is replaced by $w_0(t, \bar{t}) := w_0(|t|)$ to obtain $w$ and $W$ in \eqref{Lax pair}. Then, the compatibility condition of \eqref{Lax pair} is equivalent to the PDE
\begin{align}
    2 (w_0)_{t \bar{t}} = e^{-2w_0} - e^{4 w_0}.
\end{align}
We shall say more in section \ref{section 4} about the relationship between these two Lax pairs.

\subsection{Monodromy Data} \label{monodromy data}

Following the general theory of ODEs with singular points (see for instance  Chapter 1 of \cite{FIKN}), we will obtain the monodromy data for the first equation of system  \eqref{Lax pair with x}.

%%%%%%%%%%%%%%%%%%%%%%%%%%%%%%%%

\subsubsection{Formal Solutions} \label{Formal solution section}
We start with  the formal solutions of the equation 
\begin{align}
\Psi_{\zeta} = \left(- \frac{1}{\zeta^2} W -\frac{x}{\zeta} w_x - x^2 W^{T} \right) \Psi.\label{first eq}
\end{align}
This equation has two irregular singularities at  $\zeta = 0$ and at $\zeta = \infty$.

The leading term of \eqref{first eq} at $\zeta=0$ is $-W$, and it can be diagonalized as follows:
\begin{align}
-W = P_0 \, (-d_3) P_0^{-1}
\end{align}
where 
\begin{align} P_0 = \begin{pmatrix}
        e^{-w_0} & e^{-w_0} & e^{-w_0}\\
        1 & \omega & \omega^2\\
        e^{w_0} & \omega^2 e^{w_0} & \omega e^{w_0}
    \end{pmatrix}, \; d_3 = \diag (1, \omega, \omega^2), \quad \omega=e^{\frac{2\pi i}{3}}. \label{P_0 expression 1}
\end{align}
 Moreover, matrix $P_0$ has the following decomposition 
\begin{align}
    P_0 = \begin{pmatrix}
        e^{-w_0} & 0 & 0\\
        0 & 1 & 0\\
        0 & 0 & e^{w_0}
    \end{pmatrix}\begin{pmatrix}
        1 & 1 & 1\\
        1 & \omega & \omega^2\\
        1 & \omega^2 & \omega
    \end{pmatrix} = e^{-w}\Omega \label{P_0 expression 2},
\end{align}
where we defined
\begin{align*}
    e^{-w}=\begin{pmatrix}
        e^{-w_0} & 0 & 0\\
        0 & 1 & 0\\
        0 & 0 & e^{w_0}
    \end{pmatrix}, \quad \Omega = \begin{pmatrix}
        1 & 1 & 1\\
        1 & \omega & \omega^2\\
        1 & \omega^2 & \omega
    \end{pmatrix}.
\end{align*}
By Proposition 1.1 of \cite{FIKN}, one can obtain a unique formal solution of \eqref{first eq} at $\zeta = 0$ of the form
\begin{align}
    \Psi_{f}^{(0)}(\zeta) = P_0 (I + \O(\zeta)) e^{\frac{1}{\zeta} d_3}. \label{formal solution near 0}
\end{align}
%Note that the formal monodromy exponent is $0$ in this case.

Similarly, we can obtain the formal solution at $\zeta = \infty$. After the transformation $\zeta \mapsto \frac{1}{\zeta}$, equation \eqref{first eq} becomes
\begin{align*}
    \Psi_{\zeta} = \left(\frac{x^2}{\zeta^2} W^T +\frac{x}{\zeta} w_x + W \right) \Psi.
\end{align*}
The   leading term is $x^2 W^T$, and $W^T$ can be diagonalized as follows:  
\begin{equation}\begin{aligned}
    &W^{T}  = P_{\infty} d_3 P_{\infty}^{-1},\\
    &P_{\infty} =P_{0}^{T-1}= e^{w} \Omega^{-1}.\label{P_infty expression 1}
\end{aligned}\end{equation}
Again, by Proposition 1.1 of \cite{FIKN}, one can obtain a unique formal solution of \eqref{first eq} at $\zeta = \infty$ of the form
\begin{align}
    \Psi_{f}^{(\infty)}(\zeta) = P_{\infty} (I + \O(\zeta^{-1})) e^{-x^2 \zeta d_3}. \label{formal solution near infty}
\end{align}
 
%Note that the formal monodromy exponent is $0$ in this case also.

\begin{rem}
Observe that $\omega^2 + \omega + 1 = 0$, thus
\begin{align}
    \Omega^2= 3C, \label{omega^2=3C}
\end{align}
where
\begin{align*}
    C = \begin{pmatrix}
        1 & 0 & 0\\
        0 & 0 & 1\\
        0 & 1 & 0
    \end{pmatrix}.
\end{align*}
Since  $C^2 = I$,  we have
\begin{align}
    \Omega^{-1} = \frac{1}{3} C \Omega = \frac{1}{3} \Omega C \label{omega^{-1}}=\frac13 \bar\Omega.
\end{align}
\end{rem}

%%%%%%%%%%%%%%%%%%%%%%%%%%%%%%%%%%%%%%%%%%%%%%%%%%%%%%%%%%%%%%

\subsubsection{Stokes Sectors} \label{Stokes sectors}
Next, we describe the  {\it Stokes rays} and {\it Stokes sectors} of \eqref{first eq}.

We start with  $\zeta = 0$. 
%Applying general theory (see e.g.) to our equation \eqref{first eq}, we conclude the Stokes rays
Recall that Stokes sectors are defined so that the fundamental solution is uniquely determined there by the asymptotic solution $\Psi_f^{(0)}$. 
To begin with, suppose we have two fundamental solutions $\Psi^{(0)}$ and $\tilde{\Psi}^{(0)}$ with the same asymptotics at $\zeta=0$. Then their ratio $\Psi^{(0)}(\zeta)^{-1}\tilde{\Psi}^{(0)}(\zeta)$ is a constant matrix. To obtain the unique fundamental solution in some sector $\Omega_n^{(0)}$, we need to have 
\begin{align}
    \Psi^{(0)}(\zeta)^{-1}\tilde{\Psi}^{(0)}(\zeta)=I, \quad \zeta\in \Omega_n^{(0)}.
\end{align} 
Since $\Psi^{(0)}$ and $\tilde{\Psi}^{(0)}$ have the same asymptotics at $\zeta=0$, it holds that
\begin{align*}
    & \Psi^{(0)}(\zeta)^{-1}\tilde{\Psi}^{(0)}(\zeta)=  \lim_{\zeta\rightarrow 0}  e^{-\frac{1}{\zeta} d_3}(I + \O(\zeta)) e^{\frac{1}{\zeta} d_3}\\
    &= \lim_{\zeta\rightarrow 0}  \begin{pmatrix}
        1+* & * e^{-\frac{1}{\zeta} (1 - \omega)} & * e^{-\frac{1}{\zeta}(1 - \omega^2)}\\
        * e^{\frac{1}{\zeta} (1 - \omega)} & 1 + * & * e^{-\frac{1}{\zeta}(\omega - \omega^2)}\\
        * e^{\frac{1}{\zeta}(1 - \omega^2)} &  * e^{\frac{1}{\zeta}(\omega - \omega^2)} & 1 + *
    \end{pmatrix}.
\end{align*}
Thus, to get $\Psi^{(0)}(\zeta)^{-1}\tilde{\Psi}^{(0)}(\zeta)=I$, we need $\Omega_n^{(0)}$ to contain at least one Stokes ray which is defined by
\begin{align*}
    \Re{(\omega^{j-1}-\omega^{i-1})\zeta^{-1}}=0,
\end{align*}
for each pair $(i,j)$, $i<j$. More precisely, Stokes rays are
\begin{align*}
    l^{(1,2)}_n &= \left\{ \zeta \in \C \cup \{\infty\} \;\big |\; \arg \zeta = -\frac{2 \pi}{3}- n\pi \right\},\\
    l^{(1,3)}_n &= \left\{ \zeta \in \C \cup \{\infty\} \;\big |\; \arg \zeta = -\frac{\pi}{3} - n\pi \right\},\\
    l^{(2,3)}_n &= \left\{ \zeta \in \C \cup \{\infty\} \;\big |\; \arg \zeta = - n\pi \right\},
\end{align*}
for $n \in \Z$.

Now we define Stokes sectors by those which contain exactly one Stokes ray for each superscript $(i,j)$ with $i<j$, i.e., $\Omega_n^{(0)}$ contains $l^{(1,2)}_n$, $l^{(1, 3)}_n$, and $l^{(2,3)}_n$ for each $n \in \Z$. Thus, one can take the following Stokes sectors:
\begin{align*}
    \Omega^{(0)}_1 (\zeta) &= \left\{\zeta \in \C^* \: \big | \; -\frac{2 \pi}{3} < \arg \zeta < \frac{2 \pi}{3} \right\},\\
    \Omega^{(0)}_2 (\zeta) &= \left\{\zeta \in \C^* \: \big | \; -\frac{5 \pi}{3} < \arg \zeta < -\frac{\pi}{3} \right\},\\
    \Omega^{(0)}_3 (\zeta) &= \left\{\zeta \in \C^* \: \big | \; -\frac{8 \pi}{3} < \arg \zeta < -\frac{4\pi}{3} \right\}.
\end{align*}
\begin{figure}[htbp]
    \centering
    \includegraphics[width=3.5cm]{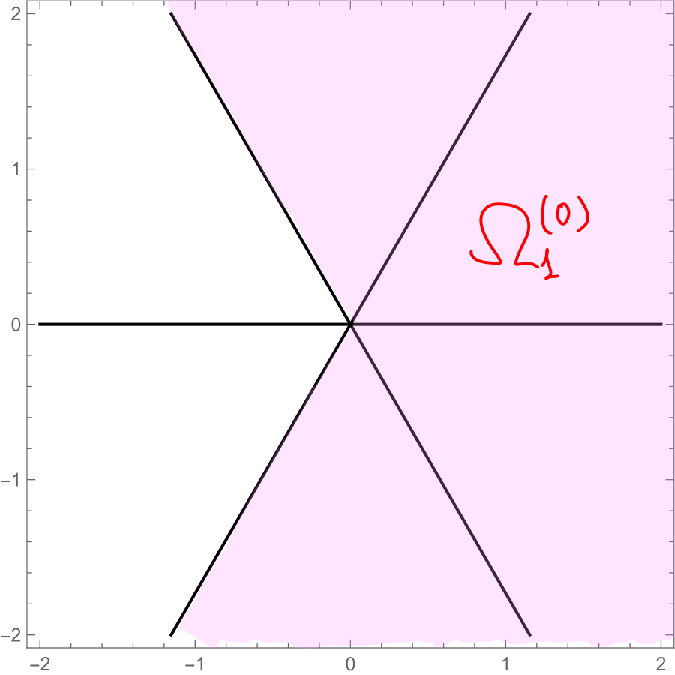}
    \includegraphics[width=3.5cm]{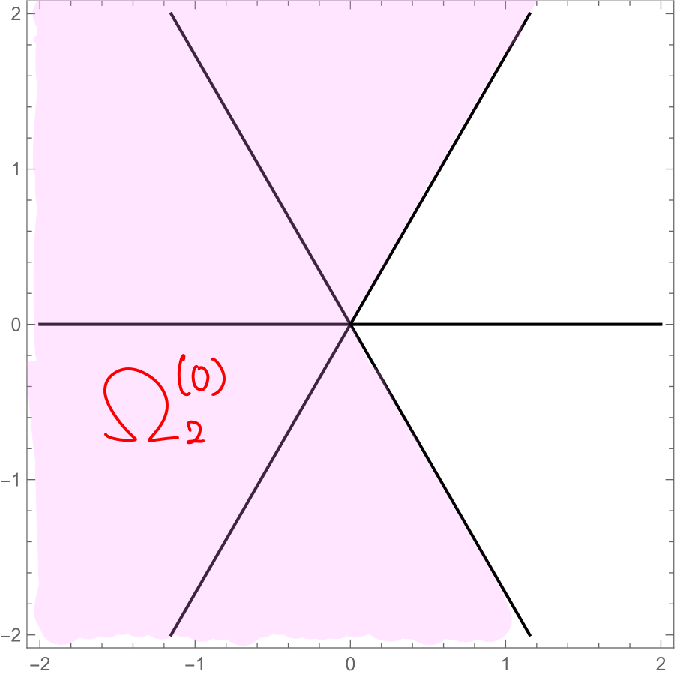}
    \caption{Stokes rays and Stokes sectors at $\zeta = 0$. (Here we only depicted $\Omega_1^{(0)}$ and $\Omega_2^{(0)}$, but the others are analogous.)}
    \label{Stokes rays and Stokes sectors at 0}
\end{figure}

Similarly, we can define the Stokes sectors at $\zeta = \infty$. First, we find the  Stokes rays at $\zeta = \infty$:
\begin{align*}
    \tilde{l}^{(1,2)}_n &= \left\{ \zeta \in \C \cup \{ \infty \} \;\big |\; \arg \zeta = \frac{2 \pi}{3} + n \pi \right\},\\
     \tilde{l}^{(1,3)}_n &= \left\{ \zeta \in \C \cup \{ \infty \} \;\big |\; \arg \zeta = \frac{\pi}{3} + n \pi \right\},\\
     \tilde{l}^{(2,3)}_n &= \left\{ \zeta \in \C \cup \{ \infty \} \;\big |\; \arg \zeta = n\pi \right\}.
\end{align*}
Then, the Stokes sectors at $\infty$ are 
\begin{align*}
    \Omega^{(\infty)}_1 (\zeta) &= \left\{\zeta \in \C^* \: \big | \; -\frac{2 \pi}{3} < \arg \zeta < \frac{2 \pi}{3} \right\},\\
    \Omega^{(\infty)}_2 (\zeta) &= \left\{\zeta \in \C^* \: \big | \; \frac{ \pi}{3} < \arg \zeta < \frac{5\pi}{3} \right\},\\
    \Omega^{(\infty)}_3 (\zeta) &= \left\{\zeta \in \C^* \: \big | \; \frac{4 \pi}{3} < \arg \zeta < \frac{8\pi}{3} \right\}.
\end{align*}
Moreover, we have
\begin{align}
\begin{aligned}
    \Omega^{(0)}_n(\zeta) &= \Omega^{(\infty)}_n(\zeta^{-1}) \; \text{ for $n \in \Z$.}
\end{aligned} \label{relation between omega^0 and omega^infty}
\end{align}
\begin{figure}[htbp]
    \centering
    \includegraphics[width=3.5cm]{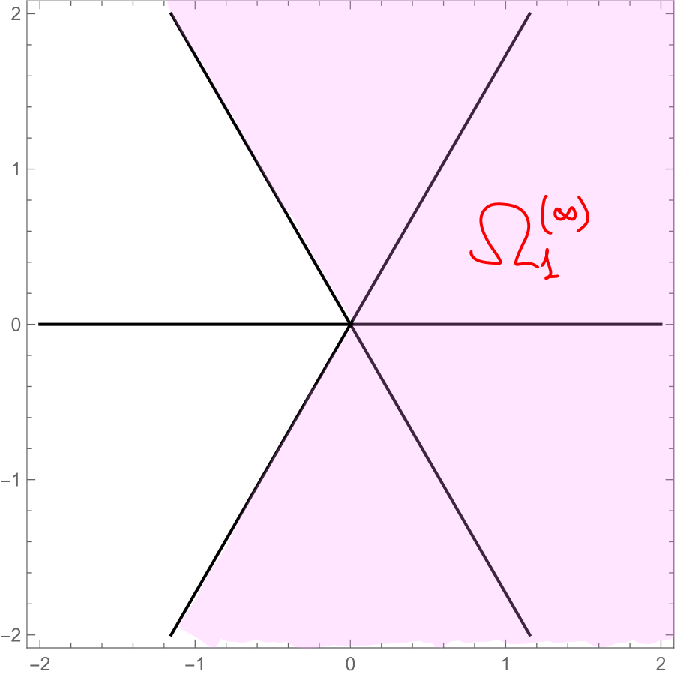}
    \includegraphics[width=3.5cm]{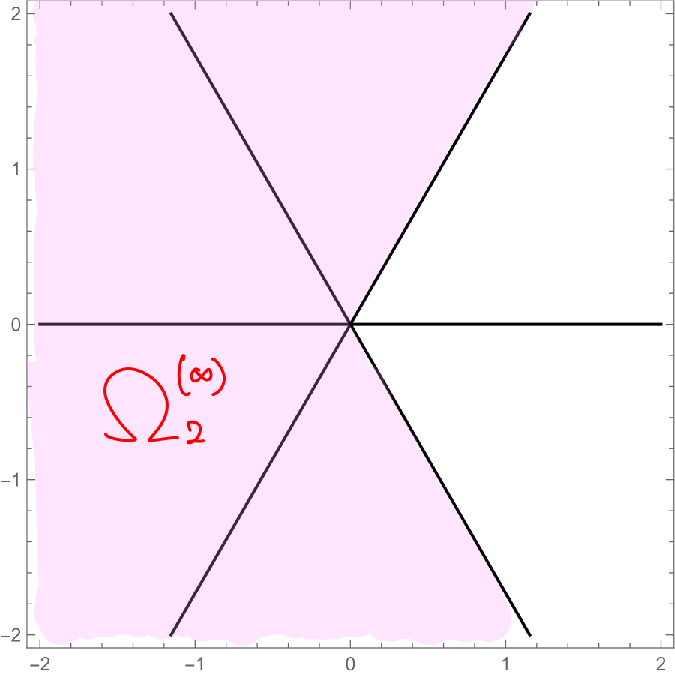}
    \caption{Stokes rays and Stokes sectors at $\zeta = \infty$.}
    \label{stokes rays and Stokes sectors at infty}
\end{figure}

Strictly speaking, sectors $\Omega_{n}^{(\infty, 0)}$ should be considered as lying  on the universal covering $\widetilde{{\mathbb C}^*}$.

%%%%%%%%%%%%%%%%%%%%%%%%%%%%%%%%%%%%%%%%%%%%%%%%%%%%%%

\subsubsection{Stokes Matrices}\label{subsection on stokes matrices}

Next, we define the {\it canonical solutions}. According to the general theory (see e.g. Theorem 1.4 in \cite{FIKN}), in each sector $\Omega_{n}^{(\infty, 0)}$ where $n \in \Z$, there exist unique solutions $\Psi^{(\infty, 0)}_n$
of \eqref{first eq} satisfying the asymptotic condition,
\begin{align*}
    &\Psi^{( 0)}_n(\zeta) \sim \Psi^{( 0)}_f(\zeta), \quad\zeta \to 0, \quad \zeta \in \Omega^{(0)}_n,\\
    &\Psi^{(\infty)}_n(\zeta) \sim \Psi^{(\infty)}_f(\zeta), \quad \zeta \to \infty, \quad \zeta \in \Omega^{(\infty)}_n.
\end{align*}
It also should be noticed that
\begin{equation}\label{Psi31}
\Psi_{3}^{(\infty)}(\zeta) = \Psi_{1}^{(\infty)}(\zeta e^{-2\pi i}) \quad \mbox{and} \quad
\Psi_{3}^{(0)}(\zeta) = \Psi_{1}^{(0)}(\zeta e^{2\pi i}).
\end{equation}

\begin{rem}
Another fundamental fact following from the general theory is that every solution of \eqref{first eq}, and, in particular, all our canonical solutions $\Psi^{(\infty, 0)}_n$ admit analytical continuation on the whole universal covering $\widetilde{{\mathbb C}^*}$ of ${\mathbb C}^* = \mathbb C \setminus \{0\}$.
Relations (\ref{Psi31}) should be then understood as equations on  $\widetilde{{\mathbb C}^*}$ where we accept a somewhat ``colloquial'' convention to denote the point on $\widetilde{{\mathbb C}^*}$ and its projection on ${\mathbb C}^*$ by the same letter $\zeta$. We also note that, in the case of ${\mathbb C}^*$, we can write that 
$$
\widetilde{{\mathbb C}^*} \equiv \left\{\tilde{\zeta} = \Bigl(\zeta, \, \arg(\zeta)\Bigr) \, \Big | \, \zeta \in {\mathbb C}^* \right\}
$$
so that the canonical ``altering sheets'' involution $\beta$ on $\widetilde{{\mathbb C}^*}$  can be indicated in short writing as
$$
\beta: \zeta  \mapsto \zeta e^{2\pi i}.
$$
With this convention, the rigorous version of relations (\ref{Psi31}) reads,
$$
\Psi_{3}^{(\infty)} = (\beta^{*})^{-1}[\Psi_1^{(\infty)}] \quad\mbox{and}\quad
\Psi_{3}^{(0)} = \beta^{*}[\Psi_1^{(0)}].
$$
In what follows we shall however use a ``colloquial'' form (\ref{Psi31}) of this relation, but when talking about $\Psi^{(\infty, 0)}_n (\zeta)$ will always recognize that $\arg \zeta$ matters.
\end{rem}

The Stokes matrices are defined by
\begin{align}
    &S_n^{( 0)} = [\Psi^{( 0)}_n(\zeta)]^{-1} \Psi^{( 0)}_{n+1}(\zeta), \label{stokes matrices of 0 and infty}\\
    &S_n^{(\infty)} = [\Psi^{(\infty)}_n(\zeta)]^{-1} \Psi^{(\infty}_{n+1}(\zeta),
\end{align}
where $n \in \Z$. By the analytic continuation property, it holds that
$S_n^{(0, \infty)} = S_{n-2}^{(0, \infty)}$. Thus, for our purposes to find monodromy data, it is enough to focus on choices $n = 1, 2$.

\begin{lem}\label{structure of S's}
Stokes matrices have the following structures,
\begin{align*}
    &S^{(\infty)}_1 = \begin{pmatrix}
        1 & * & 0\\
        0 & 1 & 0\\
        * & * & 1
    \end{pmatrix}, \;
    S^{(\infty)}_2 = \begin{pmatrix}
        1 & 0 & *\\
        * & 1 & *\\
        0 & 0 & 1
    \end{pmatrix}, \\
    &S^{(0)}_1 = \begin{pmatrix}
        1 & 0 & *\\
        * & 1 & *\\
        0 & 0 & 1
    \end{pmatrix}, \;
    S^{(0)}_2 = \begin{pmatrix}
        1 & * & 0\\
        0 & 1 & 0\\
        * & * & 1
    \end{pmatrix}.
\end{align*}
\end{lem}

\begin{proof}
We only give proof for $S_1^{(\infty)}$. The same approach can be applied to determine the structure of other Stokes matrices.  

By the definition of the Stokes matrices \eqref{stokes matrices of 0 and infty} we have  $\Psi^{(\infty)}_n S^{(\infty)}_n \Psi^{(\infty)-1}_{n+1} = I$. For $\zeta \in \Omega_n^{(\infty)} \cap \Omega_{n+1}^{(\infty)}$, it can be further written as
\begin{align*}
    &\lim_{\zeta \to \infty} \Psi^{(\infty)}_f S^{(\infty)}_n \Psi^{(\infty)-1}_{f} = I\\
    &\Longrightarrow \lim_{\zeta \to \infty} P_{\infty} (I + \O(\zeta^{-1})) e^{-x^2 \zeta d_3} S^{(\infty)}_n e^{x^2 \zeta d_3} (I + \O(\zeta^{-1}))^{-1} P_{\infty}^{-1} = I\\
    &\Longrightarrow \lim_{\zeta \to \infty} e^{-x^2 \zeta d_3} S^{(\infty)}_n e^{x^2 \zeta d_3} = I,
\end{align*}
or
\begin{align}
    \lim_{\zeta \to \infty} \begin{pmatrix}
        * & * e^{-x^2 \zeta (1 - \omega)} & * e^{-x^2 \zeta (1 - \omega^2)}\\
        * e^{x^2 \zeta (1 - \omega)} & * & * e^{-x^2 \zeta (\omega - \omega^2)}\\
        * e^{x^2 \zeta(1 - \omega^2)} &  * e^{x^2 \zeta (\omega - \omega^2)} & *
    \end{pmatrix} = I, \label{for defining S_n^infty}
\end{align}
for $\zeta \in \Omega_n^{(\infty)} \cap \Omega_{n+1}^{(\infty)}$.

When $n=1$, $\zeta$ belongs to $\Omega_1^{(\infty)} \cap \Omega_{2}^{(\infty)}$, which means
\begin{align*}
    \zeta \in \left\{\zeta \in \C^* \; \bigg | \; \frac{\pi}{3} < \arg \zeta < \frac{2 \pi}{3} \right \}.
\end{align*}
And it follows that
\begin{align*}
    &\frac{\pi}{6} < \arg \zeta(1 - \omega) < \frac{\pi}{2}\\
    &\frac{\pi}{2} < \arg \zeta(1 - \omega^2) < \frac{5\pi}{6}\\
    &\frac{5\pi}{6} < \arg \zeta(\omega - \omega^2) < \frac{7\pi}{6}.
\end{align*}
Thus equality  \eqref{for defining S_n^infty} holds if  $S^{(\infty)}_1$ have the following structure 
\begin{align*}
    S^{(\infty)}_1 = \begin{pmatrix}
        1 & * & 0\\
        0 & 1 & 0\\
        * & * & 1
    \end{pmatrix}.
\end{align*}
\end{proof}

Lemma \ref{structure of S's} describes the structure of $S_n^{(\infty, 0)}$. In the next subsections, we will discuss some symmetries of these matrices that will help us to parameterize each $S_n^{(\infty, 0)}$.  

%%%%%%%%%%%%%%%%%%%%%%%%%%%%%%%%%%%%%%%%%%%%%%%%%%%%%%%%%%

\subsubsection{Anti-symmetry and Inversion symmetry relations} \label{anti-symmetry}
In this section, we consider the anti-symmetry relations of the formal and canonical solutions and the Stokes matrices at $\zeta = \infty$, and the inversion symmetry relations of those at $\zeta = 0$.

%%%%%%%%%%%%%%%%%%%%%%%%%%%%%%%%%%%%%%%%%%%%%%%%%%%%%%%%%%

\paragraph{Anti-symmetry relation at $\zeta = \infty$}

Let $A(\zeta)$ be the coefficient matrix of \eqref{first eq}:
\begin{align*}
    A(\zeta) = - \frac{1}{\zeta^2} W -\frac{x}{\zeta} w_x - x^2 W^{T}.
\end{align*}
Then, one can observe the following symmetry:
\begin{align}
    \Delta A^T(-\zeta) \Delta = A(\zeta), \label{antisymmetry with A at infty}
\end{align}
where
\begin{align*}
    \Delta = \begin{pmatrix}
        0 & 0 & 1\\
        0 & 1 & 0\\
        1 & 0 & 0
    \end{pmatrix}.
\end{align*}
Moreover, \eqref{antisymmetry with A at infty} implies the following lemma: 
\begin{lem}\label{sym sol 1}
$\Psi(\zeta)$ is a solution of \eqref{first eq}, then so is $\tilde{\Psi}(\zeta) := \Delta \Psi^{T-1}(- \zeta)$.
\end{lem}
\begin{proof}
Since $\Psi(\zeta)$ is a solution of \eqref{first eq}, it satisfies
\begin{align*}
    \frac{d \Psi(\zeta)}{d \zeta} = A(\zeta) \Psi(\zeta).
\end{align*}
Taking a transposition and  changing  $\zeta \mapsto - \zeta$, we have
\begin{align}
    \frac{d \Psi^{T}(-\zeta)}{d \zeta} = \Psi^{T}(-\zeta) A^T(-\zeta). \label{for proving prop 3.1}
\end{align}
Thus,
\begin{align*}
    \frac{d \tilde{\Psi}(\zeta)}{d \zeta} &= \Delta \left( - \frac{d \Psi^{T-1}(- \zeta)}{d \zeta} \right)\\
    &= \Delta \Psi^{T-1} (- \zeta) \frac{d \Psi^{T} (- \zeta)}{d \zeta} \Psi^{T-1} (- \zeta)\\
    &= \Delta \Psi^{T-1} (- \zeta) \Psi^{T} (- \zeta)A^T(-\zeta) \Psi^{T-1} (- \zeta) \; \text{by \eqref{for proving prop 3.1}}\\
    &= \Delta A^T(-\zeta) \Psi^{T-1} (- \zeta)\\
    &= \Delta A^T(-\zeta) \Delta \Delta \Psi^{T-1} (- \zeta)\\
    &= A(\zeta) \tilde{\Psi}(\zeta) \; \text{by \eqref{antisymmetry with A at infty} and definition of $\tilde{\Psi}$.}
\end{align*}
\end{proof}
One can observe that 
\begin{align*}
    \Delta P_{\infty}^{T-1} \frac{d_3}{3} = P_{\infty}.
\end{align*}
Together with Lemma \ref{sym sol 1}  this leads to the anti-symmetry relation for the formal solution:
\begin{align}
    \Delta [\Psi^{(\infty)}_f(- \zeta)]^{T-1} \frac{d_3}{3} = \Psi^{(\infty)}_f(\zeta). \label{antisymmetry for formal solution at infty}
\end{align}
Indeed,
\begin{align*}
    \Delta [\Psi^{(\infty)}_f(- \zeta)]^{T-1} \frac{d_3}{3} &= \Delta [ P_{\infty} (I + \O(\zeta^{-1})) e^{x^2 \zeta d_3}]^{T-1} \frac{d_3}{3}\\
    &=\Delta P_{\infty}^{T-1} (I + \O(\zeta^{-1})) e^{-x^2 \zeta d_3} \frac{d_3}{3}\\
    &= \Delta P_{\infty}^{T-1} \frac{d_3}{3}(I + \O(\zeta^{-1}))e^{-x^2 \zeta d_3}\\
    &= P_{\infty}(I + \O(\zeta^{-1}))e^{-x^2 \zeta d_3}\\
    &= \Psi^{(\infty)}_f(\zeta).
\end{align*}

Moreover, considering the Stokes sectors of each canonical solution (see Figure~\ref{stokes rays and Stokes sectors at infty}), we have that \eqref{antisymmetry for formal solution at infty} implies the anti-symmetry of canonical solutions at $\zeta = \infty$:
\begin{align}
\begin{aligned}
    \Delta [\Psi^{(\infty)}_2(e^{i \pi} \zeta)]^{T-1} \frac{d_3}{3} &= \Psi^{(\infty)}_1(\zeta)\\
    \Delta [\Psi^{(\infty)}_3(e^{i \pi} \zeta)]^{T-1} \frac{d_3}{3} &= \Psi^{(\infty)}_2(\zeta)
\end{aligned} \label{antisymmetry for canonical solution at infty}
\end{align}
In terms of Stokes matrices, this can be interpreted as follows.
\begin{lem} One has 
\begin{align}
    S_2^{(\infty)} = d_3^{-1} [S_1^{(\infty)}]^{T-1} d_3. \label{antisymmetry for stokes matrices at infty}
\end{align}
\end{lem}
\begin{proof}
\begin{align*}
    &S_1^{(\infty)} = [\Psi^{(\infty)}_1]^{-1} [\Psi^{(\infty)}_2] = 3d_3^{-1} [\Psi^{(\infty)}_2]^T \Delta \Delta [\Psi^{(\infty)}_3]^{T-1} \frac{d_3}{3} \; \text{ by \eqref{antisymmetry for canonical solution at infty}}\\
    &\quad \; \; \; \; = d_3^{-1} (\Psi^{(\infty)-1}_3 \Psi^{(\infty)}_2)^{T} d_3 = d_3^{-1} [S_2^{(\infty)}]^{T-1} d_3.\\
    & \Leftrightarrow S_2^{(\infty)} = d_3^{-1} [S_1^{(\infty)}]^{T-1} d_3.
\end{align*}
\end{proof}
%%%%%%%%%%%%%%%%%%%%%%%%%%%%%%%

\paragraph{Inversion symmetry relation at $\zeta = 0$}

The matrix  $A(\zeta)$ admits another symmetry, namely 
\begin{align}
    \frac{1}{x^2 \zeta^2} \Delta A \left(-\frac{1}{x^2 \zeta} \right) \Delta = A(\zeta). \label{antisymmetry with A at 0}
\end{align}
For this, we have the following lemma.
\begin{lem}\label{sym of 2}
If  $\Psi(\zeta)$ is a solution of \eqref{first eq}, then so is 
 $\tilde{\Psi}(\zeta) := \Delta \Psi \left( -\frac{1}{x^2 \zeta} \right).$
 \end{lem}
On the level of Stokes matrices, it gives the following result.
\begin{lem} One has
\begin{align}
\begin{aligned}
     S^{(0)}_n &= d_3 S^{(\infty)}_{n+1} d_3^{-1}\\
     S_2^{(0)} &= d_3[S_1^{(0)}]^{T-1} d_3^{-1}.
\end{aligned} \label{antisymmetry for stokes matrices at 0}
\end{align}
\end{lem}
\begin{proof}
    
From Lemma \ref{sym of 2} one can obtain that the corresponding formulae for the formal solutions are
\begin{align}
    \Delta \Psi^{(\infty)}_f \left( -\frac{1}{x^2 \zeta} \right) 3 d_3^{-1} = \Psi^{(0)}_f (\zeta), \label{antisymmetry for formal solution at 0 ver.1}
\end{align}
and for the canonical solutions:
\begin{align}
     \Delta \Psi^{(\infty)}_{n+1} \left( e^{i \pi} \frac{1}{x^2 \zeta} \right) 3 d_3^{-1} = \Psi^{(0)}_n (\zeta), \label{antisymmetry for canonical solution at 0 ver.1}
\end{align}
where  $n = 1, 2$. In terms of $\Psi^{(0)}_k$ it can be read as
\begin{align}
    \Psi^{(0)}_2 (\zeta) = \Delta \left[\Psi^{(0)}_1 (e^{i \pi}\zeta) \right]^{T-1} 3d_3^{-1}. \label{antisymmetry for canonical solution at 0 ver.2}
\end{align}
which gives the corresponding formulae for the Stokes matrices.
\end{proof}
%%%%%%%%%%%%%%%%%%%%%%%%%%%%%%%

\subsubsection{Cyclic symmetry relations}

  The matrix $A(\zeta)$ has what is  called {\it cyclic symmetry}, namely 
\begin{align}
    \omega d_3^{-1} A(\omega \zeta) d_3 = A(\zeta). \label{cyclic symmetry with A}
\end{align}
Using this symmetry, we obtain the next lemma:
\begin{lem}\label{lem3.5}
If $\Psi(\zeta)$ is a solution of \eqref{first eq}, then so is $\tilde{\Psi}(\zeta) := d_3^{-1} \Psi(\omega \zeta)$.
\end{lem}

%%%%%%%%%%%%%%%%%%%%%%%%%%%%%%%%%%%%%%%%%%%%%%%%%%%%

\paragraph{Cyclic symmetry relations of the formal solutions}

Let us introduce
\begin{align*}
    \Pi = \begin{pmatrix}
        0 & 1 & 0\\
        0 & 0 & 1\\
        1 & 0 & 0
    \end{pmatrix}.
\end{align*}
Then Lemma \ref{lem3.5} together with equalities  
\begin{align*}
    &d_3^{-1} P_{\infty} \Pi = P_{\infty}, \quad \Pi^{-1} e^{x^2 \omega \zeta d_3} \Pi = e^{x^2 \zeta d_3},  
\end{align*}
gives the following formula for the formal solution:
\begin{align}
    d_3^{-1}\Psi^{(\infty)}_f (\omega \zeta) \Pi = \Psi^{(\infty)}_f (\zeta). \label{cyclic symmetry for formal solution at infty}
\end{align}
Similarly, for $\zeta \to 0$, we have
\begin{align}
    d_3^{-1}\Psi^{(0)}_f (\omega \zeta) \Pi^{-1} = \Psi^{(0)}_f (\zeta), \label{cyclic symmetry for formal solution at 0}
\end{align}
which follows from the equalities 
\begin{align*}
    d_3^{-1} P_{0} \Pi^{-1} = P_{0} , \; \Pi e^{\frac{1}{\omega \zeta} d_3} \Pi = e^{\frac{1}{\zeta} d_3}.
\end{align*}
 
%%%%%%%%%%%%%%%%%%%%%%%%%%%%%%%%%

\paragraph{Cyclic symmetry relations of the fundamental solutions}

In order to formulate the cyclic symmetry of the fundamental  solutions  we  need a collection of Stokes sectors compatible with this symmetry.  

First, we consider the case  $\zeta = \infty$. For $n \in \frac{1}{3} \Z$  we introduce additional sectors
\begin{align}
    \Omega_{1+n}^{(\infty)} &:= e^{i\pi n} \Omega_1^{(\infty)}. \label{rotated sectors at infty}
\end{align}
In these sectors we consider  fundamental solutions $\Psi^{(\infty)}_{1+n} (\zeta)$ of \eqref{first eq}   determined by 
\begin{align}
    \Psi^{(\infty)}_{1+n} (\zeta)\sim  \Psi^{(\infty)}_{f} (\zeta), \quad \zeta\rightarrow \infty, \quad \zeta\in \Omega_{1+n}^{(\infty)}.  \label{fundamental solutions on rotated sectors at infty}
\end{align}
Then one can show that the cyclic symmetry relations of the formal solutions \eqref{cyclic symmetry for formal solution at infty} translates to
\begin{align}
    d_3^{-1}\Psi^{(\infty)}_{n + \frac{2}{3}} (\omega \zeta) \Pi = \Psi^{(\infty)}_n (\zeta), \quad \zeta \in \Omega_{n}^{(\infty)},
    \label{cyclic symmetry for fundamental solutions at infty}
\end{align}
for $n \in \frac{1}{3}\Z$.

Similarly, near $\zeta = 0$,  for  $ n \in \frac{1}{3}\Z$ we    introduce additional sectors 
\begin{align}
    \Omega_{n}^{(0)}(\zeta) &:= \Omega_{n}^{(\infty)} (\zeta^{-1})  ,
    \label{rotated sectors at 0}
\end{align}
which is consistent with the relations \eqref{relation between omega^0 and omega^infty}  we had before.
In these sectors we consider  fundamental solutions $\Psi^{(0)}_{n} (\zeta)$ of \eqref{first eq}   determined by   
 \begin{align}
    \Psi^{(0}_{n} (\zeta)\sim  \Psi^{(0)}_{f} (\zeta), \quad \zeta\rightarrow 0, \quad \zeta\in \Omega_{n}^{(0)},  \label{fundamental solutions on rotated sectors at zero}
\end{align}
where $  n \in \frac{1}{3}\Z.$ Then the cyclic symmetry relations   \eqref{cyclic symmetry for formal solution at 0} translates to
\begin{align}
    d_3^{-1}\Psi^{(0)}_{n - \frac{2}{3}} (\omega \zeta) \Pi^{-1} &= \Psi^{(0)}_{n} (\zeta), \quad \zeta \in \Omega_{n}^{(0)},
\label{cyclic symmetry for fundamental solutions on rotated sectors at 0}
\end{align}
for $n \in \frac{1}{3} \Z$.

%%%%%%%%%%%%%%%%%%%%%%%%%%%%%%%%%%%%%%%%%%%%

\paragraph{Cyclic symmetry relations of Jump matrices} \label{5.3.3}

It is not straightforward to formulate the cyclic symmetries for the Stokes matrices. We will obtain them in several steps; first, we define jump matrices for $\pi/3$ rotated fundamental solutions.  

\begin{lem} \label{structure of Qs}
For $n \in \frac{1}{3} \Z$, consider constant matrices  
\begin{align}\label{def of Q}
   Q_n^{(\infty)} := \Psi_n^{(\infty)-1}\Psi_{n+\frac{1}{3}}^{(\infty)}.
\end{align}
Again, from the analytic continuation property, it holds that $Q_n^{(\infty)} = Q_{n-2}^{(\infty)}$ and it is enough for us to consider $n=1, 2$ only.
These jump matrices have the following structure,
\begin{align*}
    Q^{(\infty)}_1 = \begin{pmatrix}
        1 & * & 0\\
        0 & 1 & 0\\
        0 & 0 & 1
    \end{pmatrix}, \;
    Q^{(\infty)}_{1\frac{1}{3}} = \begin{pmatrix}
        1 & 0 & 0\\
        0 & 1 & 0\\
        0 & * & 1
    \end{pmatrix}, \;
    Q^{(\infty)}_{1\frac{2}{3}} = \begin{pmatrix}
        1 & 0 & 0\\
        0 & 1 & 0\\
        * & 0 & 1
    \end{pmatrix},\\
    Q^{(\infty)}_{2} = \begin{pmatrix}
        1 & 0 & 0\\
        * & 1 & 0\\
        0 & 0 & 1
    \end{pmatrix}, \;
    Q^{(\infty)}_{2\frac{1}{3}} = \begin{pmatrix}
        1 & 0 & 0\\
        0 & 1 & *\\
        0 & 0 & 1
    \end{pmatrix}, \;
    Q^{(\infty)}_{2\frac{2}{3}} = \begin{pmatrix}
        1 & 0 & *\\
        0 & 1 & 0\\
        0 & 0 & 1
    \end{pmatrix}.
\end{align*}
\end{lem}

\begin{proof}
We will prove the structure for  $Q_1^{(\infty)}$ only. The same approach can be applied to determine the structure of other jump matrices.

For $Q^{(\infty)}_1$, equation  \eqref{def of Q} can be written as 
 $\Psi^{(\infty)}_1 Q^{(\infty)}_1 \Psi^{(\infty)-1}_{1\frac{1}{3}} = I$.  
Using the definition of the canonical solutions,    it implies 
\begin{align*}
    &\lim_{\zeta \to \infty} \Psi^{(\infty)}_f Q^{(\infty)}_1 \Psi^{(\infty)-1}_{f} = I,\qquad \zeta \in \Omega_1^{(\infty)} \cap \Omega_{1\frac{1}{3}}^{(\infty)}\\
    &\Longrightarrow \lim_{\zeta \to \infty} P_{\infty} (I + \O(\zeta^{-1})) e^{-x^2 \zeta d_3} Q^{(\infty)}_1 e^{x^2 \zeta d_3} (I + \O(\zeta^{-1}))^{-1} P_{\infty}^{-1} = I\\
    &\Longrightarrow \lim_{\zeta \to \infty} e^{-x^2 \zeta d_3} Q^{(\infty)}_1 e^{x^2 \zeta d_3} = I,
\end{align*}
or
\begin{align}
    \lim_{\zeta \to \infty} \begin{pmatrix}
        * & * e^{-x^2 \zeta (1 - \omega)} & * e^{-x^2 \zeta (1 - \omega^2)}\\
        * e^{x^2 \zeta (1 - \omega)} & * & * e^{-x^2 \zeta (\omega - \omega^2)}\\
        * e^{x^2 \zeta(1 - \omega^2)} &  * e^{x^2 \zeta (\omega - \omega^2)} & *
    \end{pmatrix} = I, \label{for determining the type of Q_1^infty}
\end{align}
for $\zeta \in \Omega_1^{(\infty)} \cap \Omega_{1\frac{1}{3}}^{(\infty)}$. By the definition of $\Omega_1^{(\infty)}$ and  $\Omega_{1\frac13}^{(\infty)}$, we have that $$-\frac{\pi}{3} < \arg \zeta < \frac{2 \pi}{3},$$ and
\begin{align*}
    -\frac{\pi}{2} &< \arg \zeta(1 - \omega) < \frac{\pi}{2},\\
    -\frac{\pi}{6} &< \arg \zeta(1 - \omega^2) < \frac{5\pi}{6},\\
    \frac{\pi}{6} &< \arg \zeta(\omega - \omega^2) < \frac{7\pi}{6}.
\end{align*}
Thus equality \eqref{for determining the type of Q_1^infty} holds if   $Q^{(\infty)}_1$ has the following structure
\begin{align*}
    Q^{(\infty)}_1 = \begin{pmatrix}
        1 & * & 0\\
        0 & 1 & 0\\
        0 & 0 & 1
    \end{pmatrix}.
\end{align*}
\end{proof}

Similarly for $n \in \frac{1}{3} \Z$, we consider constant matrices 
\begin{align}Q^{(0)}_n := \Psi^{(0)-1}_{n} \Psi^{(0)}_{n + \frac{1}{3}}.
\end{align}
Using the same reasoning  as for \eqref{antisymmetry for stokes matrices at infty} and \eqref{antisymmetry for stokes matrices at 0}, one can obtain the following symmetry relations for $Q^{(\infty, 0)}_n$,
\begin{align}
\begin{aligned}
    Q_{n+1}^{(\infty)} &= d_3^{-1} [Q_{n}^{(\infty)}]^{T-1} d_3,\\
    Q_{n+1}^{(\infty)} &= d_3^{-1} Q_{n}^{(0)} d_3,\\
    Q_{n+1}^{(0)} &= d_3 [Q_{n}^{(0)}]^{T-1} d_3^{-1},
\end{aligned}\label{antisymmetry for Q ver.1}
\end{align}
where $n \in \frac{1}{3} \Z$.  Observe that from \eqref{antisymmetry for Q ver.1} it follows that $ Q_{n}^{(0)} $ has the same structure as $ Q_{n+1}^{(\infty)}$.

\begin{lem} For $n \in \frac{1}{3} \Z$, we have the following cyclic symmetry relations of jump matrices $Q^{(\infty,0)}_n$:
\begin{align}
    &Q_n^{(\infty)} = \Pi^{-1} Q_{n + \frac{2}{3}}^{(\infty)} \Pi, \label{cyclic symmetry for Q at infty}\\
    &Q_n^{(0)} = \Pi^{-1} Q_{n + \frac{2}{3}}^{(0)} \Pi \label{cyclic symmetry for Q at 0}.
\end{align}
\end{lem}

\begin{proof}
From \eqref{cyclic symmetry for fundamental solutions at infty}  we have
\begin{align*}
    Q_1^{(\infty)} &= \Psi^{(\infty)-1}_{1}(\zeta) \Psi^{(\infty)}_{1 \frac{1}{3}}(\zeta)\\
    &= (d_3^{-1}\Psi^{(\infty)}_{1 \frac{2}{3}} (\omega \zeta) \Pi)^{-1} d_3^{-1}\Psi^{(\infty)}_{2} (\omega \zeta) \Pi\\
    &= \Pi^{-1} \Psi^{(\infty)-1}_{1 \frac{2}{3}} (\omega \zeta)\Psi^{(\infty)}_{2} (\omega \zeta) \Pi\\
    &= \Pi^{-1} Q_{1\frac{2}{3}}^{(\infty)} \Pi
\end{align*}
when $\zeta \in \Omega_{1}^{(\infty)} \cap \Omega_{1 \frac{1}{3}}^{(\infty)}$. Other cases can be proved similarly.
\end{proof}
 
Finally, the Stokes matrices $S_1^{(\infty,0)}$ and $S_2^{(\infty,0)}$ can be expressed in terms of  $Q_n^{(\infty,0)}$.

\begin{prop}  For $n = 1, 2$, one has 
\begin{align}
     &S_n^{(\infty)} = Q_{n}^{(\infty)}Q_{n+\frac{1}{3}}^{(\infty)}Q_{n+\frac{2}{3}}^{(\infty)} ,\label{S's are expressed Q's}\\
     &S_n^{(0)} =  Q_n^{(0)} Q_{n+\frac{1}{3}}^{(0)}  Q_{n+\frac{2}{3}}^{(0)}. \label{S is expressed by Qs at 0}
\end{align}
\end{prop}

\begin{proof}
 For $\zeta \in \Omega_1^{(\infty)} \cap \Omega_2^{(\infty)}$, we have
\begin{align*}
    \Psi^{(\infty)}_{2} &= \Psi^{(\infty)}_{1 \frac{2}{3}} Q_{1\frac{2}{3}}^{(\infty)} = \Psi^{(\infty)}_{1 \frac{1}{3}}Q_{1\frac{1}{3}}^{(\infty)}Q_{1\frac{2}{3}}^{(\infty)} = \Psi^{(\infty)}_1 Q_{1}^{(\infty)}Q_{1\frac{1}{3}}^{(\infty)}Q_{1\frac{2}{3}}^{(\infty)}.
\end{align*}
By definition of the Stokes matrix $S_1^{(\infty)}$, we have that  $S_1^{(\infty)} = Q_{1}^{(\infty)}Q_{1\frac{1}{3}}^{(\infty)}Q_{1\frac{2}{3}}^{(\infty)}$. Other cases can be proved similarly.
\end{proof}

%%%%%%%%%%%%%%%%%%%%%%%%%%%%%%%%%%%%%%%%%%%%%%%%%%%%%%%%%%%%%%%%%%%%%

\subsubsection{Parametrization of Stokes Matrices}

In the past two subsections, we formulated many symmetry equations that will help us to parameterize the monodromy data. Let us list them here.
\begin{framed}
In summary,
\begin{itemize}
\item{Symmetry relations for formal solutions \eqref{antisymmetry for formal solution at infty}, \eqref{antisymmetry for formal solution at 0 ver.1}, \eqref{cyclic symmetry for formal solution at infty}, \eqref{cyclic symmetry for formal solution at 0}:}
\begin{align*}
    &\Delta [\Psi^{(\infty)}_f(- \zeta)]^{T-1} \frac{d_3}{3} = \Psi^{(\infty)}_f(\zeta)\\
    &\Delta \Psi^{(\infty)}_f \left( -\frac{1}{x^2 \zeta} \right) 3 d_3^{-1} = \Psi^{(0)}_f (\zeta)\\
    &\Delta [\Psi^{(0)}_f(-\zeta)]^{T-1} 3d_3^{-1} = \Psi^{(0)}_f(\zeta)\\
    &d_3^{-1}\Psi^{(\infty)}_f (\omega \zeta) \Pi = \Psi^{(\infty)}_f (\zeta)\\
    &d_3^{-1}\Psi^{(0)}_f (\omega \zeta) \Pi^{-1} = \Psi^{(0)}_f (\zeta)
\end{align*}
\item{Symmetry relations for canonical solutions \eqref{antisymmetry for canonical solution at infty}, \eqref{antisymmetry for canonical solution at 0 ver.1},\eqref{antisymmetry for canonical solution at 0 ver.2}, \eqref{cyclic symmetry for fundamental solutions at infty}, \eqref{cyclic symmetry for fundamental solutions on rotated sectors at 0}:}
\begin{align*}
    &\Delta [\Psi^{(\infty)}_{n+1}(e^{i \pi} \zeta)]^{T-1} \frac{d_3}{3} = \Psi^{(\infty)}_n(\zeta)\\
    &\Delta \Psi^{(\infty)}_{n+1} \left( e^{i \pi} \frac{1}{x^2 \zeta} \right) 3 d_3^{-1} = \Psi^{(0)}_n (\zeta)\\
    &\Delta [\Psi^{(0)}_{n-1}(e^{i \pi}\zeta)]^{T-1} 3d_3^{-1} = \Psi^{(0)}_n(\zeta)\\
    &d_3^{-1}\Psi^{(\infty)}_{n + \frac{2}{3}} (\omega \zeta) \Pi = \Psi^{(\infty)}_n (\zeta)\\
    &d_3^{-1}\Psi^{(0)}_{n-\frac{2}{3}} (\omega \zeta) \Pi^{-1} = \Psi^{(0)}_n (\zeta)
\end{align*}
\item{Symmetry relations for Jump matrices \eqref{antisymmetry for Q ver.1}, \eqref{cyclic symmetry for Q at infty}, \eqref{cyclic symmetry for Q at 0}:}
\begin{align*}
    &Q_{n+1}^{(\infty)} = d_3^{-1} [Q_{n}^{(\infty)}]^{T-1} d_3\\
    &Q_{n+1}^{(\infty)} = d_3^{-1} Q_{n}^{(0)} d_3\\
    &Q_{n+1}^{(0)} = d_3 [Q_{n}^{(0)}]^{T-1} d_3^{-1}\\
    &Q_n^{(\infty)} = \Pi^{-1} Q_{n + \frac{2}{3}}^{(\infty)} \Pi\\
    &Q_n^{(0)} = \Pi^{-1} Q_{n + \frac{2}{3}}^{(0)} \Pi
\end{align*}

\end{itemize}
\end{framed}

%%%%%%%%%%%%%%%%%%%%%%%%%%%%%%%%%%%%%%%%%%%%%%%%%%%%%%%%%%%%%%%%

\paragraph{Parametrization of $S^{(\infty)}_{n}$} \label{sec 3.6.1}

Let us parameterize $ Q^{(\infty)}_1$ by
\begin{align*}
    Q^{(\infty)}_1 = \begin{pmatrix}
        1 & a & 0\\
        0 & 1 & 0\\
        0 & 0 & 1
    \end{pmatrix}.
\end{align*}
Then, using  the anti-symmetry and inversion symmetry relations \eqref{antisymmetry for Q ver.1} and  the cyclic symmetry of Jump matrices \eqref{cyclic symmetry for Q at infty}, \eqref{cyclic symmetry for Q at 0}, we have  
\begin{align*}
&Q_{1\frac{1}{3}}^{(\infty)} = \Pi^{-1} Q_{2}^{(\infty)} \Pi
    =  \begin{pmatrix}
        1 & 0 & 0\\
        0 & 1 & 0\\
        0 & -a\omega^2 & 1
    \end{pmatrix}, \\
   & Q_{1\frac{2}{3}}^{(\infty)} = \Pi Q_{1}^{(\infty)} \Pi^{-1} =
     \begin{pmatrix}
        1 & 0 & 0\\
        0 & 1 & 0\\
        a & 0 & 1
    \end{pmatrix},\\
   & Q_{2\frac{1}{3}}^{(\infty)} = d_3^{-1} [Q^{(\infty)}_{1\frac{1}{3}}]^{T-1} d_3 =  \begin{pmatrix}
        1 & 0 & 0\\
        0 & 1 & a\\
        0 & 0 & 1
    \end{pmatrix} \\
   & Q_{2\frac{2}{3}}^{(\infty)} = d_3^{-1} [Q^{(\infty)}_{1\frac{2}{3}}]^{T-1} d_3 = \begin{pmatrix}
        1 & 0 & -a\omega^2\\
        0 & 1 & 0\\
        0 & 0 & 1
    \end{pmatrix},\\
     &Q^{(\infty)}_2 = d_3^{-1} [Q^{(\infty)}_1]^{T-1} d_3  = \begin{pmatrix}
        1 & 0 & 0\\
        -a\omega^2 & 1 & 0\\
        0 & 0 & 1
    \end{pmatrix}.
\end{align*}
Equation \eqref{S's are expressed Q's} implies that
\begin{align}
    S_1^{(\infty)} &= Q_{1}^{(\infty)}Q_{1\frac{1}{3}}^{(\infty)}Q_{1\frac{2}{3}}^{(\infty)} = \begin{pmatrix}
        1 & a & 0\\
        0 & 1 & 0\\
        a & -a\omega^2 & 1
    \end{pmatrix}. \label{parametrization of S_1 at infty}
\end{align}
and 
\begin{align}
    S_2^{(\infty)} &= Q_{2}^{(\infty)}Q_{2\frac{1}{3}}^{(\infty)}Q_{2\frac{2}{3}}^{(\infty)} = \begin{pmatrix}
        1 & 0 & -a\omega^2\\
        -a\omega^2 & 1 & a^2\omega + a\\
        0 & 0 & 1
    \end{pmatrix}. \label{parametrization of S_2 at infty}
\end{align}

%%%%%%%%%%%%%%%%%%%%%%%%%%%%%%%%%%%%%%%%%%%%%%%%%%%%%%%%%%%%%%%%

\paragraph{Parametrization of $S^{(0)}_{n}$}

Using \eqref{antisymmetry for Q ver.1} once more, we have
\begin{align*}
    &Q^{(0)}_1 = \begin{pmatrix}
        1 & 0 & 0\\
        -a & 1 & 0\\
        0 & 0 & 1
    \end{pmatrix}, \; Q_{1\frac{1}{3}}^{(0)} = \begin{pmatrix}
        1 & 0 & 0\\
        0 & 1 & a \omega^2\\
        0 & 0 & 1
    \end{pmatrix}, \; Q_{1\frac{2}{3}}^{(0)} = \begin{pmatrix}
        1 & 0 & -a\\
        0 & 1 & 0\\
        0 & 0 & 1
    \end{pmatrix},\\
    &Q^{(0)}_2 = \begin{pmatrix}
        1 & a\omega^2 & 0\\
        0 & 1 & 0\\
        0 & 0 & 1
    \end{pmatrix}, \; Q_{2\frac{1}{3}}^{(0)} = \begin{pmatrix}
        1 & 0 & 0\\
        0 & 1 & 0\\
        0 & -a & 1
    \end{pmatrix}, \; Q_{2\frac{2}{3}}^{(0)} = \begin{pmatrix}
        1 & 0 & 0\\
        0 & 1 & 0\\
        a\omega^2 & 0 & 1
    \end{pmatrix}.
\end{align*}
Equation \eqref{S is expressed by Qs at 0} gives the following parameterisation  of $S^{(0)}_n$: 
\begin{align}
    S_1^{(0)} = \begin{pmatrix}
        1 & 0 & -a\\
        -a & 1 & a^2 + a\omega^2\\
        0 & 0 & 1
    \end{pmatrix} \; \text{ and } \; 
    S_2^{(0)} = \begin{pmatrix}
        1 & a\omega^2 & 0\\
        0 & 1 & 0\\
        a\omega^2 & -a & 1
    \end{pmatrix}. \label{parametrization of S_1,2 at 0}
\end{align}
Therefore, each Stokes matrix is parameterized by a single parameter $a$ --- {\it Stokes data}.

%%%%%%%%%%%%%%%%%%%%%%%%%%%%%%%%%%%%%%%%%%%%

\subsubsection{Connection Matrices}

Next we define the connection matrices $E_k$ which connects solutions at $\zeta = \infty$ and solutions at $\zeta = 0$ by
\begin{align}
\begin{aligned}
    \Psi_1^{(\infty)}(\zeta) &= \Psi_1^{(0)}(\zeta) E_1,\\
    \Psi_2^{(\infty)}(\zeta) &= \Psi_2^{(0)}(\zeta e^{-2 \pi i}) E_2,
\end{aligned} \label{defn of connection matrices}
\end{align}
for all $\zeta$ in the universal covering $\widetilde{\C^*}$.   

%%%%%%%%%%%%%%%%%%%%%%%%%%%%%%%%%%%%%%%%%%%%%%%%

\paragraph{Anti-symmetry and inversion symmetry relations}

%$E_2$ is not necessarily introduced since it can be expressed in terms of $E_1$ and some Stokes matrices. Indeed, we have \eqref{E_2 is expressed by E_1 and so on}. However, the appearance of $E_2$ is useful for deriving the anti-symmetry relation \eqref{anti-symmetry relations of E_1} and inversion symmetry relation \eqref{inversion symmetry relations of E_1} for $E_1$.

\begin{prop}\label{E symmetry 1}
One has
\begin{align}
&E_2 = S_1^{(0)-1} E_1 S_2^{(\infty)-1} = S_2^{(0)} E_1 S_1^{(\infty)}, \label{E_2 is expressed by E_1 and so on}\\
&S_2^{(0)} E_1 S_1^{(\infty)} = \frac{1}{9} d_3 E_1^{T-1} d_3, \label{anti-symmetry relations of E_1}\\
&S_2^{(0)} E_1 S_1^{(\infty)} = \frac{1}{9} d_3 E_1^{-1} d_3. \label{inversion symmetry relations of E_1}
\end{align}
As a corollary of \eqref{anti-symmetry relations of E_1} and \eqref{inversion symmetry relations of E_1}, we have $E_1 = E_1^{T}$.
\end{prop}

\begin{proof}
For the first equality of \eqref{E_2 is expressed by E_1 and so on}: By definitions of the Stokes matrices and connection matrices, 
\begin{align}
    \Psi_2^{(\infty)}(\zeta) &= \Psi_1^{(\infty)}(\zeta) S_1^{(\infty)} = \Psi_1^{(0)}(\zeta) E_1 S_1^{(\infty)} = \Psi_2^{(0)}(\zeta) S_1^{(0)-1} E_1 S_1^{(\infty)} \label{for proving E's symmetry 1}
\end{align}
for all $\zeta \in \widetilde{\C^*}$. Similarly,  using the definition of the Stokes matrices, we have
\begin{align*}
    \Psi_2^{(\infty)}(\zeta) S_2^{(\infty)} = \Psi_3^{(\infty)}(\zeta) = \Psi_1^{(\infty)}(\zeta e^{-2 \pi i}),
\end{align*}
which implies
\begin{align*}
    \Psi_2^{(\infty)}(\zeta) S_2^{(\infty)} S_1^{(\infty)} &= \Psi_1^{(\infty)}(\zeta e^{-2 \pi i}) S_1^{(\infty)} = \Psi_2^{(\infty)}(\zeta e^{-2 \pi i}).
\end{align*}
By \eqref{defn of connection matrices}, it becomes
\begin{align}
    \Psi_2^{(\infty)}(\zeta e^{-2 \pi i}) = \Psi_2^{(\infty)}(\zeta) S_2^{(\infty)} S_1^{(\infty)} = \Psi_2^{(0)}(\zeta e^{-2 \pi i}) E_2 S_2^{(\infty)} S_1^{(\infty)} \label{for proving E's symmetry 2}
\end{align}
for all $\zeta \in \widetilde{\C^*}$. Comparing \eqref{for proving E's symmetry 1} with \eqref{for proving E's symmetry 2}, we have
\begin{align*}
    & E_2 S_2^{(\infty)} S_1^{(\infty)} = S_1^{(0)-1} E_1 S_1^{(\infty)} \Leftrightarrow E_2 = S_1^{(0)-1} E_1 S_2^{(\infty)-1}.
\end{align*}

For the second equality of \eqref{E_2 is expressed by E_1 and so on}:
By definitions of the Stokes matrices and connection matrices, 
\begin{align*}
    \Psi_2^{(\infty)}(\zeta) &= \Psi_1^{(\infty)}(\zeta) S_1^{(\infty)} = \Psi_1^{(0)}(\zeta) E_1 S_1^{(\infty)}\\
    &= \Psi_3^{(0)}(\zeta e^{-2 \pi i}) E_1 S_1^{(\infty)} = \Psi_2^{(0)}(\zeta e^{-2 \pi i}) S_2^{(0)} E_1 S_1^{(\infty)}.
\end{align*}
Thus, it follows that 
\begin{align}E_2 = S_2^{(0)} E_1 S_1^{(\infty)}. \label{def of E2}\end{align} 

For the anti-symmetry relation \eqref{anti-symmetry relations of E_1}:
From \eqref{antisymmetry for canonical solution at infty}, we have
\begin{align}
\begin{aligned}
    &\Delta [\Psi^{(\infty)}_2(e^{i \pi} \zeta)]^{T-1} \frac{d_3}{3} = \Psi^{(\infty)}_1(\zeta)\\
    &\Leftrightarrow [\Psi^{(\infty)}_2(e^{i \pi} \zeta)]^{T-1} = 3 \Delta \Psi_1^{(\infty)}(\zeta) d_3^{-1}\\
    &\Leftrightarrow \Psi^{(\infty)}_2(\zeta) = \frac{1}{3} \Delta [\Psi_1^{(\infty)}(e^{-i \pi} \zeta)]^{T-1} d_3\\
    &\qquad \quad \quad \; \; \; \; = \frac{1}{3} \Delta [\Psi_1^{(0)}(e^{-i \pi} \zeta) E_1]^{T-1} d_3\\
    &\qquad \quad \quad \; \; \; \; = \frac{1}{3} \Delta [\Psi_1^{(0)}(e^{-i \pi} \zeta)]^{T-1} d_3^{-1} d_3E_1^{T-1} d_3.
\end{aligned} \label{for proving E's symmetry 3}
\end{align}
Now, using \eqref{antisymmetry for canonical solution at 0 ver.2}:
\begin{align*}
    &\Psi^{(0)}_2 (\zeta) = \Delta [\Psi^{(0)}_1 (e^{i \pi}\zeta)]^{T-1} 3d_3^{-1}\\
    &\Leftrightarrow \frac{1}{3} \Psi^{(0)}_2 (e^{-2 \pi i} \zeta) = \Delta [\Psi^{(0)}_1 (e^{-i \pi}\zeta)]^{T-1} d_3^{-1},
\end{align*}
\eqref{for proving E's symmetry 3} can be proceeded as follows,
\begin{align*}
   \Psi^{(\infty)}_2(\zeta) = \frac{1}{9} \Psi^{(0)}_2 (e^{-2 \pi i} \zeta) d_3E_1^{T-1} d_3.
\end{align*}
Comparing it with the definition of the connection matrices, we get $E_2 = \frac{1}{9}d_3E_1^{T-1} d_3$.

For the inversion symmetry relation \eqref{inversion symmetry relations of E_1}:
From \eqref{antisymmetry for canonical solution at 0 ver.1} and definition of $E_1$, we have
\begin{align}
    &\Delta \Psi^{(\infty)}_2 \left( e^{i \pi} \frac{1}{x^2\zeta} \right) 3 d_3^{-1} = \Psi^{(0)}_1(\zeta) = \Psi^{(\infty)}_1(\zeta) E_1^{-1}. \label{inv of E_2 eq 1}
\end{align}
If we use \eqref{antisymmetry for canonical solution at 0 ver.1} once more with \eqref{Psi31}, we could obtain
\begin{align}
    &\Delta \Psi^{(0)}_2 \left( e^{-i \pi} \frac{1}{x^2\zeta} \right) \frac{d_3}{3} = \Psi^{(\infty)}_1(\zeta). \label{inv of E_2 eq 2}
\end{align}
From \eqref{inv of E_2 eq 1} and \eqref{inv of E_2 eq 2}, it follows that
\begin{align*}
    \Psi^{(\infty)}_2 \left( e^{i \pi} \frac{1}{x^2\zeta} \right) = \Psi^{(0)}_2 \left( e^{-i \pi} \frac{1}{x^2\zeta} \right) \frac{1}{9} d_3 E_1^{-1} d_3.
\end{align*}
Thus, we have $E_2 = \frac{1}{9}d_3E_1^{-1} d_3$.
\end{proof}

\begin{rem} Using the anti-symmetry relation of $E_1$,
\begin{align}
    S_2^{(0)} E_1 S_1^{(\infty)} = \frac{1}{9} d_3 E_1^{T-1} d_3, \label{antisymmetry for E}
\end{align}
we can calculate the determinant of $E_1$. By \eqref{antisymmetry for stokes matrices at 0}, we have
\begin{align}
    S_1^{(\infty)} d_3^{-1} E_1 S_1^{(\infty)} = \frac{1}{9} E_1^{T-1} d_3. \label{antisymmetry for E updated}
\end{align}
Taking a determinant of previous equality, we get
\begin{align*}
    &\det S_1^{(\infty)} \det d_3^{-1} \det E_1 \det S_1^{(\infty)} = \left (\frac{1}{9}\right)^3 \det E_1^{T-1} \det d_3\\
    &\Leftrightarrow (\det E_1)^2 = \left ( \frac{1}{3} \right )^6
    \Leftrightarrow \det E_1 = \pm \left ( \frac{1}{3} \right )^3,
\end{align*}
where we used $\det d_3^{\pm 1} = \det S_n^{(\infty)} = 1$. 
We can also determine the sign of $\det E_1$. Taking a determinant of \eqref{defn of connection matrices}, we have 
\begin{align*}
    \det E_1 = \frac{\det \Psi_1^{(\infty)}}{\det \Psi_1^{(0)}}.
\end{align*}
Using the Jacobi's formula 
\begin{align}
    \dfrac{d}{d\zeta} \det \Psi(\zeta)= \det \Psi(\zeta)\cdot \left(\operatorname{tr}  \Psi^{-1}(\zeta) \; \dfrac{d}{d \zeta} \Psi(\zeta)\right)
\end{align}
and the fact that $\operatorname{tr}  \Psi^{-1}(\zeta) \; \dfrac{d}{d \zeta} \Psi(\zeta) = \operatorname{tr} A(\zeta) = 0$ we have that  $\det \Psi_n^{(\infty, 0)}$ is a constant. Thus, we can calculate   $\det \Psi_n^{(\infty, 0)}$ using its asymptotic:  $\det \Psi_f^{(\infty, 0)} = \det P_{\infty, 0}$.  Since 
\begin{align*}
   \det P_0 = (\det P_{\infty})^{-1} = \det e^{-w}\cdot \det \Omega= \det \Omega = -i 3\sqrt{3}, 
\end{align*}
we obtain
\begin{align}
    \det E_1 &= \frac{1}{(\det \Omega)^2} = -\frac{1}{27}.
\end{align}
\end{rem}

%%%%%%%%%%%%%%%%%%%%%%%%%%%%%%%%%%%%%%%%%

\paragraph{Cyclic symmetry relation of $E_1$}

By definition of the jump matrices $Q_n^{(\infty, 0)}$ for $n \in \frac{1}{3}\Z$ introduced in subsubsection \ref{5.3.3}, we have
\begin{align*}
    \Psi_{1\frac{2}{3}}^{(\infty, 0)}(\zeta) = \Psi_{1}^{(\infty, 0)}(\zeta) Q_1^{(\infty, 0)} Q_{1\frac{1}{3}}^{(\infty, 0)}.
\end{align*}
Using the cyclic symmetry relation, \eqref{cyclic symmetry for Q at 0} and \eqref{cyclic symmetry for fundamental solutions on rotated sectors at 0}, it follows that
\begin{align}
    \Psi_1^{(\infty)} = d_3^{-1}\Psi_{1\frac{2}{3}}^{(\infty)}(\omega \zeta) \Pi = d_3^{-1}\Psi_{1}^{(\infty)}(\omega \zeta) Q_1^{(\infty)} Q_{1\frac{1}{3}}^{(\infty)} \Pi \label{for cyclic symmetry for E_1 item 1}
\end{align}
and
\begin{align}
    \Psi_1^{(0)} &= \Psi_{1\frac{2}{3}}^{(0)}(\zeta) \left(Q_{1}^{(0)} Q_{1\frac{1}{3}}^{(0)} \right)^{-1} = d_3^{-1} \Psi_{1}^{(0)}(\omega \zeta) \Pi^{-1} \left( Q_{1}^{(0)} Q_{1\frac{1}{3}}^{(0)} \right)^{-1}. \label{for cyclic symmetry for E_1 item 2}
\end{align}
Substituting \eqref{for cyclic symmetry for E_1 item 1} and \eqref{for cyclic symmetry for E_1 item 2} into \eqref{defn of connection matrices}, we have
\begin{align*}
    & d_3^{-1}\Psi_{1}^{(\infty)}(\omega \zeta) Q_1^{(\infty)} Q_{1\frac{1}{3}}^{(\infty)} \Pi = d_3^{-1} \Psi_{1}^{(0)}(\omega \zeta) \Pi^{-1} \left( Q_{1}^{(0)} Q_{1\frac{1}{3}}^{(0)} \right)^{-1} E_1\\
    & \Longrightarrow d_3^{-1}\Psi_{1}^{(0)}(\omega \zeta) E_1 Q_1^{(\infty)} Q_{1\frac{1}{3}}^{(\infty)} \Pi = d_3^{-1} \Psi_{1}^{(0)}(\omega \zeta) \Pi^{-1} \left( Q_{1}^{(0)} Q_{1\frac{1}{3}}^{(0)} \right)^{-1} E_1\\
    & \Longrightarrow E_1 Q_1^{(\infty)} Q_{1\frac{1}{3}}^{(\infty)} \Pi = \Pi^{-1} \left( Q_{1}^{(0)} Q_{1\frac{1}{3}}^{(0)} \right)^{-1} E_1.
\end{align*}
Thus, we obtain that
\begin{align}
    E_1 = \left( Q_1^{(0)} Q_{1\frac{1}{3}}^{(0)} \Pi \right) E_1 \left(Q_1^{(\infty)} Q_{1\frac{1}{3}}^{(\infty)} \Pi \right). \label{cyclic symmetry for E_1}
\end{align}
Moreover,  using the symmetry relation \eqref{antisymmetry for Q ver.1}, equation \eqref{cyclic symmetry for E_1} becomes
\begin{align*}
    &E_1 = \left( d_3 Q_2^{(\infty)} d_3^{-1} d_3 Q_{2\frac{1}{3}}^{(\infty)} d_3^{-1} \Pi \right) E_1 \left(Q_1^{(\infty)} Q_{1\frac{1}{3}}^{(\infty)} \Pi \right)\\
    & \Longrightarrow d_3^{-1} E_1 = \left( Q_2^{(\infty)}  Q_{2\frac{1}{3}}^{(\infty)} d_3^{-1} \Pi d_3 \right) d_3^{-1} E_1 \left(Q_1^{(\infty)} Q_{1\frac{1}{3}}^{(\infty)} \Pi \right).
\end{align*}
One can verify that $d_3^{-1} \Pi d_3 = \omega \Pi$, thus  
\begin{align}
    d_3^{-1} E_1 = \omega \left( Q_2^{(\infty)}  Q_{2\frac{1}{3}}^{(\infty)} \Pi \right) d_3^{-1} E_1 \left(Q_1^{(\infty)} Q_{1\frac{1}{3}}^{(\infty)} \Pi \right). \label{cyclic symmetry for E_1 updated}
\end{align}
This is what we call the \textit{cyclic symmetry relation of $E_1$}.

\begin{framed}
In summary,  
\begin{itemize}
    \item Anti-symmetry and inversion relation of $E_1$ \eqref{antisymmetry for E}, \eqref{inversion symmetry relations of E_1}:
    \begin{align*}
        S_2^{(0)} E_1 S_1^{(\infty)} = \frac{1}{9} d_3 E_1^{T-1} d_3 = \frac{1}{9} d_3 E_1^{-1} d_3
    \end{align*}
    \item Cyclic symmetry relation of $E_1$ \eqref{cyclic symmetry for E_1 updated}:
    \begin{align*}
        d_3^{-1} E_1 = \omega \left( Q_2^{(\infty)}  Q_{2\frac{1}{3}}^{(\infty)} \Pi \right) d_3^{-1} E_1 \left(Q_1^{(\infty)} Q_{1\frac{1}{3}}^{(\infty)} \Pi \right)
    \end{align*}
\end{itemize}
\end{framed}

%%%%%%%%%%%%%%%%%%%%%%%%%%%%%%%%%%%%%%%%%%%%%%%%%%%

\subsection{Reality Condition}

Recall that we considered only real solutions $w_0(x) \in \R$  of the radial  Toda equation  \eqref{negative tt*-Toda with x when n=2}. This implies that the coefficient matrix $A(\zeta)$ of \eqref{first eq}:
\begin{align}
    &\Psi_{\zeta} = A(\zeta) \Psi, \label{one of the Lax pair again}\\
    &A(\zeta) = - \frac{1}{\zeta^2} W -\frac{x}{\zeta} w_x - x^2 W^{T},
\end{align}
satisfy  $\overline{A(\overline{\zeta})} = A(\zeta)$, which gives us the following Lemma.

\begin{lem}\label{prop3.4}
If $\Psi(\zeta)$ is a solution of \eqref{first eq}, then so is $\overline{\Psi(\bar\zeta)} $.
\end{lem}

\subsubsection{Reality Condition of the Formal Solutions}

Recall that the formal solution at $\zeta = \infty$ is given by \eqref{formal solution near infty}, so  we have
\begin{align}
\begin{aligned}
    \overline{\Psi_f^{(\infty)}(\overline{\zeta})} &= \overline{P_{\infty}}(I + \O(\zeta^{-1}))e^{-x^2 \zeta \overline{d_3}}\\
    &= \overline{P_{\infty}}C(I + \O(\zeta^{-1}))Ce^{-x^2 \zeta \overline{d_3}},\\
    &=\overline { e^{w} \Omega^{-1}C} (I + \O(\zeta^{-1})) e^{-x^2 \zeta {d_3}}C,
\end{aligned}\label{reality of formal solution at infty}
\end{align}
where we used the property  $C^2 = I$ , the definition $P_{\infty} = e^{w} \Omega^{-1}$, and equality 
\begin{align}
    C\overline{d_3} = d_3 C. \label{reality of d_3}
\end{align}
Observe that by \eqref{omega^{-1}}, we have 
\begin{align*}
    \overline{\Omega^{-1} C}=\Omega^{-1}.
\end{align*}
Thus,  we get
\begin{align*}
    \overline{\Psi_f^{(\infty)}(\overline{\zeta})} &= e^{w} \Omega^{-1} (I + \O(\zeta^{-1}))Ce^{-x^2 \zeta \overline{d_3}}\\
    &= P_{\infty}(I + \O(\zeta^{-1}))CCe^{-x^2 \zeta d_3}C\\
    &= P_{\infty}(I + \O(\zeta^{-1}))e^{-x^2 \zeta d_3}C.
\end{align*}
Therefore, we obtain the reality condition of the formal solutions:
\begin{align}
    \overline{\Psi_f^{(\infty)}(\overline{\zeta})} C = \Psi_f^{(\infty)}(\zeta). \label{reality of formal solution at infty updated}
\end{align}
Similarly, one can show that the formal solution at $\zeta = 0$ has the same property:
\begin{align}
    \overline{\Psi_f^{(0)}(\overline{\zeta})} C = \Psi_f^{(0)}(\zeta).\label{reality of formal solution at 0}
\end{align}

%%%%%%%%%%%%%%%%%%%%%%%%%%%%%%%%%%%%%%%%%%%%%%%%%

\subsubsection{Reality Condition of the Stokes Matrices}

Next, we will describe the reality condition on the level of the canonical solutions and the Stokes matrices. Recall that the Stokes sectors for the canonical solution at $\zeta = \infty$ are given by:
\begin{align*}
    \Omega^{(\infty)}_1 &= \left\{\zeta \in \C^* \: \big | \; -\frac{2 \pi}{3} < \arg \zeta < \frac{2 \pi}{3} \right\},\\
    \Omega^{(\infty)}_2 &= \left\{\zeta \in \C^* \: \big | \; \frac{ \pi}{3} < \arg \zeta < \frac{5\pi}{3} \right\}.
\end{align*}
By taking complex conjugate of $\zeta$, i.e. $\arg \overline{\zeta} = 2\pi - \arg \zeta$, the Stokes sectors for the solutions $\overline{\Psi_n^{(\infty)}(\overline{\zeta})}$ for $n=1,2$ are
\begin{align*}
    \overline{\Omega^{(\infty)}_1} &:= \left\{\zeta \in \C^* \: \big | \; \frac{4 \pi}{3} < \arg \zeta < \frac{8 \pi}{3} \right\},\\
    \overline{\Omega^{(\infty)}_2} &:= \left\{\zeta \in \C^* \: \big | \; \frac{\pi}{3} < \arg \zeta < \frac{5\pi}{3} \right\}.
\end{align*}
Thus, we have 
\begin{align}
\begin{aligned}
    \overline{\Omega^{(\infty)}_1} &= \Omega_3^{(\infty)},\\
    \overline{\Omega^{(\infty)}_2} &= \Omega_2^{(\infty)},
\end{aligned} \label{reality of stokes sector at infty}
\end{align}
  and the relations for the    canonical solutions are
\begin{align}
\begin{aligned}
    \overline{\Psi_1^{(\infty)}(\overline{\zeta})}C &= \Psi_3^{(\infty)}(\zeta), \\ \overline{\Psi_2^{(\infty)}(\overline{\zeta})}C &= \Psi_2^{(\infty)}(\zeta).
\end{aligned} \label{reality of canonical solutions at infty}
\end{align}
Finally, the reality conditions for the Stokes matrices are
\begin{align}
\begin{aligned}
    \overline{S_1^{(\infty)}} &= [\overline{\Psi_1^{(\infty)}(\overline{\zeta})}]^{-1} \overline{\Psi_2^{(\infty)}(\overline{\zeta})} \; \text{ by \eqref{stokes matrices of 0 and infty}}\\
    &= C [\Psi_3^{(\infty)}(\zeta)]^{-1} \Psi_2^{(\infty)}(\zeta) C \; \text{ by \eqref{reality of canonical solutions at infty}}\\
    &= C S_2^{(\infty)-1}C
\end{aligned} \label{reality of stokes matrices at infty}
\end{align}
Recall that we have parametrization of $S_1^{(\infty)}$ given by \eqref{parametrization of S_1 at infty} and  of $S_2^{(\infty)}$, given by \eqref{parametrization of S_2 at infty}. Using these parametrizations, \eqref{reality of stokes matrices at infty} becomes 
 
\begin{align*}
    \begin{pmatrix}
        1 & \overline{a} & 0\\
        0 & 1 & 0\\
        \overline{a} & - \overline{a} \omega & 1
    \end{pmatrix} = \begin{pmatrix}
        1 & a \omega^2 & 0\\
        0 & 1 & 0\\
        a \omega^2 & -a & 1
    \end{pmatrix} \iff \overline{a} = a \omega^2
\end{align*}
Thus, introducing  a real parameter  $s^{\R}$, one can write $a$ as
\begin{align}
     a = \omega^{2} s^{\R}. \label{a = omega^2 s}
\end{align}

%%%%%%%%%%%%%%%%%%%%%%%%%%%%%%%%%%%%%%%%%%%%%%%%%%%%%%

\subsubsection{Reality Condition of the Connection Matrices}

In this section, we will describe the reality condition for the connection matrix $E_1$. Since it does not contain $\zeta$, we have freedom in the definition of the complex conjugate of $\zeta$. Let $\arg \overline{\zeta} = - \arg \zeta$ this time. Then, the Stokes sectors have the property,
\begin{align*}
    \overline{\Omega^{(\infty)}_1} &= \Omega_1^{(\infty)} \; \text{ and } \;\overline{\Omega^{(0)}_1} = \Omega_1^{(0)},
\end{align*}
which implies the following  reality conditions of the canonical solutions,
\begin{align}
\begin{aligned}
    \overline{\Psi_1^{(\infty)}(\overline{\zeta})}C &= \Psi_1^{(\infty)}(\zeta), \\ 
    \overline{\Psi_1^{(0)}(\overline{\zeta})}C &= \Psi_1^{(0)}(\zeta).
\end{aligned} \label{reality of canonical solutions updated}
\end{align}
By \eqref{defn of connection matrices} and \eqref{reality of canonical solutions updated}, it follows that
\begin{align}
    E_1 = \left[\Psi_1^{(0)} \right]^{-1}\Psi_1^{(\infty)} = C \left[ \overline{\Psi_1^{(0)}} \right]^{-1} \overline{\Psi_1^{(\infty)}} C = C \overline{E_1} C. \label{reality of E_1}
\end{align}

%%After that, take Maxim's computation.

\begin{framed}
In summary, 
\begin{itemize}
    \item reality condition of the formal solutions \eqref{reality of formal solution at infty updated}, \eqref{reality of formal solution at 0}:
    \begin{align*}
        \overline{\Psi_f^{(\infty)}(\overline{\zeta})} C &= \Psi_f^{(\infty)}(\zeta).\\
        \overline{\Psi_f^{(0)}(\overline{\zeta})} C &= \Psi_f^{(0)}(\zeta).
    \end{align*}
    \item reality condition of the canonical solutions \eqref{reality of canonical solutions at infty}:
    \begin{align*}
        \overline{\Psi_1^{(\infty)}(\overline{\zeta})}C &= \Psi_3^{(\infty)}(\zeta) \\ \overline{\Psi_2^{(\infty)}(\overline{\zeta})}C &= \Psi_2^{(\infty)}(\zeta).
    \end{align*}
    \item reality condition of the Stokes matrices \eqref{reality of stokes matrices at infty}:
    \begin{align*}
        \overline{S_1^{(\infty)}} &= C S_2^{(\infty)-1}C
    \end{align*}
    \item reality condition of the connection matrices \eqref{reality of E_1}:
    \begin{align*}
        E_1 = C \overline{E_1} C.
    \end{align*}
\end{itemize}
\end{framed}

%%%%%%%%%%%%%%%%%%%%%%%%%%%%%%%%%%%%%%%%%%%%%%%

\subsubsection{More about the Monodromy Data} \label{section titled more about the monodromy data}

In this section, we will use symmetry relations to parameterize the connection matrix $E_1$ and will show that the monodromy data consists only of two real parameters.

\begin{prop} \label{E parametrization} The connection matrix $E_1$ can be parameterized by two real numbers $A^\R$ and $s^\R$, and one complex number $B$:
\begin{align}
    E_1 = \begin{pmatrix}
        A^{\R} & B & \overline{B}\\
        B & \omega s^{\R} A^{\R} - \omega^2 s^{\R} B + \overline{B} & A^{\R}\\
        \overline{B} & A^{\R} & \omega^2 s^{\R} A^{\R} + B - \omega s^{\R} \overline{B}
    \end{pmatrix} \label{parametrization of E_1 ver.2}
\end{align}
where the parameter $s^\R$ was introduced in \eqref{a = omega^2 s} and parameters $A^\R,s^\R,B$ satisfy 
\begin{align}
    &(A^{\R})^2 - \frac{1}{3} A^{\R} = |B|^2, \label{identity 7 at parametrization of E_1}\\
    &(1 + s^{\R})A^{\R} + \omega B + \omega^2 \overline{B} = \frac{1}{3}. \label{identity 6 at parametrization of E_1}
\end{align}
\end{prop}

The proof of this proposition is given in Appendix \ref{proof of E}. 
%As we will see in section \ref{section 4}, the behavior at $0$, $w_0 (x) = \gamma \ln x + \rho + o(1)$, suggests $A^{\R} > 0$. This, together with \eqref{identity 7 at parametrization of E_1}, implies $A^{\R} \geq 1/3$. 

%As we shall see later, subsection \ref{connection formulae sec} infers $A^{\R} > 0$ from the construction of a solution to \eqref{negative tt*-Toda with x when n=2} with asymptotics $w_0 (x) = \gamma \ln x + \rho + o(1)$, $x \to 0$ through holomorphic data (see remark after Lemma \ref{E_1 = D_1 D_1^* C}). Together with \eqref{identity 7 at parametrization of E_1} it gives $A^\R \geq \frac13$. 

\begin{rem}
In the following paper \cite{gikmo2} in this series we shall show that the case $A^\R < 0$ corresponds to having a complex solution to \eqref{negative tt*-Toda with x when n=2} of the form $w_0 = v_0 + \frac{i \pi}{2}$ where $v_0$ is real. Hence we must have
$A^\R > 0$ for any real solution of our equation \eqref{negative tt*-Toda with x when n=2}. Together with \eqref{identity 7 at parametrization of E_1} it gives $A^\R \geq \frac13$.
%Also, there is a fact that the choice $A^{\R} < 0$ corresponds with a complex solution to \eqref{negative tt*-Toda with x when n=2} of the form $w_0 = v_0 + \frac{i \pi}{2}$ where $v_0$ is real. This will be thoroughly explained in the following paper \cite{gikmo2} in the series. Hence, the real solution of our equation \eqref{negative tt*-Toda with x when n=2} necessarily implies $A^{\R} > 0$. 
\end{rem}

Summarizing, we have  that all Stokes matrices $S_n^{(0,\infty)}$ and connection matrices $E_n$ for $n=1,2$ depend on the  parameters $(A^\R,B,s^\R)$.

\begin{prop}\label{ABS TH}
The data $(A^{\R},B,s^\R)$ can be further parameterized by two real parameters $\{ s^{\R}, y^{\R} \}$, where $s^{\R}$ satisfies $-3 < s^{\R} < 1$ and $y^\R$ is arbitrary. In fact,
\begin{align}   
    & A^{\R} = \frac{1}{3(3 + s^{\R})} + \frac{2}{3 + s^{\R}} \sqrt{ \frac{1}{36} + \frac{3 + s^{\R}}{1 - s^{\R}} \left( \frac{1}{36} + (y^{\R})^2 \right)}, \label{parametrization of A}\\
    & B = \omega^2 \left( \frac{1}{3(3 + s^{\R})} - \frac{1 + s^{\R}}{3 + s^{\R}} \sqrt{ \frac{1}{36} + \frac{3 + s^{\R}}{1 - s^{\R}} \left( \frac{1}{36} + (y^{\R})^2 \right)} + i y^{\R} \right), \label{parametrization of B}
\end{align}
where $y^\R = \Im(\omega B)$.
\end{prop}

The proof of this proposition will be given in Appendix \ref{Proof of ABS}.
  
Recall that we called the parameter $s^{\R}$ Stokes data. Together with the new data $y^{\R}$ coming from connection matrix, we introduce the set of {\it monodromy data} $\mathcal{M}$ by
\begin{align}\label{monodormy data}
    \mathcal{M}=\{ (s^{\R}, y^{\R})\in \R\times\R \;\bigl| \; -3<s^\R<1 \}.
\end{align}

%% file: section3.tex
\section{Inverse Monodromy Problem and Asymptotics near $x = \infty$} \label{section 3}
\subsection{Riemann-Hilbert Problem} \label{section 6}

Recall that $\Psi_n^{(\infty, 0)}$ were originally defined on the Stokes sector $\Omega_n^{(\infty, 0)}$ and then were extended to the universal covering $\widetilde{\C^*}$ of $\C^*$ as was explained in the subsection \ref{subsection on stokes matrices}. In this section instead we consider $\Psi_n^{(\infty, 0)}$ as  holomorphic functions on the projected regions $\Omega_n^{(\infty, 0)}$ on $\C^*$. (We will indicate from which sheet $\Psi_n^{(\infty, 0)}$ are projected down to $\C^*$ in the following subsection.) In other words, these functions are sectionally holomorphic on $\C^*$ whose jumps are given by either the Stokes matrices $S_n^{(\infty, 0)}$ or the connection matrices $E_n$ which are parametrized by the monodromy data $m \in \mathcal{M}$, and whose asymptotics can be given as formal solutions. 

Then, we will reconstruct the coefficient $A(\zeta)$ in the equation \eqref{first eq} (and so the solution $w_0(x)$ of \eqref{negative tt*-Toda with x when n=2}) from the given monodromy data. 
\begin{align*}
    \mathcal{M} \ni m = \left( s^{\R}, \, y^{\R} \right) \mapsto w_0(x,\,m).
\end{align*}
This inverse monodromy problem forms a \textit{Riemann-Hilbert Problem}.

%%%%%%%%%%%
\subsubsection{$\hat{\Psi}$-problem} \label{hat-Psi RH problem}

Let us define a sectionally holomorphic function $\hat{\Psi}$ by Figure \ref{psi-hat problem}. The contour on this figure consists of the imaginary axis and a circle $S^1_{\rho} = \{ \zeta \in \C \,|\, |\zeta| = \rho \, (> 0) \}$ where the radius $\rho$ is arbitrary. We call this oriented contour $\Gamma_1$. Note that we picked holomorphic functions $\Psi_{n}^{(\infty, 0)}$ on $\widetilde{\C^*}$ by specifying which sheet they are from in Figure \ref{psi-hat problem} to define this sectionally holomorphic function $\hat{\Psi}$ on $\C^*$.
\begin{figure}[htbp]
    \centering
    \includegraphics[width=12cm]{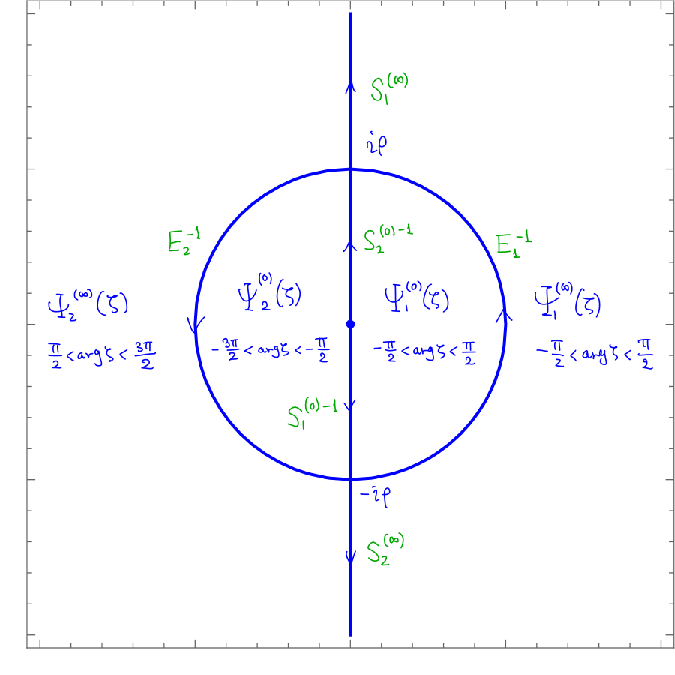}
    \caption{Riemann-Hilbert Problem of $\hat{\Psi}$ (original domains of $\Psi$ shown, for clarity).}
    \label{psi-hat problem}
\end{figure}

$\hat{\Psi}$ has jumps only on $\Gamma_1$ with the jump matrices shown in Figure \ref{psi-hat problem}. For example 
\begin{itemize}
    \item For $\zeta$ on the outer part of the $\arg \zeta = \frac{\pi}{2}$ ray, 
    $$\hat{\Psi}_+(\zeta) = \Psi^{(\infty)}_2(\zeta) = \Psi^{(\infty)}_1(\zeta) S_1^{(\infty)} = \hat{\Psi}_-(\zeta) S_1^{(\infty)}$$
    by \eqref{stokes matrices of 0 and infty}.
    \item For $\zeta$ on the semi-circle in the left half-plane,
    \begin{align*}
        \hat{\Psi}_+(\zeta) &= \Psi^{(0)}_2(\zeta) \underset{\text{$\pi / 2 < \arg \zeta < 3\pi/2$}}{=} \Psi^{(0)}_2(\zeta e^{- 2 \pi i})\\
        &= \Psi^{(\infty)}_2(\zeta) E_2^{-1} = \hat{\Psi}_-(\zeta) E_2^{-1}
    \end{align*}
    by \eqref{defn of connection matrices}.
\end{itemize}
At the self-intersection of $\Gamma_1$, points  $\zeta = i \rho$ and $\zeta = -i \rho$ in Figure \ref{psi-hat problem}, by Proposition \ref{E symmetry 1}, we have
\begin{align*}
      E_2^{-1}S_2^{(0)} E_1 S_1^{(\infty)} =  E_2^{-1}S_1^{(0)-1} E_1 S_2^{(\infty)-1}=I.
\end{align*}
In particular, we have no formal monodromy around these points.

From \eqref{formal solution near 0} and \eqref{formal solution near infty} the asymptotic behaviour of $\hat{\Psi}(\zeta)$ is 
\begin{align}
    &\hat{\Psi} (\zeta) = P_{\infty} (I + \O(\zeta^{-1})) e^{-x^2 \zeta d_3},\qquad  \zeta\rightarrow \infty, \\
    &\hat{\Psi} (\zeta) = P_0 (I + \O(\zeta)) e^{\frac{1}{\zeta} d_3},\qquad  \zeta\rightarrow 0. 
\end{align}
It is important to emphasize that $\hat{\Psi}(\zeta)$ is defined as a piecewise analytic function on the complex plane ${\mathbb C}^*$ and {\it not} on its universal covering; i.e., for $\hat{\Psi}(\zeta)$ we do not need to indicate $\arg \zeta$.

%%%%%%%%%%%%%%%%%%%%%%%%%%%%%%%%%%%%%%%

In the next subsection, we will deform the original $\hat{\Psi}$-problem by a series of transformations.

%%%%%%%%%%%%%%%%%%%%%%%%%%%%%%%%%%%%%%%%%%%%%%%%%%%%%%%%%%%%%%%%%%%

\subsubsection{$\Phi$-problem} \label{Phi RH problem}

First, we define a $\Phi$-problem by Figure \ref{phi problem}. The function $\Phi$ is a piecewise analytic function whose analytic pieces are obtained from the function $\hat{\Psi}$ by the proper right matrix multiplication as shown in Figure \ref{phi problem}.
This transformation comes from the decomposition of the Stokes matrices proved in equation \eqref{S's are expressed Q's} and equation \eqref{S is expressed by Qs at 0}. Note that our new oriented jump contours are the unit circle and infinite rays with $\frac{n \pi}{3}$ angles and we call it $\Gamma_2$.
\begin{figure}[htbp]
    \centering
    \includegraphics[width=12cm]{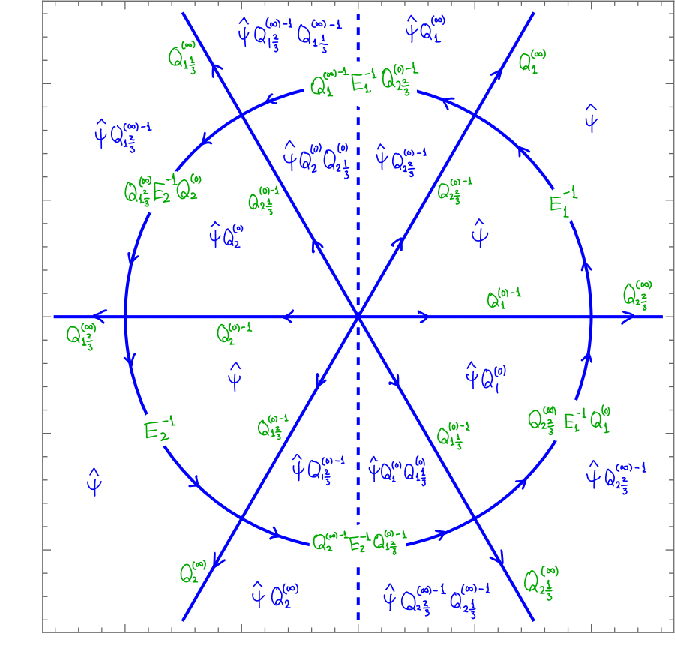}
    \caption{Contour $\Gamma_2$ and jump matrix $G_\Phi$ of  $\Phi$.}
    \label{phi problem}
\end{figure}
One can check that the new jump matrices on $\Gamma_2$ for this problem are the ones written in Figure \ref{phi problem} and we will denote them by $G_{\Phi}$. Moreover, $\Phi(\zeta)$ takes the same asymptotic behaviors at $\zeta=0$ and $\zeta=\infty$ as $\hat{\Psi}(\zeta)$ does.

%%%%%%%%%%%%%%%%%%%%%%%%%%%%%%%%%%%%%%%%%%%%%%%%%%%%%%%%%%%%%%%%%

\subsubsection{$\tilde{\Phi}$-problem}

Next, we define a $\tilde{\Phi}$-problem by
\begin{align}
    \tilde{\Phi} (\zeta) = \begin{dcases}
    \Phi(\zeta) & \text{ if } |\zeta| > \rho\\
    \Phi(\zeta) \frac{1}{3}C & \text{ if } |\zeta| < \rho
\end{dcases}. \label{defn of phi tilde}
\end{align}

One can check that this transformation doesn't affect the contour but the jump matrices will be changed: 
\begin{itemize}
    \item Jumps on the ray outside of $S^1_{\rho}$ do not change.
    \item Jumps on the ray inside of $S^1_{\rho}$ change from $G_{\Phi}$ to $C G_{\Phi} C$.
    \item Jumps on $S^1_{\rho}$ change from $G_{\Phi}$ to $\frac{1}{3} G_{\Phi} C$.
\end{itemize}
The jump matrices inside the circle can be further simplified. For example,
\begin{align*}
    C Q_{2\frac{2}{3}}^{(0)-1} C \underset{\text{\eqref{antisymmetry for Q ver.1}}}{=} C d_3 Q_{1\frac{2}{3}}^{(\infty)-1} d_3^{-1} C \underset{\text{\eqref{antisymmetry for Q ver.1}}}{=} C \left[ Q_{2\frac{2}{3}}^{(\infty)} \right]^T C \underset{\text{computation}}{=} Q_2^{(\infty)}
\end{align*}
One can repeat a similar discussion for the others.
Next, introducing $\tilde{E}_n^{-1} = \frac{1}{3} E_n^{-1} C$ for $n = 1,2$,  the jump matrices on $S^1_{\rho}$ can be also simplified, for instance,
\begin{align*}
    \frac{1}{3} Q_{1}^{(\infty)-1} E_1^{-1} Q_{2\frac{2}{3}}^{(0)-1} C = Q_{1}^{(\infty)-1} \left( \frac{1}{3} E_1^{-1} C \right ) \left( C Q_{2\frac{2}{3}}^{(0)-1} C \right) = Q_{1}^{(\infty)-1} \tilde{E}_1^{-1} Q_{2}^{(\infty)}.
\end{align*}
The updated jump matrices on $\Gamma_2$ are depicted in Figure \ref{tilde Phi problem}.
\begin{figure}[htbp]
    \centering
    \includegraphics[width=12cm]{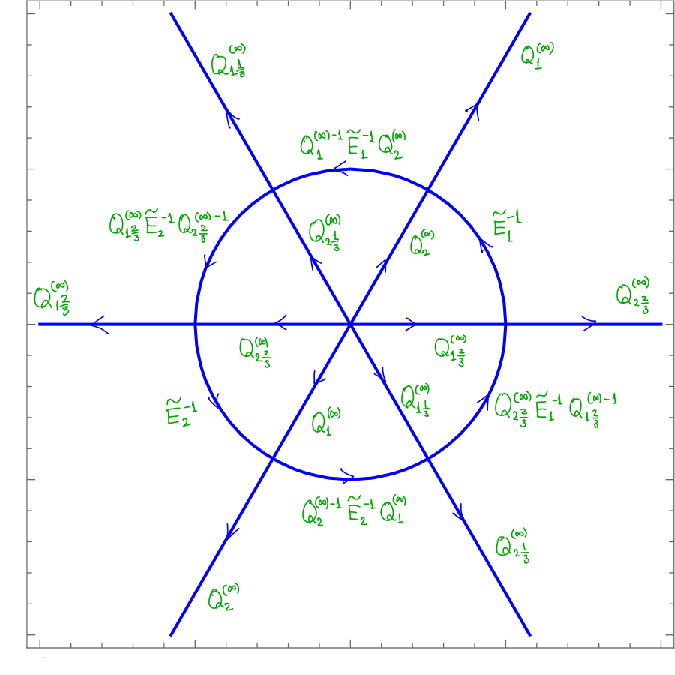}
    \caption{Contour $\Gamma_2$ and jump matrix $G_{\tilde{\Phi}}$ of    $\tilde{\Phi}$.}
    \label{tilde Phi problem}
\end{figure}

By definition of $\tilde{\Phi}(\zeta)$ and \eqref{formal solution near infty}, we observe that this transformation doesn't change  the   asymptotics at $\zeta=\infty$ :
\begin{align} 
    \tilde{\Phi}(\zeta) = P_{\infty} (I + \O(\zeta^{-1})) e^{-x^2 \zeta d_3}, \label{asymptotic solution of phi tilde near infty}
\end{align}
but at $\zeta = 0$, definition of $\tilde{\Phi}(\zeta)$ and \eqref{formal solution near 0} imply the following change:
\begin{align}
\begin{aligned}
    \tilde{\Phi}(\zeta) &= \frac{1}{3} P_{0} (I + \O(\zeta)) e^{\frac{1}{\zeta} d_3} C =\frac{1}{3} e^{-w} \Omega (I + \O(\zeta)) C  e^{\frac{1}{\zeta} d_3^{-1}}   \\
    &= \frac{1}{3} e^{-w} \Omega C (I + \O(\zeta)) e^{\frac{1}{\zeta} d_3^{-1}} = e^{-w} \Omega^{-1} (I + \O(\zeta)) e^{\frac{1}{\zeta} d_3^{-1}}.
\end{aligned} \label{asymptotic solution of phi tilde near 0}
\end{align}

%%%%%%%%%%%%%%%%%%%%%%%%%%%%%%%%%%

\subsubsection{$Y$-problem} \label{section Y problem}

Recall that $\rho$ in the original RHP was arbitrary. We now set $\rho = \frac{1}{x}$.
Motivated by the asymptotics \eqref{asymptotic solution of phi tilde near infty}, we introduce the following $Y$ function
\begin{align}
    Y(\zeta) := P_{\infty}^{-1} \tilde{\Phi} \left( \frac{\zeta}{x} \right) e^{-x \theta(\zeta)}, \label{defn of Y}
\end{align}
where $\theta (\zeta) = -\zeta d_3 + \frac{1}{\zeta} d_3^{-1}$. Since \eqref{defn of Y} is just a scaling, the jump contour does not change from $\Gamma_2$ as in Figure \ref{tilde Phi problem}; the only difference is that previously the circle has a radius $\rho > 0$, but now it is a unit circle.
However, jump matrices do change from $G_{\tilde{\Phi}}$ to $G_{Y}$:
\begin{align}\label{Gy}
\begin{aligned}
    &Y_{+}(\zeta) = P_{\infty}^{-1} \tilde{\Phi}_+ \left( \frac{\zeta}{x} \right) e^{-x \theta(\zeta)} = P_{\infty}^{-1} \tilde{\Phi}_- \left( \frac{\zeta}{x} \right) G_{\tilde{\Phi}} e^{-x \theta(\zeta)}\\
    &\qquad \; \; = Y_{-}(\zeta) e^{x \theta(\zeta)} G_{\tilde{\Phi}} e^{-x \theta(\zeta)},\\
    & \Rightarrow G_{Y}=  e^{x \theta(\zeta)} G_{\tilde{\Phi}} e^{-x \theta(\zeta)}.
\end{aligned}
\end{align}
For convenience, let us define 
\begin{align}
    &\check{Q}_m^{(0,\infty)}= e^{x \theta(\zeta)} Q_m^{(0,\infty)} e^{-x \theta(\zeta)},\qquad m\in\frac13 \Z, \\
    & \check{E}_n=e^{x \theta(\zeta)} \tilde{E}_n e^{-x \theta(\zeta)},\qquad n=1,2. \label{check E}
\end{align}

The normalization condition of this problem becomes the following:
\begin{itemize}
    \item As $\zeta \to \infty$, we have
    \begin{align*}
        Y(\zeta) = I + \O(\zeta^{-1})
    \end{align*}
    by \eqref{asymptotic solution of phi tilde near infty}.
    \item As $\zeta \to 0$, we have
    \begin{align}
        Y(\zeta) =  \Omega e^{-2w} \Omega^{-1} + \O(\zeta) \label{Y at 0}
    \end{align}
    by \eqref{asymptotic solution of phi tilde near 0} and the fact $P_{\infty} = e^{w} \Omega^{-1}$.
\end{itemize}

Therefore we pose the following RHP.
\begin{framed}
    \begin{RHP}\label{RHP1}
    Find a  matrix-valued function $Y$ satisfying the following conditions 
    \begin{itemize}
    \item $Y(\zeta) \in H(\C \setminus \Gamma_2)$,  
    \item The jump conditions are
    \begin{align*}
    Y_{+}(\zeta) &= Y_{-}(\zeta)   G_{Y}  ,
    \end{align*}
    \item The normalization condition is 
    \begin{align*}
        Y(\zeta) &= \begin{dcases}
        I + \O(\zeta^{-1}) & \zeta \to \infty,\\
        \Omega e^{-2w} \Omega^{-1} + \O(\zeta) & \zeta \to 0.
        \end{dcases}
    \end{align*} 
\end{itemize}
Here the contour $\Gamma_2$ and jump matrices $G_Y$ are depicted in Figure \ref{Y problem}.
    \end{RHP}
\end{framed}

\begin{figure}[htbp]
    \centering
    \includegraphics[width=12cm]{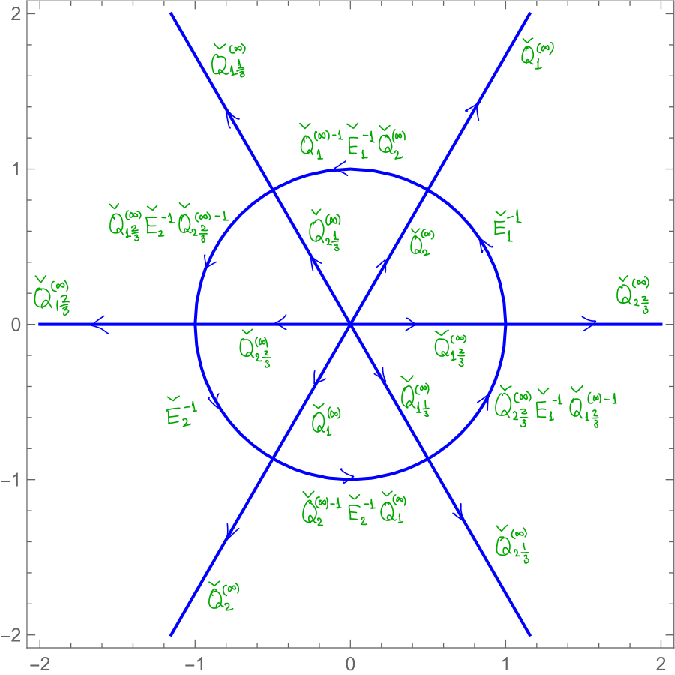}
    \caption{Contour $\Gamma_2$ and jump matrix $G_Y$   of $Y$.}
    \label{Y problem}
\end{figure}

%%%%%%%%%%%%%%%%%%%%%%%%%%%%%%%%%%%%%%%%%%%

\subsection{Further Deformation of the $Y$-problem}

Our goal is to solve the $Y$-problem asymptotically as $x \to \infty$. To this end, we will transform the $Y$-problem (if needed) to the one where jump matrices are close to the identity matrix with respect to the $L^2 \cap L^{\infty}$ norm so that one can apply the small norm theorem (see \cite{FIKN}, Theorem 8.1) and prove the asymptotic solvability of the RHP. Let us first analyze matrix $G_Y$.

\subsubsection{Large $x$ behavior of $G_{Y}$}\label{small norb sec}

Let 
\begin{equation}\label{phii} \begin{aligned} 
    \varphi_1(\zeta) &= (1 - \omega)\zeta - (1 - \omega^2)\frac{1}{\zeta} = \sqrt{3} \left(\zeta e^{- i\frac{\pi}{6}} - \frac{1}{\zeta} e^{i \frac{\pi}{6}}\right),\\
    \varphi_2(\zeta) &= (1 - \omega^2)\zeta - (1 - \omega)\frac{1}{\zeta} = \sqrt{3} \left(\zeta e^{i\frac{\pi}{6}} - \frac{1}{\zeta} e^{-i \frac{\pi}{6}}\right),\\
    \varphi_3(\zeta) &= (\omega - \omega^2)\zeta + (\omega - \omega^2)\frac{1}{\zeta} = i \sqrt{3} \left( \zeta + \frac{1}{\zeta} \right).
\end{aligned}\end{equation}
By \eqref{Gy}, we have
\begin{align}
    G_{Y} = \begin{pmatrix}
        * & * e^{-x \varphi_1 } & * e^{-x\varphi_2}\\
        * e^{x \varphi_1 } & * & * e^{-x \varphi_3}\\
        * e^{x \varphi_2} &  * e^{x \varphi_3} & *
    \end{pmatrix}, \label{jump matrices of Y problem}
\end{align}
where $*$ do not depend on $\zeta$.

To understand the sign-change of the real part of each exponent in \eqref{jump matrices of Y problem}, we shall compute the stationary points. First, consider $\varphi_1$. We have 
\begin{align}
&0 = \varphi_1'(\zeta) = \sqrt{3} (e^{-i\frac{\pi}{6}} + \zeta^{-2} e^{i \frac{\pi}{6}}) \quad \Rightarrow \quad \zeta = \pm e^{-i\frac{\pi}{3}},\\
&\varphi_1 \left( \pm e^{-i\frac{\pi}{3}} \right) = \mp 2\sqrt{3} i,\\
&\varphi_1''( \pm e^{-i\frac{\pi}{3}}) =   \pm 2 \sqrt{3} e^{i \frac{\pi}{6}}.
\end{align}
Let $\zeta = \xi + i\eta$, then 
\begin{align*}
    \Re (\varphi_1( \zeta )) =   \frac{(\sqrt{3} \xi+\eta)( \xi^2  + \eta^2-1) }{2(\xi^2 + \eta^2)}  .
\end{align*}
The zero level curve and the sign of $\Re(\varphi_1(\zeta))$ are shown in Figure \ref{Re phi_1}.
\begin{figure}[htbp]
\centering
\includegraphics[width=6cm]{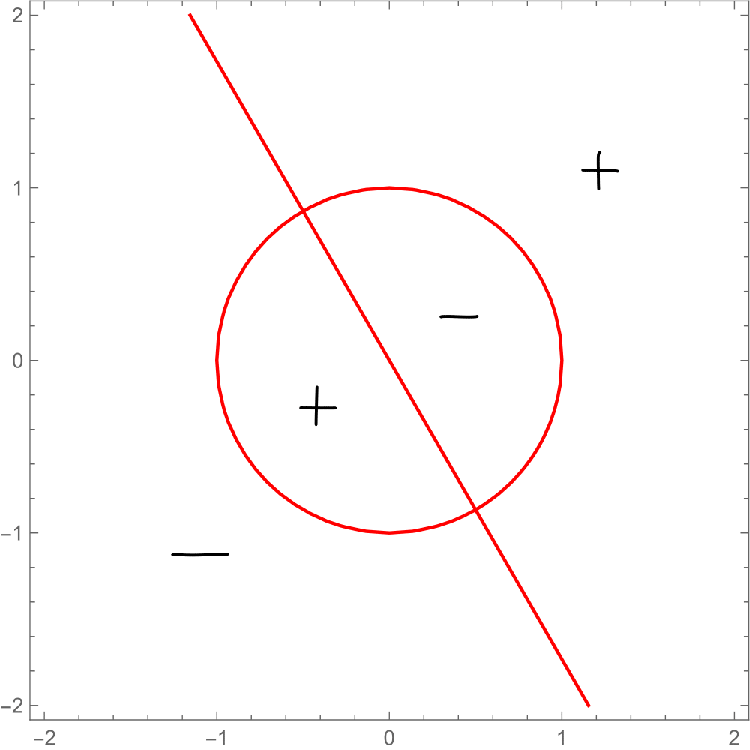}
\caption{The zero level curve and the sign of $\Re (\varphi_1(\zeta)) $.}
\label{Re phi_1}
\end{figure}

Next, consider  $\varphi_2$, we have 
\begin{align}
&0 = \varphi_2'(\zeta) = \sqrt{3} (e^{i\frac{\pi}{6}} + \zeta^{-2} e^{-i \frac{\pi}{6}}) \quad \Rightarrow\quad \zeta = \pm e^{i\frac{\pi}{3}},\\
 &\varphi_2 \left( \pm e^{i\frac{\pi}{3}} \right) = \pm 2\sqrt{3} i,\\
 &\varphi_2''( \pm e^{i\frac{\pi}{3}})   =  \pm 2 \sqrt{3} e^{-i \frac{\pi}{6}}.
\end{align}
and 
\begin{align*}
  \Re(\varphi_2(\zeta)) = \frac{(\sqrt{3} \xi-\eta)( \xi^2  + \eta^2-1) }{2(\xi^2 + \eta^2)}  .
\end{align*}
The zero level curve and the sign of $\Re(\varphi_2(\zeta))$ are shown in Figure \ref{Re phi_2}.
\begin{figure}[htbp]
\centering
\includegraphics[width=6cm]{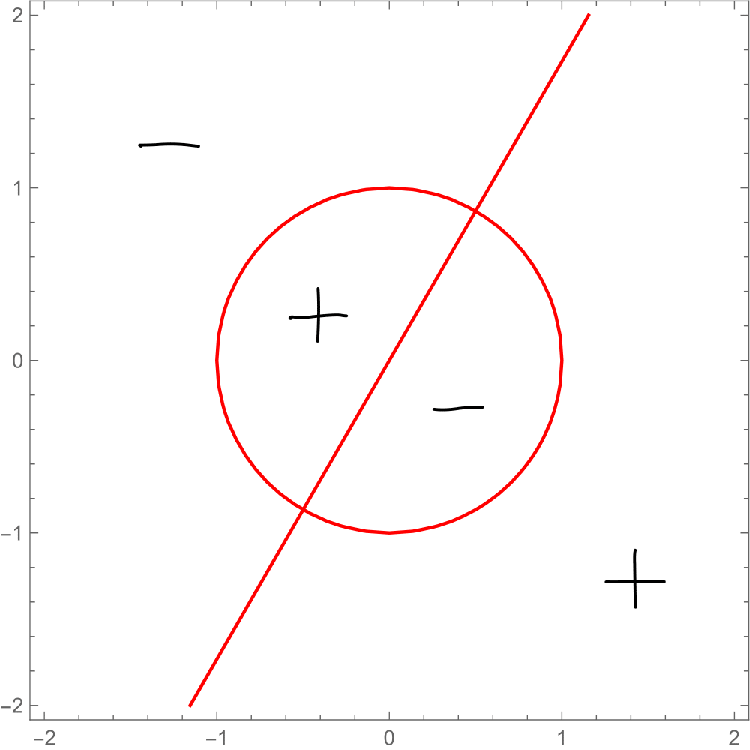}
\caption{The zero level curve and the sign of $\Re(\varphi_2(\zeta))$.}
\label{Re phi_2}
\end{figure}

Finally, consider $\varphi_3$; we have
\begin{align} 
&0 = \varphi_3'(\zeta) = i\sqrt{3} (1 - \zeta^{-2}) \quad \Rightarrow \quad \zeta = \pm 1,\\
&\varphi_3 \left( \pm 1 \right) = \pm 2\sqrt{3} i,\\
&\varphi_3''( \pm 1) = \pm 2 \sqrt{3} i,
\end{align}
 and 
\begin{align*}
    \Re (\varphi_3(\zeta))  = \sqrt{3}\left( -\eta + \frac{\eta}{\xi^2 + \eta^2} \right)  .
\end{align*}
The zero level curve and the sign of $\Re(\varphi_3(\zeta))$ are shown in Figure \ref{Re phi_3}.
\begin{figure}[htbp]
\centering
\includegraphics[width=6cm]{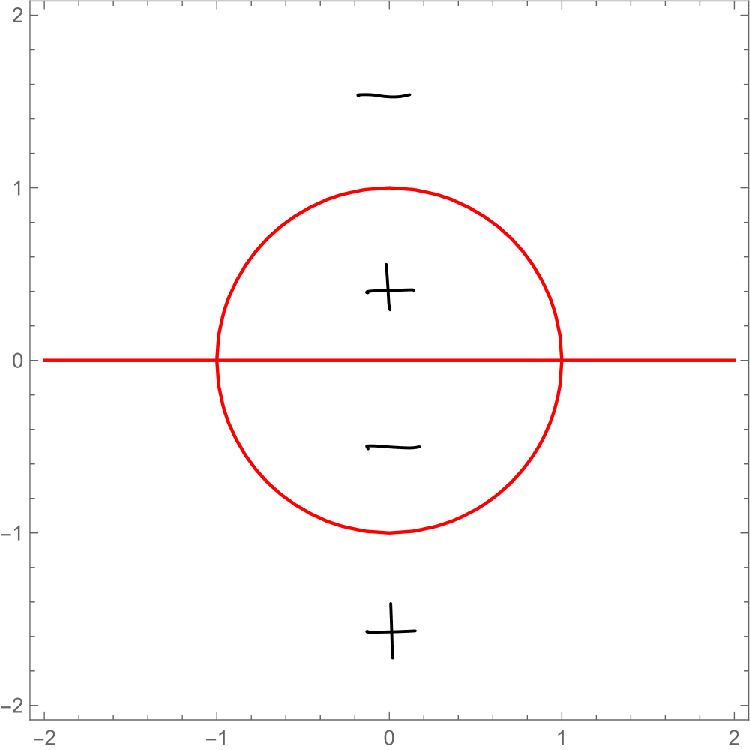}
\caption{The zero level curve and the sign of $\Re(\varphi_3(\zeta))$}
\label{Re phi_3}
\end{figure}

Using Figures \ref{Re phi_1}-\ref{Re phi_3}, one can verify that the jump matrices of the $Y$-problem on each ray approach the identity matrix in the $L^2 \cap L^{\infty}$ norm as $x \to \infty$. For instance,  
\begin{align*}
    e^{x \theta} Q_{1}^{(\infty)} e^{-x \theta} = 
    \begin{pmatrix}
        1 & a e^{-x \varphi_1(\zeta) } & 0\\
        0 & 1 & 0\\
        0 & 0 & 1
    \end{pmatrix},
\end{align*}
and from Figure \ref{Re phi_1} we have that  $|| ae^{-x \varphi_1(\zeta)} || \to 0$ when $\zeta$ belongs to the region described by Figure \ref{Y problem}. All other jumps on the rays can be checked similarly.

The remaining jump matrices to be considered are the ones on the unit circle, i.e., conjugated connection matrices. Since $G_Y$ has structure \eqref{jump matrices of Y problem} and each exponent is purely imaginary on the unit circle by \eqref{phii}, we see that they are oscillating (no decay), so we cannot apply the small norm theorem immediately to the $Y$-problem. To overcome this issue, we need to make further modifications.

%%%%%%%%%%%%%%%%%%%%%%%%%%%%%%%%%%%%%%%%%%%%%%%%%%%%%%%%%

\subsubsection{$\check{Y}$-problem}

First, we decompose the jump matrices of the $Y$-problem as follows
\begin{align}
\begin{aligned}
    & \check{E}_1^{-1}  = \check{L}_1 D_1 \check{R}_1, \\
    &\check{Q}_1^{(\infty)-1} \check{E}_1^{-1} \check{Q}_2^{(\infty)}  = \check{L}_2 D_2 \check{R}_2,\\ 
    & \check{Q}_{1\frac{2}{3}}^{(\infty)} \check{E}_2^{-1} \check{Q}_{2\frac{2}{3}}^{(\infty)-1}  = \check{L}_3 D_3 \check{R}_3, \\
    & \check{E}_2^{-1}  = \check{L}_4 D_4 \check{R}_4, \\
    & \check{Q}_2^{(\infty)-1} \check{E}_2^{-1} \check{Q}_1^{(\infty)}  = \check{L}_5 D_5 \check{R}_5, \\ 
    & \check{Q}_{2\frac{2}{3}}^{(\infty)} \check{E}_1^{-1} \check{Q}_{1\frac{2}{3}}^{(\infty)-1}  = \check{L}_6 D_6 \check{R}_6,
\end{aligned} \label{LDR decompositions}
\end{align}
so that $\check{L}_k$ and $\check{R}_k $ have diagonal entries equal to $1$ and non-diagonal entries exponentially decaying as $x \to \infty$, and $D_k$ are some constant diagonal matrices. We discuss it in subsection \ref{decomp of jump mat of Y-check prob}.

Then, one can transform the $Y$-problem to the following $\check{Y}$-problem defined in Figure \ref{Y check problem}. Such procedure is called  {\it opening lenses}.  We call the new oriented contour $\Gamma_3$ and the jump matrices $G_{\check{Y}}$. After this decomposition all jump matrices $G_{\check{Y}}$ except $D_k$'s tend to the identity matrix as $x \to \infty$. 

\begin{framed}
    \begin{RHP}\label{YcheckRHP}
    Find a  matrix-valued function $\check{Y}$ satisfying the following conditions 
    \begin{itemize}
    \item $\check{Y}(\zeta) \in H(\C \setminus \Gamma_3)$,  
    \item The jump conditions are
    \begin{align*}
    \check{Y}_{+}(\zeta) &= \check{Y}_{-}(\zeta)   G_{\check{Y}}  ,
    \end{align*}
    \item The normalization condition is 
    \begin{align*}
        \check{Y}(\zeta) &= \begin{dcases}
        I + \O(\zeta^{-1}) & \zeta \to \infty,\\
        \Omega e^{-2w} \Omega^{-1} + \O(\zeta) & \zeta \to 0.
        \end{dcases}
    \end{align*} 
    \end{itemize}
    \end{RHP}
\end{framed}

Next, we solve a model RHP that has jumps only on the unit circle and the corresponding jump matrices are   $D_k$'s. We call its solution a \textit{global parametrix} $\check{Y}^D(\zeta)$ and discuss it in subsection  \ref{global param}. Then the function $\check{Y}(\zeta) \left[ \check{Y}^D \right]^{-1}(\zeta)$ will have jump matrices which tend to the identity as $x \to \infty$. However, the convergence of such jump matrices to the identity matrix is not in the $L^{\infty}$ norm if $\zeta$ is near the stationary points.
This suggests constructing other model RHPs near the stationary points. We call their solutions \textit{local parametrices} $P^{(p_k)}(\zeta)$ where $p_k$'s are stationary points and discuss them in subsections \ref{param 1}, \ref{other param}. 
Finally, in subsection \ref{error sec} we consider a small norm RHP for an error function $R(\zeta)$  by comparing $\check{Y}(\zeta)$ with its parametrices.

\begin{rem}
In analogy with the classical steepest descent method, the main contribution to the asymptotics of $\check{Y}(\zeta)$ as $x \to \infty$ comes from the unit circle and from the neighborhood of the stationary points.
\end{rem}

\begin{figure}[htbp]
\centering
\includegraphics[width=12cm]{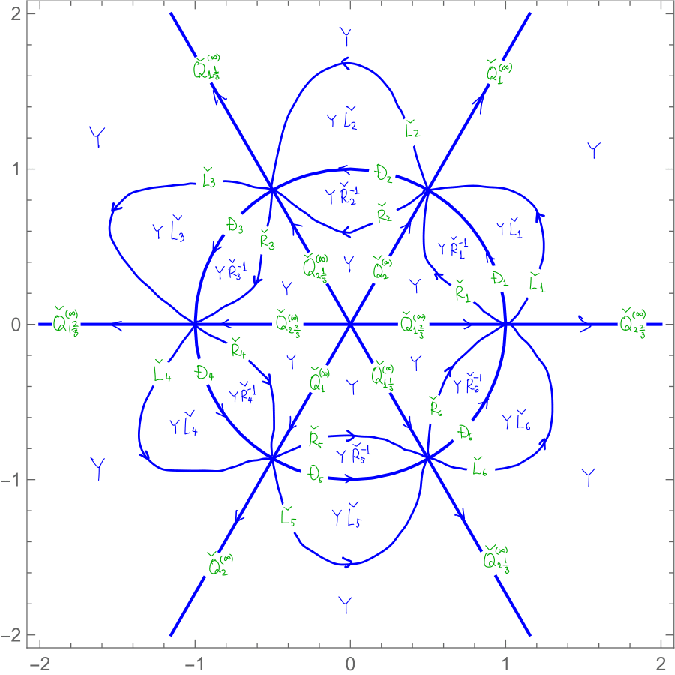}
\caption{Contour $\Gamma_3$ and jump matrices of the $\check{Y}$-problem.}
\label{Y check problem}
\end{figure}

%%%%%%%%%%%%%%%%%%%%%%%%%%%%%%

\subsubsection{Decomposition of Jump Matrices on the Unit Circle} \label{decomp of jump mat of Y-check prob}

The structures of the $\check{L}_k$ and $\check{R}_k$ in \eqref{LDR decompositions}  are dictated by the  $\Re(\varphi_1)$, $\Re(\varphi_2)$, and $\Re(\varphi_3)$. 
For instance, $\check{L}_1$ and $\check{R_1}$ in        $\check{E}_1^{-1} =\check{L}_1 D_1 \check{R}_1 $, need to be of the form
\begin{align}\label{L1R1}
   \check{L}_1= \begin{pmatrix}
        1 & * e^{-x \varphi_1 } & * e^{-x\varphi_2}\\
        0  & 1 & 0\\
        0  &  * e^{x \varphi_3} & 1
    \end{pmatrix},\quad 
    \check{R}_1=
    \begin{pmatrix}
        1 & 0 & 0\\
        * e^{x \varphi_1 } & 1 & * e^{-x \varphi_3}\\
        * e^{x \varphi_2} &  0 & 1
    \end{pmatrix},
\end{align} 
since the non-diagonal entries are exponentially decaying when $\zeta$ belongs to the corresponding regions in Figure \ref{Y check problem}, as $x \to \infty$. All the other structures of $\check{L}_k$ and $\check{R}_k$ can be obtained similarly. One can verify that indeed such structured matrix decompositions exist by direct computations written in Appendix \ref{Appendix E decomp}. Thus, we obtain the following proposition.

\begin{prop}\label{Prop E decomp}
Recall $\check{L}_k$ and $\check{R}_k$ introduced in \eqref{LDR decompositions}. Let
\begin{align}\label{LLcheck}
    L_k\equiv e^{-x \theta(\zeta)} \check{L}_k e^{x \theta(\zeta)},\quad R_k\equiv e^{-x \theta(\zeta)}\check{R}_k e^{x \theta(\zeta)}.
\end{align} 
Then, one has
\begin{align}
\begin{aligned}
    L_1 &= \begin{pmatrix}
     1 & - \frac{B}{A^{\R}} - \omega^2 \frac{|B|^2}{(A^{\R})^2} & -\frac{\overline{B}}{A^{\R}}\\
     0 & 1 & 0\\
     0 & \omega^2 \frac{B}{A^{\R}} & 1
    \end{pmatrix}, \;
    D_1 = \begin{pmatrix}
        \frac{1}{3A^{\R}} & & \\
        & 3 A^{\R} & \\
        & & 1
    \end{pmatrix},\\
    R_1 &= \begin{pmatrix}
        1 & 0 & 0\\
        - \frac{\overline{B}}{A^{\R}} - \omega \frac{|B|^2}{(A^{\R})^2} & 1 & \omega\frac{\overline{B}}{A^{\R}}\\
        - \frac{B}{A^{\R}} & 0 & 1
    \end{pmatrix}.
\end{aligned} \label{parametrization of L_1 D_1 R_1 ver.2}
\end{align}
\begin{align}
\begin{aligned}
    L_2 &= \begin{pmatrix}
     1 & \omega\frac{\overline{B}}{A^{\R}} & 0\\
     0 & 1 & 0\\
     - \frac{B}{A^{\R}} & - \omega \frac{|B|^2}{(A^{\R})^2} - \omega s^{\R} - \frac{\overline{B}}{A^{\R}} & 1
    \end{pmatrix}, \;
    D_2 = \begin{pmatrix}
       1 & & \\
        & 3 A^{\R} & \\
        & & \frac{1}{3A^{\R}}
    \end{pmatrix},\\
    R_2 &= \begin{pmatrix}
        1 & 0 & - \frac{\overline{B}}{A^{\R}}\\
        \omega^2 \frac{B}{A^{\R}} & 1 & - \omega^2 \frac{|B|^2}{(A^{\R})^2} - \omega^2 s^{\R} - \frac{B}{A^{\R}}\\
        0 & 0 & 1
    \end{pmatrix}.
\end{aligned} \label{parametrization of L_2 D_2 R_2 ver.2}
\end{align}
\begin{align}
\begin{aligned}
    L_3 &= \begin{pmatrix}
     1 & 0 & 0\\
     \omega^2 \frac{B}{A^{\R}} & 1 & 0\\
     -\omega^2 \frac{|B|^2}{(A^{\R})^2} - \frac{B}{A^{\R}} & - \frac{\overline{B}}{A^{\R}} & 1
    \end{pmatrix}, \; 
    D_3 = \begin{pmatrix}
       3 A^{\R} & & \\
        & 1 & \\
        & & \frac{1}{3A^{\R}}
    \end{pmatrix},\\
    R_3 &= \begin{pmatrix}
        1 & \omega \frac{\overline{B}}{A^{\R}} & -\omega \frac{|B|^2}{(A^{\R})^2} - \frac{\overline{B}}{A^{\R}}\\
        0 & 1 & -\frac{B}{A^{\R}}\\
        0 & 0 & 1
    \end{pmatrix}.
\end{aligned} \label{parametrization of L_3 D_3 R_3 ver.2}
\end{align}
\begin{align}
\begin{aligned}
    L_4 &= \begin{pmatrix}
     1 & 0 & 0\\
    -\omega \frac{|B|^2}{(A^{\R})^2} - \omega s^{\R} - \frac{\overline{B}}{A^{\R}} & 1 & - \frac{B}{A^{\R}}\\
     \omega \frac{\overline{B}}{A^{\R}} & 0 & 1
    \end{pmatrix}, \;
    D_4 = \begin{pmatrix}
       3 A^{\R} & & \\
        & \frac{1}{3 A^{\R}} & \\
        & & 1
    \end{pmatrix},\\
    R_4 &= \begin{pmatrix}
        1 & -\omega^2 \frac{|B|^2}{(A^{\R})^2} - \omega^2 s^{\R} - \frac{B}{A^{\R}} & \omega^2 \frac{B}{A^{\R}}\\
        0 & 1 & 0\\
        0 & -\frac{\overline{B}}{A^{\R}} & 1
    \end{pmatrix}.
\end{aligned} \label{parametrization of L_4 D_4 R_4 ver.2}
\end{align}
\begin{align}
\begin{aligned}
    L_5 &= \begin{pmatrix}
     1 & 0 & \omega^2 \frac{B}{A^{\R}}\\
     -\frac{\overline{B}}{A^{\R}} & 1 & - \omega^2 \frac{|B|^2}{(A^{\R})^2} - \frac{B}{A^{\R}}\\
     0 & 0 & 1
    \end{pmatrix}, \; 
    D_5 = \begin{pmatrix}
       1 & & \\
        & \frac{1}{3A^{\R}} & \\
        & & 3 A^{\R}
    \end{pmatrix},\\
    R_5 &= \begin{pmatrix}
        1 & -\frac{B}{A^{\R}} & 0\\
        0 & 1 & 0\\
        \omega \frac{\overline{B}}{A} & - \omega \frac{|B|^2}{(A^{\R})^2} - \frac{\overline{B}}{A^{\R}} & 1
    \end{pmatrix}.
\end{aligned} \label{parametrization of L_5 D_5 R_5 ver.2}
\end{align}
\begin{align}
\begin{aligned}
    L_6 &= \begin{pmatrix}
     1 & - \frac{B}{A^{\R}} & -\omega \frac{|B|^2}{(A^{\R})^2} - \omega s^{\R} - \frac{\overline{B}}{A^{\R}}\\
     0 & 1 & \omega \frac{\overline{B}}{A^{\R}}\\
     0 & 0 & 1
    \end{pmatrix}, \; 
    D_6 = \begin{pmatrix}
       \frac{1}{3A^{\R}} & & \\
        & 1 & \\
        & & 3 A^{\R}
    \end{pmatrix}\\
    R_6 &= \begin{pmatrix}
        1 & 0 & 0\\
        -\frac{\overline{B}}{A^{\R}} & 1 & 0\\
        -\omega^2 \frac{|B|^2}{(A^{\R})^2} - \omega^2 s^{\R} - \frac{B}{A^{\R}} & \omega^2 \frac{B}{A^{\R}} & 1
    \end{pmatrix}.
\end{aligned} \label{parametrization of L_6 D_6 R_6 ver.2}
\end{align}
\end{prop}

%%%%%%%%%%%%%%%%%%%%%%%%

\subsubsection{Model Riemann-Hilbert Problems}

Let us name the arcs of the unit circle in Figure \ref{arcs} by $C_k$ and the intersection points of the unit circle with the rays by $p_k$ where $k = 1, 2, \cdots, 6$. Thus,
\begin{align}
    p_1=1, \;\; p_2=-\bar{\omega},\;\; p_3=\omega,\;\; p_4=-1,\; p_5=\bar{\omega},\;\; p_6=-\omega. \label{defn of points}
\end{align}

\begin{figure}[htbp]
\centering%\label{arcs}
\includegraphics[width=8cm]{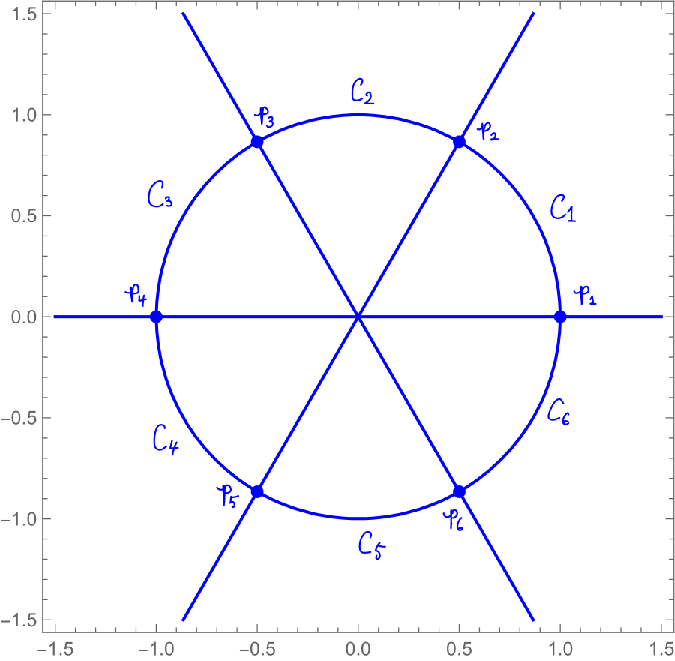}
\caption{Arcs $C_k$ and stationary points $p_k$.}
\label{arcs}
\end{figure}

The following sections of model RHPs consist of three parts. First, we solve a model RHP \ref{RHP2} on the unit circle $S^1 = \bigcup_{k=1}^{6} C_k$ to get the global parametrix $\check{Y}^D$. Then, we consider a model RHP near $p_1$ to get a local parametrix $P^{(1)}(\zeta)$. Lastly, we obtain the other local parametrices: 
\begin{align*}
    P^{(p_k)}(\zeta) \; \text{ for } \; k =2, \cdots, 6,
\end{align*}
from $P^{(1)}(\zeta)$ and symmetry relations.

%%%%%%%%%%%%%%%%%%%%%%%%%%%%%%%%%%%%%%%%%%%%%%%

\subsubsection{Global Parametrix}\label{global param}

Consider the following model RHP for the global parametrix.

\begin{framed}
\begin{RHP}\label{RHP2}
Find a matrix-valued function $\check{Y}^D(\zeta)$ satisfying
\begin{itemize}
    \item $\check{Y}^D(\zeta) \in H(\C \setminus S^{1})$.
    \item On the unit circle oriented counterclockwise, the following jump condition is valid:
    \begin{align*}
        \check{Y}^D_+(\zeta) = \check{Y}^D_-(\zeta) G_D,
    \end{align*}
    where the jump matrices are
    \begin{align*}
        G_D = \begin{dcases}
                D_1 & \text{if } \zeta \in C_1\\
                D_2 & \text{if } \zeta \in C_2\\
                D_3 & \text{if } \zeta \in C_3\\
                D_4 & \text{if } \zeta \in C_4\\
                D_5 & \text{if } \zeta \in C_5\\
                D_6 & \text{if } \zeta \in C_6\\
    \end{dcases}
    \end{align*}
    with $D_1, \cdots, D_6$ given in Proposition \ref{Prop E decomp} .
    \item The normalization condition is $$\check{Y}^D(\zeta) =I+\O(\zeta^{-1})\; \text{ as}\;  \zeta \to \infty.$$
\end{itemize}
\end{RHP}
\end{framed}

\begin{thm}
The solution of the RHP \ref{RHP2} is
\begin{align}
\begin{aligned}
    \check{Y}^{D}(\zeta) &= f_1(\zeta) f_2(\zeta) f_3(\zeta) f_4(\zeta) f_5(\zeta) f_6(\zeta)
\end{aligned} \label{global parametrix}
\end{align}
where
\begin{align*}
    f_1(\zeta) &= \left( \frac{\zeta - 1}{\zeta - 1/2 -i\sqrt{3}/2} \right)^{-\frac{1}{2 \pi i} \ln D_1}, \quad f_2(\zeta) = \left( \frac{\zeta - 1/2 -i\sqrt{3}/2}{\zeta + 1/2 -i\sqrt{3}/2} \right)^{-\frac{1}{2 \pi i} \ln D_2},\\
    f_3(\zeta) &= \left( \frac{\zeta + 1/2 -i\sqrt{3}/2}{\zeta + 1} \right)^{-\frac{1}{2 \pi i} \ln D_3},\quad f_4(\zeta) = \left( \frac{\zeta + 1}{\zeta + 1/2 + i\sqrt{3}/2} \right)^{-\frac{1}{2 \pi i} \ln D_4},\\
    f_5(\zeta) &= \left( \frac{\zeta + 1/2 + i\sqrt{3}/2}{\zeta - 1/2 + i\sqrt{3}/2} \right)^{-\frac{1}{2 \pi i} \ln D_5},\quad f_6(\zeta) = \left( \frac{\zeta - 1/2 + i\sqrt{3}/2}{\zeta -1} \right)^{-\frac{1}{2 \pi i} \ln D_6},
\end{align*}
and each $f_k(\zeta)$ is holomorphic in $\overline{\C} \setminus C_k$ and normalized so that $f_k (0) = D_k^{1/6}$ for $k = 1, 2, \cdots, 6$.
\end{thm}

\begin{proof}
We check that $\check{Y}^D(\zeta)$ defined by \eqref{global parametrix} satisfies RHP \ref{RHP2}. 

(i)
Observe that each $f_k(\zeta)$ is a composition of a linear fractional transformation and a power function:
\begin{align*}
    \zeta \mapsto g_k(\zeta) = \frac{\zeta - p_k}{\zeta - p_{k+1}} \mapsto g_k(\zeta)^{z_k},
\end{align*}
where $p_k$ and $p_{k+1}$ are the endpoints of the arc $C_k$ (see Figure \ref{arcs}), $z_k = -\frac{1}{2\pi i} \ln D_k$. Such $g_k$ maps the arc $C_k$ to the ray from the origin with angle $5\pi/6$. For each $k$, we select this ray as the branch cut of the complex logarithm to obtain
\begin{align}
f_k (\lambda) = \exp \left\{ -\frac{\ln D_k}{2 \pi i} L_{5\pi/6} \left( \frac{\zeta - p_k}{\zeta - p_{k+1}} \right) \right\}
\end{align}
where \( L_{5\pi/6}(z) := \ln|z| + iA_{5\pi/6} (z) \) and $A_{5\pi/6} (z)$ is the argument of \( z \) that belongs to \( ( -7\pi/6, 5\pi/6] \) (we will use the fact that \( L_{5\pi/6} \left( \frac{1}{2} - \frac{\sqrt{3}}{2}i \right) = -\frac{\pi i}{3} \) later). 

This proves that each $f_k$ is holomorphic everywhere except the arc $C_k$, and therefore, their product,
$$
\check{Y}^{D}(\zeta) = f_1(\zeta) f_2(\zeta) \cdots f_6(\zeta),
$$
is holomorphic everywhere except  $S^{1}$.

(ii)
Then, one can readily check that
\[
L_{5\pi/6, +}(z) = L_{5\pi/6, -}(z) - 2\pi i
\]
on the ray \( \{z \in \C \,|\, \arg(z) = 5\pi/6 \} \) oriented from the origin to infinity. Thus, it holds that
\begin{align}
    f_{k,+}(\zeta) = f_{k,-}(\zeta) D_k
\end{align}
on each \( C_k \). This immediately shows that the jump condition is fulfilled since all matrices are diagonal and commute. 

(iii) Finally, as \( L_{5\pi/6}(1) = 0 \), it holds that \( f_k (\infty) = I\) for each $k$, which shows that the normalization condition is satisfied and finishes the proof.
\end{proof}

\begin{rem}
In the neighborhood of $p_1 = 1$, we consider a small disk $U_1$ centered at that point. Introduce the following regions:
\begin{align*}
    K^{(1)} &= U_1 \cap \text{(outside of the unit circle)}\\
    K^{(2)} &= U_1 \cap \text{(inside of the unit circle)} \cap \text{(upper half plane)}\\
    K^{(3)} &= U_1 \cap \text{(inside of the unit circle)} \cap \text{(lower half plane)}.
\end{align*}

\begin{figure}[H]
    \centering
    \begin{subfigure}{.33\textwidth}
    \includegraphics[width=0.95\linewidth]{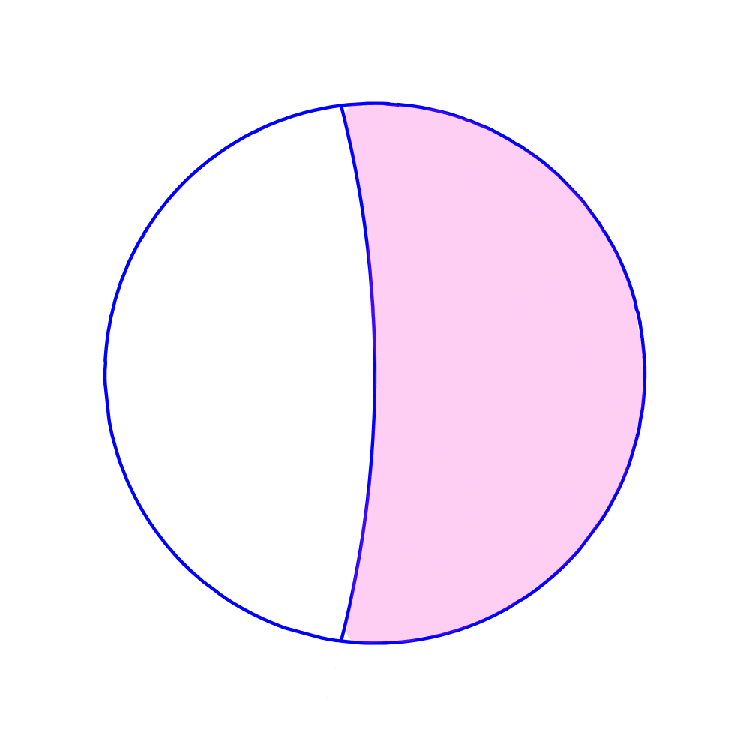}
    \caption{$K^{(1)}$}
    \label{Picture of K^{(1)}}
    \end{subfigure}%
    \begin{subfigure}{.33\textwidth}
    \includegraphics[width=0.95\linewidth]{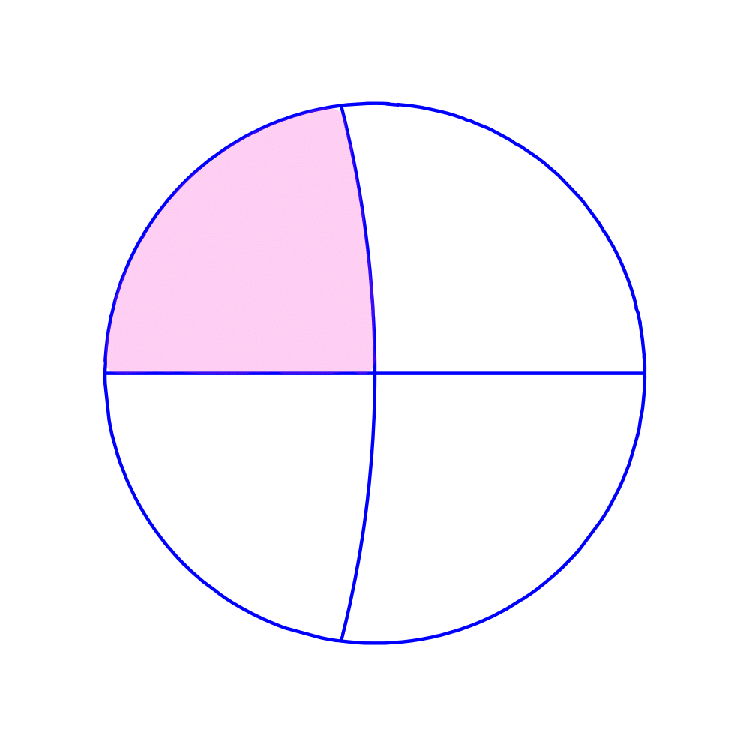}
    \caption{$K^{(2)}$}
    \label{Picture of K^{(2)}}
    \end{subfigure}
    \begin{subfigure}{.33\textwidth}
    \includegraphics[width=0.95\linewidth]{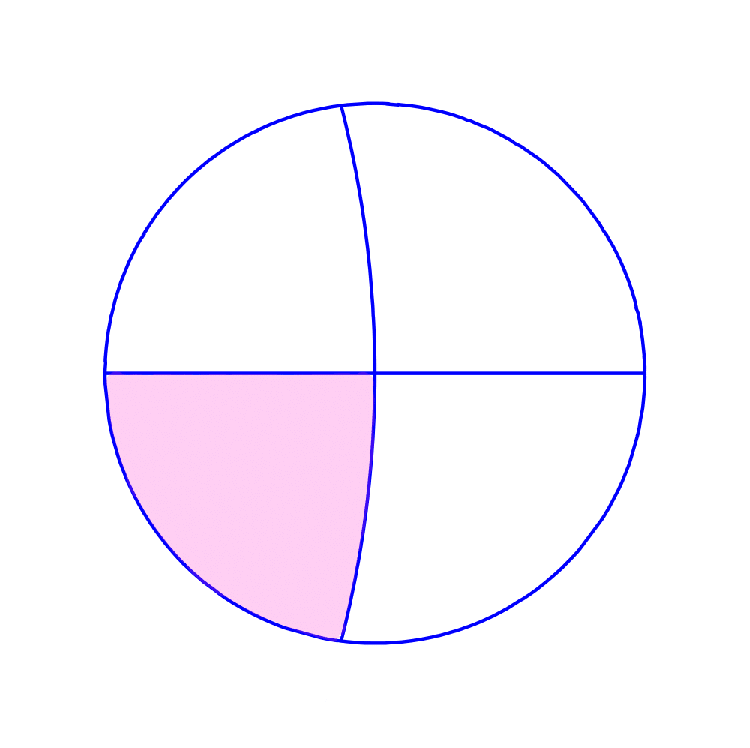}
    \caption{$K^{(3)}$}
    \label{Picture of K^{(3)}}
    \end{subfigure}
    
    \caption{Regions $K^{(i)}$}
    \label{Pictures of K^{(i)}}
\end{figure}

Originally, we chose $\overline{\C} \setminus C_k$ as the domain for a branch $f_k$. But here we let $(\zeta - 1)^{\nu}$ with $\nu = -\frac{1}{2 \pi i} \ln 3A^{\R}$ be a branch with the cut $\{ \zeta \leq 1\}$ where we take $\arg (\zeta - 1) \in (-\pi, \pi]$.
Then the global parametrix $\check{Y}^D(\zeta)$ admits the following representation in the neighborhood $U_1$ of $\zeta = 1$,
\begin{align}
\begin{aligned}
\check{Y}^D (\zeta) = \Theta (\zeta)
\begin{pmatrix}
1 & 0 & 0 \\
0 & (\zeta - 1)^{\nu} & 0\\
0 & 0 & (\zeta - 1)^{-\nu}
\end{pmatrix}
\begin{dcases}
    I & \text{for $\zeta \in K^{(1)}$}\\
    D_1 & \text{for $\zeta \in K^{(2)}$}\\
    D_6 & \text{for $\zeta \in K^{(3)}$},
\end{dcases}
\end{aligned} \label{asymptotics of Y^D}
\end{align}
where $\Theta (\zeta)$ is locally analytic for $\zeta \in U_1$: 
\begin{align*}
     &\Theta (\zeta)= F \left(I +\sum_{l=1}^\infty \Theta_l (\zeta-1)^l \right), \quad F =
    \begin{pNiceMatrix}[margin]
         e^{\pi i \nu} & & 0  & & & 0  & \\
         0 & & & \Block{2-2}{e^{\frac{\pi i \nu}{2}} \left( 2\sqrt{3}\right)^{-\nu \sigma_3} } \\
         0      
    \end{pNiceMatrix}.
\end{align*}

\end{rem}

%%%%%%%%%%%%%%%%%%%%%%%%%%%%%%%%%%%%%%%

\subsubsection{Local Parametrix near $\zeta=1$}\label{param 1}
 We are interested in finding a {\it local parametrix} $P^{(1)}(\zeta)$, such that  

\begin{itemize}
    \item $P^{(1)}(\zeta)$ has the same jumps as  the $\check{Y}$-problem in the disk $U_1$ of $\zeta=1$,
    \item on the boundary $\partial U_1$, we have the following matching relation
    \begin{align}
        &P^{(1)}(\zeta) = \check{Y}^D(\zeta) (I + \O(1/\sqrt{x})), \quad x \to \infty. \label{matching up cond appearing for the first time}
    \end{align} 
\end{itemize}

Recall that the contour and the jump matrices of the $\check{Y}$-problem near $\zeta = p_1$ (Note that $p_1 = 1$ from \eqref{defn of points}; we might switch from one to another but they both mean the same point) are the ones in Figure \ref{local picture of Phi0 check near 1}, where  
\begin{align*}
    \check{L}_k\equiv e^{x\theta}L_ke^{-x\theta},\quad \check{R}_k\equiv e^{x\theta}R_ke^{-x\theta},
\end{align*}
and the matrices $L_k, R_k$ are given in Proposition \ref{Prop E decomp}.
\begin{figure}[H]
    \centering
    \includegraphics[width=8cm]{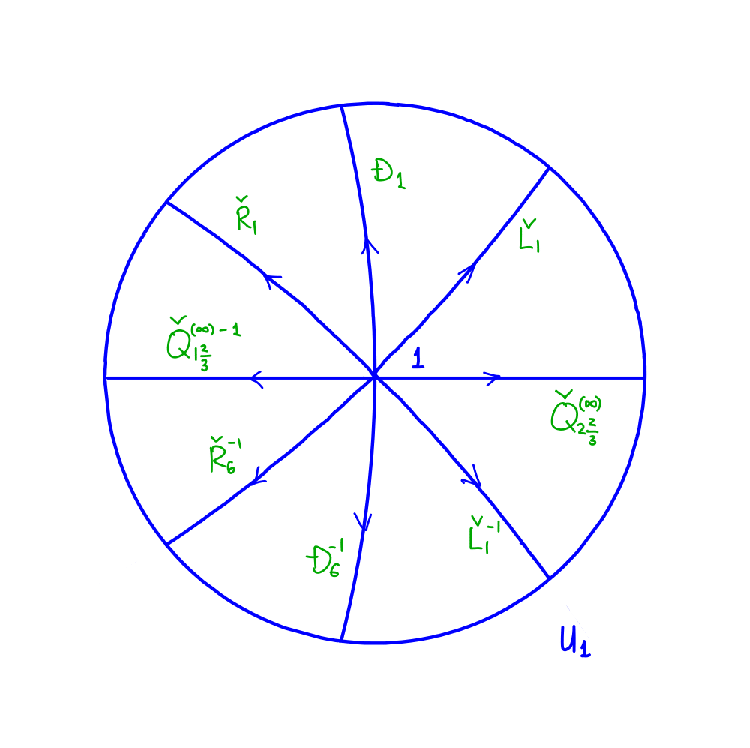}
    \caption{$\check{Y}$-problem near $\zeta=1$.}
    \label{local picture of Phi0 check near 1}
\end{figure}

Observe that, if we locally define $\check{Y}_{p_1}(\zeta)$ by the transformation shown in Figure \ref{defn of Phi1 check near 1}, where
\begin{align}
    \check{N} = \begin{pmatrix}
        1 & 0 & 0\\
        \frac{\overline{B}}{A^{\R}} e^{x \varphi_1} & 1 & 0\\
        \frac{B}{A^{\R}} e^{x \varphi_2} & 0 & 1
    \end{pmatrix}, \quad
    \check{M} = \begin{pmatrix}
        1 & -\frac{B}{A^{\R}} e^{-x \varphi_1} & -\frac{\overline{B}}{A^{\R}} e^{-x \varphi_2}\\
        0 & 1 & 0\\
        0 & 0 & 1
    \end{pmatrix}
\end{align}
then we obtain a new contour and jump matrices shown in Figure \ref{local picture of Phi1 check near 1}. There we introduced
\renewcommand\thesubfigure{\Alph{subfigure}}
\begin{figure}[H]
    \centering
    \begin{subfigure}{.5\textwidth}
    \includegraphics[width=1\linewidth]{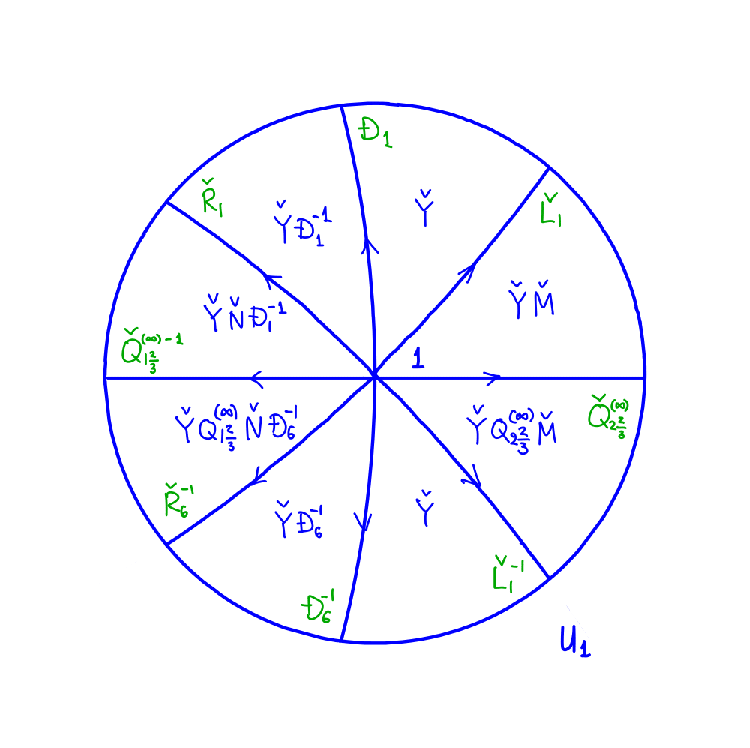}
    \caption{}
    \label{defn of Phi1 check near 1}
    \end{subfigure}%
    \begin{subfigure}{.5\textwidth}
    \includegraphics[width=1\linewidth]{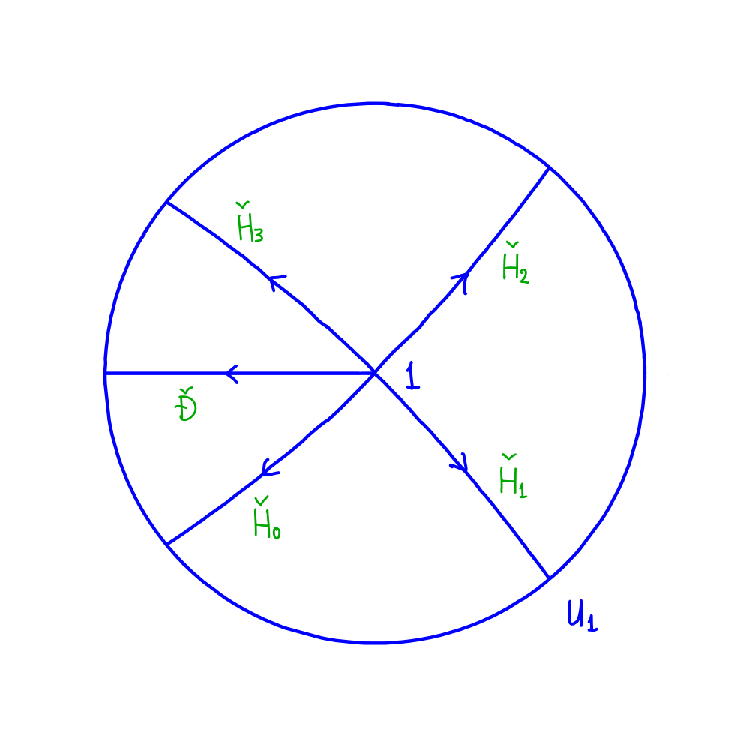}
    \caption{}
    \label{local picture of Phi1 check near 1}
    \end{subfigure}
    \caption{Definition and contour of $\check{Y}_{p_1}$.}
    \label{Local phi cont}
\end{figure}
\begin{equation} \label{ H matrix}
\begin{aligned}
    &\check{H}_2 = \check{M}^{-1} \check{L}_1 = \begin{pmatrix}
        1 & 0 & 0\\
        0 & 1 & 0\\
        0 & s_1 e^{x \varphi_3} & 1
    \end{pmatrix}, \\
    &\check{H}_3 = D_1 \check{R}_1 \check{N} D_1^{-1} = \begin{pmatrix}
        1 & 0 & 0 \\
        0 & 1 & 3 A^{\R} \overline{s_1} e^{-x\varphi_3}\\
        0 & 0 & 1
    \end{pmatrix},  \\    
    &\check{D} = D_1 D_6^{-1} = \begin{pmatrix}
        1 & 0 & 0\\
        0 & 3A^{\R} & 0\\
        0 & 0 & \frac{1}{3A^{\R}}
    \end{pmatrix}, \\
    &\check{H}_0 = D_6  \check{N}^{-1}   \check{Q}_{1\frac{2}{3}}^{(\infty)-1} \check{R}_6^{-1} D_6^{-1} = \begin{pmatrix}
        1 & 0 & 0\\
        0 & 1 & 0\\
        0 & -3 A^{\R} s_1 e^{x \varphi_3} & 1
    \end{pmatrix}, \\
    &\check{H}_1 = \check{L}_{6}^{-1} \check{Q}_{2\frac{2}{3}}^{(\infty) }  \check{M} = \begin{pmatrix}
        1 & 0 & 0\\
        0 & 1 & - \overline{s_1} e^{- x \varphi_3}\\
        0 & 0 & 1
    \end{pmatrix},   
\end{aligned}
\end{equation}
and  $s_1 = \omega^2 \frac{B}{A^{\R}}$. Note that $A^{\R}$ cannot be zero (see \eqref{identity 7 at parametrization of E_1} and \eqref{identity 6 at parametrization of E_1}), so $s_1$ is always well-defined.

Out of these observations, notice that we see that the local problem can be reduced to a $2\times 2$ setting because of the block structure of the jump matrices above. Moreover, near the stationary point $p_1$ (i.e., $\zeta = 1$) the exponent $\varphi_3(\zeta) - 2\sqrt{3} i$ has a second order zero:
\begin{align}\label{phi 3 near 1}
    \varphi_3(\zeta)=2\sqrt{3}i+\sqrt{3}i(\zeta-1)^2+ \O\left( (\zeta - 1)^3\right).
\end{align}
All these facts suggest the {\it parabolic cylinder function problem} as the candidate for the model RHP in the neighborhood of the point $p_1$. This intuitive idea can be implemented as follows.

Consider the following $2\times 2$ RHP for the function $\phi(z)$.

\begin{framed}
\begin{RHP}\label{RHP 3}
Find a matrix-valued function $\phi(z)$ satisfying the following conditions 
\begin{itemize}
    \item $\phi(z)$ is holomorphic in $\C \setminus \Gamma_4$ where $\Gamma_4$ is the jump contour defined in Figure \ref{defn of Psi1}.
    \item $\phi_+(z) = \phi_-(z) G_{\phi}$, where the jump matrix $G_{\phi}$ is given by \eqref{jump matrices for 2 by 2 problem}.
    \item $\phi (z) = (I + \O(1/z)) z^{\nu \sigma_3}$ as $z \to \infty$.
\end{itemize}
\end{RHP}
\end{framed}

\begin{equation}\label{jump matrices for 2 by 2 problem}
\begin{aligned}
    &H_2(z) =\begin{pmatrix}
        1 & 0\\
        s_1 e^{2\sqrt{3}ix-\frac{z^2}{2}} & 1
    \end{pmatrix} = \begin{pmatrix}
        1 & 0\\
        \widehat{s_1} e^{ - \frac{z^2}{2}} & 1
    \end{pmatrix} , \\
    &H_3(z)= \begin{pmatrix}
        1 & 3 A^{\R} \overline{s_1} e^{-2\sqrt{3}ix+\frac{z^2}{2}}\\
        0 & 1    
    \end{pmatrix} = \begin{pmatrix}
        1 & e^{- 2 \pi i \nu} \overline{\widehat{s_1}} e^{ \frac{z^2}{2}}\\
        0 & 1
    \end{pmatrix} , \\
    & D  = \begin{pmatrix}
         3A^{\R} & 0\\
         0 & \frac{1}{3A^{\R}}
     \end{pmatrix} = e^{- 2\pi i \nu \sigma_3}, \\
    & H_0(z)= \begin{pmatrix}
        1 & 0\\
        -3 A^{\R} s_1 e^{2\sqrt{3}ix-\frac{z^2}{2}} & 1
    \end{pmatrix} = \begin{pmatrix}
        1 & 0\\
        -e^{- 2 \pi i \nu} \widehat{s_1} e^{ - \frac{z^2}{2}} & 1
    \end{pmatrix}, \\
    &H_1(z)= \begin{pmatrix}
        1 & - \overline{s_1} e^{-2\sqrt{3}ix+\frac{z^2}{2}}\\
        0 & 1
    \end{pmatrix} = \begin{pmatrix}
        1 & - \overline{\widehat{s_1}} e^{ \frac{z^2}{2}}\\
        0 & 1
    \end{pmatrix}, 
\end{aligned} 
\end{equation}
where 
\begin{align*}
    &\widehat{s_1} = s_1 e^{2 \sqrt{3} ix}\\
    &\nu = -\frac{1}{2 \pi i} \ln 3A^{\R} = \frac{1}{2 \pi i } \ln (1 - |s_1|^2).
\end{align*}

\begin{figure}
    \centering
    \includegraphics[width=8cm]{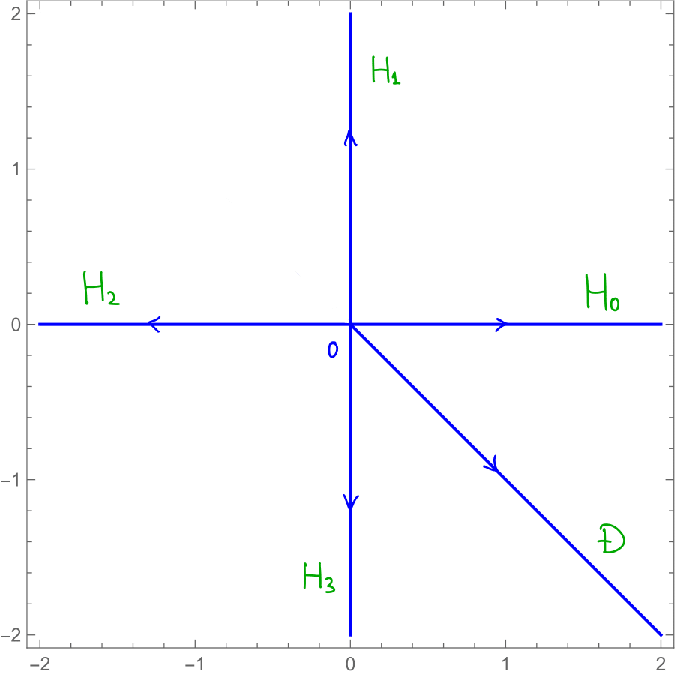}
    \caption{Contour $\Gamma_4$ and jump matrix $G_\phi$ of $\phi(z)$. }
    \label{defn of Psi1}
\end{figure}

One can solve this RHP using the parabolic cylinder function as demonstrated in section 1.5 in Chapter 2 of \cite{FIKN}. The asymptotic behavior of $\phi$ as $z \to \infty$ is
\begin{align}
    \phi(z) = \left( I + \frac{1}{z} \begin{pmatrix}
        0 &  -\alpha\\
        \nu/\alpha & 0\\
    \end{pmatrix} + \O(z^{-2}) \right) z^{\nu \sigma_3}, \label{parabolic cylinder asymptotic}
\end{align}
where
\begin{align*}
    \alpha &= - \frac{i}{ s_1 } \frac{\sqrt{2 \pi} e^{2\pi i \nu-2\sqrt3 ix }}{\Gamma(-\nu)}.
\end{align*}
Next, we introduce $Z(z)$ by
\begin{align}
    Z(z)=   \begin{pNiceMatrix}
        1 & 0 & 0 \\
        0 & \Block{2-2}{\phi(z)}\;\;\phantom{r} \\
        0 \\
    \end{pNiceMatrix},
\end{align}
so the asymptotic behavior of $Z(z)$   is
\begin{align}
   Z(z) = \left( I + \frac{1}{z} \begin{pNiceMatrix}
        0 & 0 & 0 \\
        0 & \Block{2-2}{m} \\
        0 \\
    \end{pNiceMatrix} + \O\left( z^{-2} \right) \right) \begin{pmatrix}
        1 & & \\
        & z^{\nu} & \\
        & & z^{-\nu}
    \end{pmatrix}, \quad \text{as } z\rightarrow \infty,
\end{align}
where we introduced 
\begin{align*}
    m = \begin{pmatrix}
        0 &  -\alpha\\
        \nu/\alpha & 0\\
    \end{pmatrix}.
\end{align*}
\begin{rem}
Observe that $s_1=0$ gives $B=0$, $A^\R=\frac13$, and $s^\R=0$. This implies $S_1^{(0, \infty)} = S_2^{(0, \infty)} = I$ and $E_1=\frac13 C$. Thus, the $Y$-problem of this case has no jump and Liouville's Theorem tells us that $Y(\zeta) = I$ entirely. Therefore, the corresponding solution to the radial Toda equation is $w_0(x)=i n \pi$, $n\in \Z$. With the reality condition, what we have here is the trivial solution and this agrees with the direct substitution of $s_1 = 0$ in \eqref{w_0(x) expression 1} below.
\end{rem}

The jump graph shown in Figure \ref{defn of Psi1}, if rotated by $\frac{3\pi}{4}$ clockwise, approximates the jump graph near the stationary point $p_1$ on Figure \ref{local picture of Phi1 check near 1}. Comparison of the quadratic terms in the exponentials in \eqref{ H matrix} and \eqref{jump matrices for 2 by 2 problem} suggests the relation 
\begin{align}
    x\varphi_3(\zeta)\sim2\sqrt{3}ix -\frac{z^2}{2} 
\end{align}
Let us define a local change of the variables by 
\begin{align}\label{change of var}
    z(\zeta) :&= \sqrt{2x} e^{i \frac{\pi}{2}} \sqrt{ \varphi_3(\zeta) - i 2 \sqrt{3}},
\end{align}
where its asymptotics is given by
\begin{align}
    z( \zeta ) \sim \sqrt{2x} e^{i \frac{3 \pi}{4}} 3^{1/4} (\zeta - 1) \label{z(zeta)}
\end{align}
using the branch determined to obtain \eqref{asymptotics of Y^D}.

Since $p_1$  is a stationary point, i.e. $\varphi_3(p_1)=0$,  equation \eqref{change of var} defines a holomorphic change of variables in some neighborhood $U_1$:
\begin{align}
    |\zeta-1|<r<1.
\end{align}
Observe that  $Z(\zeta):=Z(z(\zeta))$ has exactly the same jumps in $U_1$ as $\check{Y}_{p_1}$ does:
\begin{align}
    \check{H}_l(\zeta)= \begin{pNiceMatrix}
        1 & \phantom{f}0\phantom{f} & \phantom{f}0 \phantom{f} \\
        0 & \Block{2-2}{H_l(z(\zeta))}   \\
        0 \\
    \end{pNiceMatrix} , \quad \zeta\in U_1. 
\end{align}
Thus, applying the inverse transformation described in Figure \ref{Local phi cont}  to $Z(\zeta)$, we obtain $\check{Z}(\zeta)$, which has exactly the same jumps in $U_1$ as $\check{Y}(\zeta)$ does. Furthermore, $\check{Z}(\zeta)$ has the following asymptotics:  
\begin{align}
\begin{aligned}
\check{Z}(\zeta) = \left( I + \frac{1}{z(\zeta)} \begin{pNiceMatrix}
        0 & 0 & 0 \\
        0 & \Block{2-2}{m} \\
        0 \\
    \end{pNiceMatrix} + \O\left( z^{-2} \right) \right)
\begin{pmatrix}
1 &  &  \\
 & z(\zeta)^{\nu} & \\
 & & z(\zeta)^{-\nu}
\end{pmatrix}
\begin{dcases}
    I & \text{for $\zeta \in K^{(1)}$}\\
    D_1 & \text{for $\zeta \in K^{(2)}$}\\
    D_6 & \text{for $\zeta \in K^{(3)}$},
\end{dcases}
\end{aligned} \label{asymptotics of Phi^0}
\end{align}
as $x\rightarrow \infty$ and $0<r_0\leq |\zeta-1|\leq r_1<1$ (and hence $|z|\rightarrow \infty)$.

Finally, to get the matching with the asymptotics  \eqref{asymptotics of Y^D} of $\check{Y}^D(\zeta)$, we define the local parametrix $P^{(1)}(\zeta)$ near $\zeta=1$ by
\begin{align}
    P^{(1)}(\zeta) = V(\zeta) \check{Z}(\zeta), \label{defn of P^(1)}
\end{align}
 where 
\begin{align}
    V(\zeta) &= \check{Y}^D(\zeta) \begin{pmatrix}
        1 &  &  \\
         & z(\zeta)^{-\nu} & \\ 
         &  & z(\zeta)^{\nu}
    \end{pmatrix}
    \begin{dcases}
        I & \text{for $\zeta \in K^{(1)}$}\\
        D_1^{-1} & \text{for $\zeta \in K^{(2)}$}\\
        D_6^{-1} & \text{for $\zeta \in K^{(3)}$}
    \end{dcases} \label{defn of E(zeta)}
\end{align}
One can observe that  $V(\zeta)$ is holomorphic in $U_1$ and has the  following Taylor expansion: 
\begin{align*}
    V(\zeta) = \Theta (\zeta) \begin{pmatrix}
        1 & 0 & 0 \\
        0 & (\sqrt{2x} e^{i \frac{3 \pi}{4}} 3^{1/4} )^{-\nu} & 0\\
        0 & 0 & (\sqrt{2x} e^{i \frac{3 \pi}{4}} 3^{1/4})^{\nu}
    \end{pmatrix}
\end{align*}
by \eqref{asymptotics of Y^D} and \eqref{defn of E(zeta)}. Moreover, on the boundary of $U_1$, since  $\check{Z}(\zeta)$ has the asymptotic expansion \eqref{asymptotics of Phi^0}, we have
\begin{align}
\begin{aligned}
    &P^{(1)}(\zeta) = \\
    & \check{Y}^D(\zeta) \left(I + \frac{1}{z(\zeta)} 
    \left\{\!\begin{aligned}
        &I \\[1ex]
        &D_1^{-1} \\[1ex]
        &D_6^{-1}
    \end{aligned}\right\}
    \begin{pNiceMatrix}
        1 & 0 & 0 \\
        0 & \Block{2-2}{z^{-\nu \sigma_3}} \\
        0 \\
    \end{pNiceMatrix}
    \begin{pNiceMatrix}
        0 & 0 & 0 \\
        0 & \Block{2-2}{m} \\
        0 \\
    \end{pNiceMatrix}
    \begin{pNiceMatrix}
        1 & 0 & 0 \\
        0 & \Block{2-2}{z^{\nu \sigma_3}} \\
        0 \\
    \end{pNiceMatrix}
    \left\{\!\begin{aligned}
        &I \\[1ex]
        &D_1 \\[1ex]
        &D_6
    \end{aligned}\right\}
    + \O\left(z^{-2}\right) \right)\\
    &= \check{Y}^D(\zeta) (I + \O(1/\sqrt{x})). \label{matching up cond for parametrix near 1}
\end{aligned}
\end{align}
In other words, the matching condition \eqref{matching up cond appearing for the first time} is satisfied.

%%%%%%%%%%%%%%%%%%%%%%%%%%%%%%%%%

\subsubsection{Other Local Parametrices}\label{other param}

Similarly, for the other stationary points $\zeta = p_k$ where $k = 2,\cdots,6$, we are interested in finding a local parametrix  $P^{(p_k)}(\zeta)$, such that 
\begin{itemize}
    \item $P^{(p_k)}(\zeta)$ has the same jumps as the $\check{Y}$-problem in some neighborhood $U_{p_k}$ of $\zeta=p_k$,
    \item on the boundary $\partial U_{p_k}$, asymptotics of $P^{(p_k)}(\zeta)$ matches with the asymptotics  \eqref{asymptotics of Y^D} of $\check{Y}^D(\zeta).$ 
\end{itemize}

\begin{prop}
These local parametrices can be found using the following symmetries: 
\begin{align}
    &P^{(-1)}(\zeta) = d_3^{-1} \left[ P^{(1)}(-\zeta) \right]^{T-1} d_3. \label{local parametrix near -1}\\
    &P^{(-\overline{\omega})} (\zeta) = \Pi^{-1} P^{(-1)}(\omega \zeta) \Pi. \label{local parametrix near - omega bar}\\
    &P^{(\omega)} (\zeta) = \Pi P^{(1)}(\overline{\omega} \zeta) \Pi^{-1}. \label{local parametrix near omega}\\
    &P^{(\overline{\omega})}(\zeta) = d_3^{-1} \left[ P^{(-\overline{\omega})}(-\zeta) \right]^{T-1} d_3. \label{local parametrix near omega bar} \\
    &P^{(-\omega)} (\zeta) = d_3^{-1} \left[ P^{(\omega)}(-\zeta) \right]^{T-1} d_3. \label{local parametrix near minus omega}
\end{align}
\end{prop}

\begin{proof}
Let $G^{(p_k)}(\zeta)$ denote the jump matrices for the $P^{(p_k)}$-problem.

Observe that we have the following sequence  of equivalent statements 
\begin{align}
    &P^{(-1)}_+(\zeta) = P^{(-1)}_-(\zeta) G^{(-1)}(\zeta) \nonumber\\ &\Leftrightarrow d_3^{-1} \left[ P^{(1)}_+(-\zeta) \right]^{T-1} d_3 = d_3^{-1} \left[ P^{(1)}_-(-\zeta) \right]^{T-1} d_3 G^{(-1)}(\zeta) \nonumber\\
    &\Leftrightarrow \left[ P^{(1)}_+(-\zeta) \right]^{T-1} = \left[ P^{(1)}_-(-\zeta) \right]^{T-1} d_3 G^{(-1)}(\zeta) d_3^{-1} \nonumber\\
    &\Leftrightarrow \left[ G^{(1)}(-\zeta) \right]^{T-1} = d_3 G^{(-1)}(\zeta) d_3^{-1} \nonumber\\
    &\Leftrightarrow G^{(-1)}(\zeta) = d_3^{-1} \left[ G^{(1)}(-\zeta) \right]^{T-1} d_3.\label{G at -1}
\end{align}
Thus we need to check if \eqref{G at -1} is satisfied for all jump matrices of the RHP near   $\zeta = 1$  and the ones near $\zeta=-1$. Because of the anti-symmetry for $Q^{(\infty)}_n$ we   have that 
\begin{align*}
    Q_{2\frac23}^{(\infty)} = d_3^{-1} \left[ Q_{1\frac13}^{(\infty)} \right]^{T-1} d_3,
\end{align*}
and direct computation  shows that indeed
\begin{align*}
    &L_3^{-1} = d_3^{-1} \left[ L_6 \right]^{T} d_3, \, L_4 = d_3^{-1} \left[ L_1 \right]^{T-1} d_3,\\
     & R_3^{-1} = d_3^{-1} \left[ R_6 \right]^{T} d_3, \, R_4 = d_3^{-1} \left[ R_1 \right]^{T-1} d_3,\\
       & D_4 = D_1^{-1}, D_3=D_6^{-1}.
\end{align*}
Moreover, the matching condition with $\check{Y}^D$ is automatically satisfied near $\zeta=-1$, since by \eqref{matching up cond for parametrix near 1} 
\begin{align*}
    P^{(-1)}(\zeta) &= d_3^{-1} \left[\check{Y}^D (- \zeta)\right]^{T-1} d_3 (I + \O(1/\sqrt{x}))\\
    &=\left[ \check{Y}^D (- \zeta)\right]^{-1}(I + \O(1/\sqrt{x}))\\
   &=\check{Y}^D (\zeta) (I + \O(1/\sqrt{x}))
\end{align*}
So,  $P^{(-1)}(\zeta)$ satisfies all   local parametrix conditions. 

Similar reasoning can be applied to prove the other relations, \eqref{local parametrix near - omega bar} -- \eqref{local parametrix near minus omega}.
\end{proof}

%%%%%%%%%%%%%%%%%%%%%%%%%%%%

\subsection{Approximate Solution of the Original RHP as $x \to \infty$}\label{ Aprox solution section}

Let us call the union of all neighborhoods $U_{p_k}$ of the stationary points $\zeta=p_k$ from the previous subsection $U_{all}$:  
\begin{align*}
    U_{all} = U_1 \cup U_{- \overline{\omega}} \cup U_{\omega} \cup U_{- 1} \cup U_{\overline{\omega}} \cup U_{-\omega}.
\end{align*}
Then, we define a piecewise holomorphic function $\check{Y}^{(approx)} (\zeta)$ by
\begin{align}
    \check{Y}^{(approx)} (\zeta) = \begin{dcases}
    \check{Y}^D(\zeta) & \text{if } \zeta \in \C \setminus (S^1 \cup U_{all})\\
    P^{(1)}(\zeta) & \text{if } \zeta \in U_1\\
    P^{(- \overline{\omega})}(\zeta) & \text{if } \zeta \in U_{- \overline{\omega}} \\
    P^{(\omega)}(\zeta) & \text{if } \zeta \in U_{\omega} \\
    P^{(-1)}(\zeta) & \text{if } \zeta \in U_{-1} \\
    P^{(\overline{\omega})}(\zeta) & \text{if } \zeta \in U_{\overline{\omega}} \\
    P^{(- \omega)}(\zeta) & \text{if } \zeta \in U_{-\omega}
    \end{dcases}
\end{align}
which is discontinuous on $S^1 \setminus U_{all}$ and on $\partial U_{all}$.

%%%%%%%%%%%%%%%%%%%%%%%%%%%%%%%%%%%%%%%%%%%%%%

\subsubsection{Riemann-Hilbert Problem for the Error Function}
\label{error sec}

Recall that in subsection \ref{small norb sec} we discussed that we cannot apply the small norm theorem to the $\check{Y}$-problem directly. Instead, using   $\check{Y}^{(approx)}(\zeta)$ constructed in the previous subsections, we would be able to apply the small norm theorem to  the {\it error function} $R(\zeta)$,
\begin{align}
    R(\zeta) := \check{Y}(\zeta) \left( \check{Y}^{(approx)}(\zeta) \right)^{-1}.
\end{align}
Observe that the error function $R(\zeta)$ solves the following RHP.

\begin{framed}
\begin{RHP}\label{RHP 4}
Find a matrix-valued function $R( \zeta )$ satisfying the following conditions:
\begin{itemize}
    \item $R( \zeta )$ is holomorphic in $\C \setminus \Gamma_5$ where the contour  $\Gamma_5$ is shown in Figure \ref{contour of R problem}.
    \item $R_+(\zeta ) = R_-( \zeta ) G_{R}$, where the jump matrix $G_{R}$ are
    \begin{align*}
        G_R = \begin{dcases}
            \check{Y}^D(\zeta) G_{\check{Y}} \left( \check{Y}^{D}(\zeta) \right)^{-1} & \text{if $\zeta \in \Gamma_5 \setminus \partial U_{all}$}\\
            \check{Y}^D(\zeta) \left( P^{(p_k)}(\zeta) \right)^{-1} & \text{if $\zeta \in \partial U_{p_k}$}, \; k= 1, \cdots, 6.
        \end{dcases}
    \end{align*}
    \item $R(\zeta ) = I + \O(1 / \zeta )$ as $\zeta \to \infty$.
\end{itemize}
\end{RHP}
\end{framed}

\begin{figure}
    \centering
    \includegraphics[width=8cm]{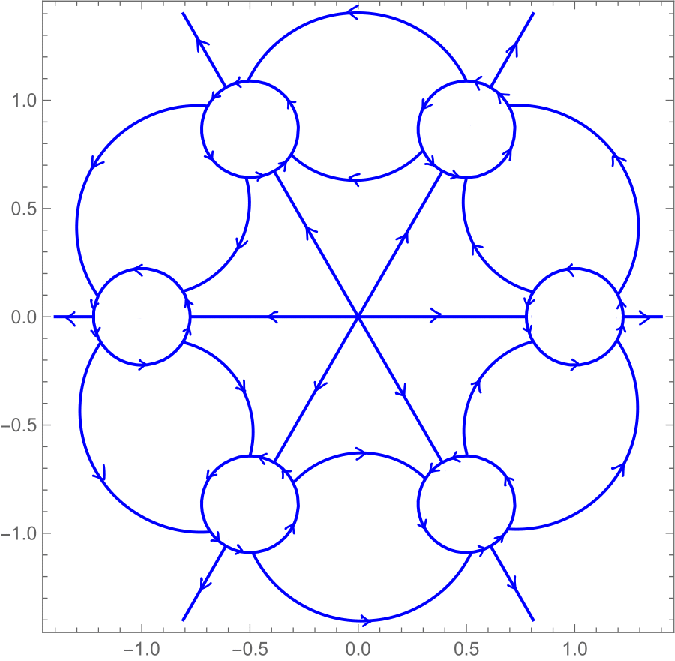}
    \caption{The jump contour $\Gamma_5$ of the $R$-problem.}
    \label{contour of R problem}
\end{figure}

Due to the structure of jump matrices of the $\check{Y}$-problem and the function $\check{Y}^{(approx)} (\zeta)$ 
%on $\Gamma_4$,
and the matching up conditions such as \eqref{matching up cond for parametrix near 1}, we have that 
\begin{align}
    || G_R - I ||_{L^2(\Gamma_{5}) \cap L^{\infty}(\Gamma_{5}) } \leq c_1 \frac{1}{\sqrt{x}}, \quad x \to \infty. \label{small norm ineq}
\end{align}
Thus, by the small norm theorem, there is a unique solution $R(\zeta)$  of the RHP \ref{RHP 4} for $x$ large. Moreover, one can show the following fact:
\begin{thm} Given any point of $\mathcal{M} = \{(s^\R,y^\R) \}$ defined by \eqref{monodormy data} and for every $x \in (0, \infty)$, all the above introduced Rieman-Hilbert Problems are solvable.
\end{thm}

\begin{proof}
Let $(s^{\R}, y^{\R}) \in \mathcal{M}$ be arbitrary. Then, for sufficiently large $x \in \R_{>0}$, say $x > x_0$, there exists a unique solution $R(\zeta)$ by the small norm theorem. Thus, all the other RH problems ($\check{Y}$, $Y$, $\Tilde{\Phi}$, $\Phi$, $\hat{\Psi}$ problems) are solvable for such $x$. It means that there exists a solution $w_0(x)$ of the radial Toda equation \eqref{negative tt*-Toda with x when n=2} for $x > x_0$ whose monodromy data $(s^{\R}, y^{\R})$ is given.
Since a solution $w_0(x)$ of \eqref{negative tt*-Toda with x when n=2} can be smoothly extended to the whole half line $x>0$ (see Proposition \ref{smoothness prop}), the Riemann-Hilbert Problem for $\hat{\Psi}$ will be solvable for all $x>0$ by the canonical solutions of the system \eqref{first eq} corresponding to the extended $w_0(x)$. Simultaneously, all the other RH problems will be solvable for all $x>0$.
\end{proof}

%%%%%%%%%%%%%%%%%%%%%%%%%%%%%%%%%%%%%

\subsubsection{Asymptotics Analysis of $w_0(x)$ as $x \to \infty$}\label{as at inf}

In this subsection, we will find the desired asymptotics of $w_0(x)$ as $x \to \infty$. First, we note that   by   definition  
\begin{align}
    \lim_{\zeta\rightarrow 0}   \check{Y} (\zeta) =\lim_{\zeta\rightarrow 0}   R(\zeta) \check{Y}^{D} (\zeta). \label{Y(0) = R(0) Y^D(0)}
\end{align}
From equation \eqref{global parametrix}, one can see that $\check{Y}^D(0) = I$. Thus, by the definition of $\check{Y}$ and \eqref{Y at 0}, we have that
\begin{equation}\begin{aligned}\label{R(0) asymptotic equation 3}
    &\lim_{\zeta\rightarrow 0}   {R}(\zeta)  =\lim_{\zeta\rightarrow 0}   \check{Y}(\zeta)= \Omega e^{-2w} \Omega^{-1}  \\
    &= \frac{1}{3} \begin{pmatrix}
        1 + e^{2 w_0} + e^{-2 w_0} & \omega^2 + \omega e^{2 w_0} + e^{-2 w_0} & \omega + \omega^2 e^{2 w_0} + e^{-2 w_0} \\
        \omega + \omega^2 e^{2 w_0} + e^{-2 w_0} & 1 + e^{2 w_0} + e^{-2 w_0} & \omega^2 + \omega e^{2 w_0} + e^{-2 w_0}\\
        \omega^2 + \omega e^{2 w_0} + e^{-2 w_0} & \omega + \omega^2 e^{2 w_0} + e^{-2 w_0} & 1 + e^{2 w_0} + e^{-2 w_0}
    \end{pmatrix}.
\end{aligned}\end{equation}

On the other hand, we can write $R(\zeta)$ in a form of Cauchy integral,
\begin{align}
    R(\zeta) = I + \frac{1}{2 \pi i } \int_{\Gamma_{5}} \frac{\rho(\zeta') (G_R (\zeta') - I)}{\zeta' - \zeta} d\zeta', \; \zeta \notin \Gamma_{5}, \label{singular integral equation}
\end{align}
where $\displaystyle{\rho(\zeta') = R_{-}(\zeta') = \lim_{\text{\{$-$ side of $\Gamma_5$\}} \ni \lambda \to \zeta'} R(\lambda)}$ for $\zeta' \in \Gamma_5.$
Observe that inequality \eqref{small norm ineq} allows us to apply the small norm theorem, which in turn implies
\begin{align}
    ||\rho - I||_{L^2(\Gamma_{5})} \leq c_2 \frac{1}{\sqrt{x}}. \label{rho - I has a small norm}
\end{align}

Taking the limit $\zeta \to 0$ of \eqref{singular integral equation} gives
\begin{align}
\begin{aligned}
    \lim_{\zeta\rightarrow 0} R(\zeta) &= I + \frac{1}{2\pi i} \int_{\Gamma_{5}} \frac{\rho(\zeta')(G_{R}(\zeta') - I)}{\zeta'} d\zeta'\\
    &= I + \frac{1}{2\pi i} \int_{\Gamma_{5}} \frac{G_{R}(\zeta') - I}{\zeta'} d\zeta' + \frac{1}{2\pi i} \int_{\Gamma_{5}} \frac{(\rho(\zeta') - I)(G_{R}(\zeta') - I)}{\zeta'} d\zeta'\\
    &= I + \frac{1}{2\pi i} \int_{\Gamma_{5}} \frac{G_{R}(\zeta') - I}{\zeta'} d\zeta' + \O \left(\frac{1}{x} \right),
\end{aligned} \label{R(0) asymptotic equation 1}
\end{align}
where for the last equality we used the following estimate
\begin{align*}
      \left| \int_{\Gamma_{5}} \frac{(\rho(\zeta') - I)(G_{R}(\zeta') - I)}{\zeta'} d\zeta' \right| &\leq \frac{1}{\text{dist}(\Gamma_5, 0)} || \rho - I ||_{L^2(\Gamma_5)} || G_R - I ||_{L^2(\Gamma_5)}\\
      &\leq c_3 \frac{1}{x}
      %\label{G(.)-I/. has a small norm}
\end{align*}
due to the Cauchy Schwarz inequality, \eqref{small norm ineq}, and \eqref{rho - I has a small norm}.

For convenience, let us  denote the boundary of $U_{p_k}$'s as follows:
\begin{align*}
    \gamma_1 = \partial U_1, \;
    \gamma_2 = \partial U_{- \overline{\omega}},\;
    \gamma_3 = \partial U_{\omega}, \;
    \gamma_4 = \partial U_{- 1}, \;
    \gamma_5 = \partial U_{\overline{\omega}},\;
    \gamma_6 = \partial U_{-\omega}.
\end{align*}
Let $\widetilde{\Gamma_5} := \Gamma_5 \setminus \cup_{k = 1}^6 \gamma_k$. Since jump matrices on $\widetilde{\Gamma_5}$ are exponentially close to the identity matrix, we can rewrite the last equation \eqref{R(0) asymptotic equation 1}
as
\begin{align}
    \lim_{\zeta\rightarrow 0} R(\zeta) &= I + \sum_{k = 1}^{6} \frac{1}{2 \pi i} \int_{\gamma_k} \frac{G_R (\zeta') - I}{\zeta'} d \zeta' + \O \left( \frac{1}{x} \right). \label{R(0) asymptotic equation 2}
\end{align}

\begin{prop}\label{Residue prop}
    One has 
    \begin{align}
    \lim_{\zeta\rightarrow 0} R(\zeta) &= \begin{pmatrix}
        1 & \widehat{\alpha} - \omega \widehat{\beta} & \widehat{\beta} - \omega^2 \widehat{\alpha}\\
        \widehat{\beta} - \omega^2 \widehat{\alpha} & 1 & \widehat{\alpha} - \omega \widehat{\beta}\\
        \widehat{\alpha} - \omega \widehat{\beta} & \widehat{\beta} - \omega^2 \widehat{\alpha} & 1
    \end{pmatrix} + \O\left( \frac{1}{x} \right), \label{R(0) asymptotic equation 4}
\end{align}
where 
\begin{align*}
    &\widehat{\alpha} = \frac{\alpha}{\sqrt{2x}} (24 \sqrt{3} x)^{- \nu} 3^{-1/4} e^{-\frac{3 \pi i}{4} - \frac{3 \pi i}{2} \nu},\\
    &\widehat{\beta} = -\frac{1}{\sqrt{2x}} \frac{\nu}{\alpha} (24 \sqrt{3} x)^{\nu} 3^{-1/4} e^{-\frac{3 \pi i}{4} + \frac{3 \pi i}{2} \nu},\\
    &\alpha = - \frac{i}{s_1 e^{i 2\sqrt{3} x}} \frac{\sqrt{2 \pi} e^{2\pi i \nu}}{\Gamma(-\nu)}, \; \nu  = \frac{1}{2 \pi i } \ln (1 - |s_1|^2), \; s_1 = \omega^2 \frac{B}{A^{\R}}.
\end{align*}
\end{prop}
The proof of this proposition is given in Appendix \ref{proof of res}.

Comparing \eqref{R(0) asymptotic equation 3} with \eqref{R(0) asymptotic equation 4}, we  have 
\begin{align}
    e^{-2w_0}-1=(1-\omega^2)(\widehat{\alpha} - \omega \widehat{\beta}) +\O \left( \frac{1}{x}\right),
\end{align}
where the leading term, as it is shown below, is of order $\O(x^{-1/2})$, thus the real solution should have the following asymptotics: 
\begin{align}
w_0 (x) = \frac{\sqrt{3}}{2} e^{-\frac{5 \pi i}{6}} (\widehat{\alpha} - \omega \widehat{\beta})  + \O \left( \frac{1}{x} \right),
\end{align}
  substituting $\widehat{\alpha}$ and $\widehat{\beta}$ we obtain the following theorem. 

\begin{thm} \label{result 1}
For every $s^{\R} \in (-3,1)$ and every $y^{\R} \in \R$, the asymptotics of the corresponding solution $w_0(x)$ of the radial Toda equation \eqref{negative tt*-Toda with x when n=2} is given by the formula
\begin{align}
    w_0(x) = \frac{\sigma}{\sqrt{x}} \cos \left( 2 \sqrt{3} x - i \nu \ln x + \psi \right) + \O \left( \frac{1}{x} \right) \label{w_0(x) expression 1}, \quad x\rightarrow \infty 
\end{align}
where
\begin{align*}
    \sigma^2 &= - \frac{\sqrt{3}}{4 \pi} \ln (1 - |s_1|^2), \; \sigma > 0,\\
    \psi &= - i \nu \ln (24\sqrt{3}) + \frac{\pi}{12} + \arg s_1 - \arg \Gamma(\nu),\\
    \nu & = \frac{1}{2 \pi i } \ln (1 - |s_1|^2), \\
    s_1 &= \omega^2 \frac{B}{A^{\R}}.
\end{align*}

\end{thm}
The proof of this theorem is given in Appendix \ref{proof of main}.

\begin{cor}
For every $s^{\R} \in (-3,1)$ and every $y^{\R} \in \R$, the asymptotics of the corresponding solution of the Painlev\'e III ($D_7$) equation \eqref{Painleve III} is given by the formula
\begin{align*}
        \Tilde{w}(s) = s^{1/3} -\frac{4}{\sqrt{3}} \sigma \cos \left( \frac{3 \sqrt{3}}{2} s^{2/3} - i \nu \ln \left( \frac{3}{4} s^{2/3} \right) + \psi \right) + \O \left( \frac{1}{s^{1/3}} \right), \quad s \to \infty
\end{align*}
where $\sigma$ and $\psi$ are the ones we introduced in Theorem \ref{result 1}.
\end{cor}

%% file: section4.tex
\section{Asymptotics near $x = 0$ and Connection Formulae} \label{section 4}

\subsection{Asymptotic Data of Local Solutions near $x = 0$}\label{section 4.1}

It was pointed out in \cite{GuLi14, GIL3} that the method of loop group Iwasawa factorization produces a local solution of the 2D periodic Toda equations near $z_0 \in \C$ from any matrix of the form
\begin{align*}
\begin{pmatrix}
    0 & 0 & p_0\\
    p_1 & 0 & 0\\
    0 & p_2 & 0\\
\end{pmatrix}
\end{align*}
where each function $p_i = p_i(z)$ is holomorphic in a neighborhood of $z_0$. To obtain radial solutions near zero, we need to take $p_i(z)=c_i z^{k_i}$ for certain real $k_i$ and $c_i$.
Indeed, if we take $p_i(z) = c_i z^{k_i}$ with $c_i > 0$ and $k_i>-1$,
then radial solutions near zero satisfy
\begin{align}
    w_0(x) = \gamma \ln x + \rho + o(1) \label{asymptotics of w_0 near 0}
\end{align}
where $-\frac{1}{2} < \gamma < 1$. All such solutions with this asymptotic behavior arise this way. As we have pointed out earlier, all of these solutions are in fact global solutions.
We call the two real parameters $\{\gamma,\rho\}$ asymptotic data and the two real parameters $\{c_i, k_i\}$ holomorphic data. The relation between these two data will be given below.

As we obtained the asymptotic formula for the solution near $x = \infty$ as in Theorem \ref{result 1}, our goal is now to find connection formulae between two kinds of asymptotic parameters. This will be done by making use of the monodromy data $\{s^\R, y^\R\}$, and a method based on the Iwasawa factorization.

%%%%%%%%%%%%%%%%%%%%%%%%%%%%%%%%%%%%%%%%%%%%%%%%%%%%%%%%%%%%

\subsubsection{Iwasawa Factorization}

The following discussion is parallel to the one in \cite{GIL3}, so we omit many details.
Consider the (possibly multi-valued) holomorphic matrix function $L(z, \lambda)$ that satisfies the following ODE,
\begin{align}
    \frac{d L}{d z} = \frac{1}{\lambda} L \eta, \quad L \big |_{z = 0} = I, \label{iwasawa ODE}
\end{align}
where $\lambda \in \C^*$ and 
\begin{equation}\begin{aligned}\label{eta def}
    &\eta = \begin{pmatrix}
        0 & 0 & p_0 \\
        p_1 & 0 & 0\\
        0 & p_1 & 0
    \end{pmatrix}, \\
  &p_i = c_i z^{k_i}, \; c_i > 0, \; k_i > - 1.
\end{aligned}\end{equation}

\begin{rem}
For convenience, we omit arguments of functions when no confusion is likely.
\end{rem}
 
Consider the {\it twisted loop group}   $\Lambda SL_3^t \C$:
\begin{align*}
    \Lambda SL_3^t \C &= \left\{ M(\lambda) = \sum_{j = - \infty}^{\infty} \lambda^{-j} M_j \, \Bigg| \, \det M = 1, \, \Delta M^{T-1}(-\lambda) \Delta = M(\lambda) \right\}
\end{align*}
where 
\begin{align*}
    \Delta = \begin{pmatrix}
        0 & 0 & 1\\
        0 & 1 & 0\\
        1 & 0 & 0
    \end{pmatrix}.
\end{align*}
From \eqref{eta def} it follows that $L \in \Lambda SL_3^t \C$. 

\begin{thm}
Let $L(z, \lambda)$ be as in \eqref{iwasawa ODE}. Then, there is a unique Iwasawa factorization of $L(z, \lambda)$ :
\begin{align*}
    L(z, \lambda) = L_{\R} (z, \overline{z}, \lambda) L_+ (z, \overline{z}, \lambda)
\end{align*}
where
\begin{itemize}
    \item   $\left[ L_{\R}^* \left( 1/ \overline{\lambda} \right) \right]^{-1} = L_{\R} (\lambda)$,
    \item   $L_{+} = \sum_{j = 0}^{\infty} L_j (z, \overline{z}) \lambda^j$,  with $ L_0 = diag (b_0, 1, b_0^{-1}), \; b_0 > 0$, where $b_0$ is smooth in a small neighbourhood of $0$ and depends on $|z|$.  
    \item The maps $L_{\R}, \; L_+$   are (possibly multi-valued) maps from a small neighborhood of $0$ into $\Lambda SL_3^t \C$,
    \item  $L_{\R} \big|_{z = 0} = L_{+} \big|_{z = 0} = I$.
\end{itemize}
\end{thm}

\begin{proof}
This is a special case of Theorem 8.1.1 of \cite{PressleySegal}. See also \cite{McIntosh}.
\end{proof}

Moreover, we have that 
\begin{align}
    \begin{dcases}
        L_{\R}^{-1} \frac{d L_{\R}}{d z} = a + \frac{1}{\lambda} A^T\\
        L_{\R}^{-1} \frac{d L_{\R}}{d \overline{z}} = - \overline{a} - \lambda \overline{A}
    \end{dcases} \label{Iwasawa linear system I}
\end{align}
where
\begin{align*}
    & a = \diag \left( \frac{d}{dz} \ln b_0, 0, - \frac{d}{dz} \ln b_0 \right), \\
    & A^T = \begin{pmatrix}
        0 & 0 & A_0\\
        A_1 & 0 & 0\\
        0 & A_1 & 0
    \end{pmatrix}, \\
    &A_0 = p_0 b_0^2, \; A_1 = p_1 b_0^{-1}.
\end{align*}
Introduce
\begin{align*}
    G(z, \overline{z}) := \begin{pmatrix}
        \frac{|h_0|}{h_0} & 0 & 0 \\
        0 & 1 & 0\\
        0 & 0 & \frac{h_0}{|h_0|}
    \end{pmatrix}, \quad h_0 = p_1^{1/3} p_0^{-1/3}.
\end{align*}
By setting
\begin{align*}
    w_0 = \ln \frac{b_0}{|h_0|} \quad \text{and} \quad \nu = p_0^{1/3} p_1^{2/3}=(c_0 c_1^2z^{k_0+2k_1})^{\frac13},
\end{align*}
one can verify that 
\begin{align}
\begin{dcases}
    \frac{d}{dz} (L_{\R} G) = L_{\R} G \left( w_z + \frac{1}{\lambda} \nu W^{T} \right)\\
    \frac{d}{d \overline{z}} (L_{\R} G) = L_{\R} G \left( -w_{\overline{z}} - \lambda \overline{\nu} W \right)
\end{dcases} \label{Iwasawa linear system II}
\end{align}
where $w$, $W$ are as in \eqref{w and W}.

Next, we introduce a new variable $t$ by
\begin{align*}
    t = (c_0 c_1^2)^{1/3} z^{N/3} \frac{3}{N}, \quad N = k_0 + 2 k_1 + 3.
\end{align*}
Note that $\frac{dt}{dz} = \nu$. Then, define 
\begin{align}
  \Tilde{\Psi} (t, \overline{t}, \lambda) := (L_{\R} G)^T(z(t,\bar{t}),\bar{z}(t,\bar{t}),\lambda)  .
\end{align}
After this transformation, equation \eqref{Iwasawa linear system II}  becomes 
\begin{align}
    \begin{dcases}
        \frac{d \Tilde{\Psi}}{d t} = \left( w_t + \frac{1}{\lambda} W \right) \Tilde{\Psi}\\
        \frac{d \Tilde{\Psi}}{d \overline{t}} = \left( - w_{\overline{t}} - \lambda W^T \right) \Tilde{\Psi},
    \end{dcases} \label{Iwasawa Lax I}
\end{align}
which is exactly the Lax pair \eqref{Lax pair} for the 2D periodic Toda equation of type $A_2$. This change of variables gives
\begin{align*}
    w_0(t) = - \frac{k_1 - k_0}{N} \ln |t| - \frac{k_1 - k_0}{N} \ln \left( \frac{N}{3}(c_0 c_1)^{-1/3} \right) - \frac{1}{3} \ln \left( \frac{c_1}{c_0} \right) + o(1).
\end{align*}
Comparing it with \eqref{asymptotics of w_0 near 0}, we have
\begin{align*}
    \gamma = \frac{k_0 - k_1}{N}, \quad \rho = \ln \left[ \left( \frac{N}{3} \right)^{\gamma} c_0^{\frac{-\gamma + 1}{3}} c_1^{\frac{-2 \gamma - 1}{3}} \right].
\end{align*}

Let
\begin{align*} 
    g(\lambda) = \lambda^m, \quad m = \diag(- \gamma, 0, \gamma ), \quad  \gamma = \frac{k_0 - k_1}{N}.
\end{align*}
As in \cite{GIL3}, the radial symmetry of $w_0$ implies that 
$\Tilde{\Psi} (t, \overline{t}, \lambda) g(\lambda)^T$ depends only on $|t|$ and $\lambda/t$. We then write
\begin{align*}
    \Tilde{\Psi} (t, \overline{t}, \lambda) g(\lambda)^T =: \Psi(x, \zeta), \quad x = |t|, \quad \zeta = \frac{\lambda}{t},
\end{align*}
and this transforms \eqref{Iwasawa Lax I} into 
\begin{align}
    \begin{dcases}
            \frac{d \Psi}{d \zeta} = \left(- \frac{1}{\zeta^2} W -\frac{x}{\zeta} w_x - x^2 W^{T} \right) \Psi\\
            \frac{d \Psi}{d x} = \left( -w_x - 2x\zeta W^T \right) \Psi.
    \end{dcases} \label{Iwasawa Lax II}
\end{align}

Simultaneously, let us introduce
\begin{align*}
    \Phi(\lambda) := (g L)^T
\end{align*}
Then, it transforms equation \eqref{iwasawa ODE} into
\begin{align}
        \frac{d \Phi}{d \lambda} &= \left[ -\frac{3}{N} \frac{z}{\lambda^2} \eta^T + \frac{1}{\lambda} m \right] \Phi. \label{Iwasawa Phi}
\end{align}
It will be convenient to introduce
\begin{align*}
    \mathring{\Phi}(\zeta) := t^{-m} \Phi(\lambda) \big|_{\lambda = \zeta t},
\end{align*}
and from \eqref{Iwasawa Phi}, using similar arguments to those in \cite{GIL3}, one can obtain
\begin{align}
    \frac{d \mathring{\Phi}}{d \zeta} = \left[ -\frac{1}{\zeta^2} \begin{pmatrix}
        0 & \kappa & 0\\
        0 & 0 & \kappa\\
        \kappa^{-2} & 0 & 0
    \end{pmatrix} + \frac{1}{\zeta} m \right] \mathring{\Phi}, \label{Iwasawa Phi0}
\end{align}
where $\kappa = \left( \frac{3}{N} \right)^{\gamma} c_1^{\frac{2 \gamma + 1}{3}} c_0^{\frac{\gamma - 1}{3}} = e^{- \rho}$.

\begin{note}
    The above $\Phi$ is not related to the $\Phi$ in section \ref{Phi RH problem}.
\end{note}

Summarizing, we have obtained the Lax pair \eqref{Iwasawa Lax II} and the subsidiary ODE \eqref{Iwasawa Phi0} from the Iwasawa factorization and the connection between the holomorphic data $\{ c_i, k_i \}$ and the asymptotic data $\{\gamma, \rho\}$:
\begin{align}
    \gamma = \frac{k_0 - k_1}{N}, \quad \kappa = \left( \frac{3}{N} \right)^{\gamma} c_1^{\frac{2 \gamma + 1}{3}} c_0^{\frac{\gamma - 1}{3}} = e^{- \rho}. \label{Iwasawa data and asymptotic data}
\end{align}

%%%%%%%%%%%%%%%%%%%%%%%%%%%%%%%%%%%%%%%%%%%%%%%

\subsubsection{Monodromy Data for $\mathring{\Phi}$}

We shall see that $\Psi$ and $\mathring{\Phi}$ have very similar behavior at $0$, in particular the same Stokes data there. 
Moreover, as we will see, the Stokes data of \eqref{Iwasawa Phi0} can be calculated explicitly because of the very specific structure of \eqref{Iwasawa Phi0}.
Recall from section \ref{monodromy data} that we have
\begin{itemize}
    \item Canonical solutions near singularities:
    \begin{align*}
        &\Psi_k^{(0)} = P_{0} \left( I + \O(\zeta) \right) e^{\frac{1}{\zeta} d_3}, \; \zeta \to 0, \; \zeta \in \Omega_{k}^{(0)}, \; k = 1,2,3.\\
        &\Psi_k^{(\infty)} = P_{\infty} \left( I + \O(1/\zeta) \right) e^{x^2 \zeta d_3}, \; \zeta \to \infty, \; \zeta \in \Omega_{k}^{(\infty)}, \; k = 1,2,3.
    \end{align*}
    \item Stokes matrices:
    \begin{align*}
        S_k^{(0, \infty)} = [\Psi^{(0, \infty)}_k(\zeta)]^{-1} \Psi^{(0, \infty)}_{k+1}(\zeta), \; k = 1, 2.
    \end{align*}
    \item Connection matrix:
    \begin{align*}
        E_1 = [\Psi_1^{(0)}(\zeta)]^{-1}\Psi^{(\infty)}_{1}(\zeta).
    \end{align*}    
\end{itemize}

For $\mathring{\Phi}$, one can derive similar formulae (see \cite{GIL3}):
\begin{itemize}
    \item Canonical solution near singularities:
    \begin{align*}
        &\mathring{\Phi}_k^{(0)} = \mathring{O}_{0} \left( I + \O(\zeta) \right) e^{\frac{1}{\zeta} d_3}, \; \zeta \to 0, \; \zeta \in \Omega_{k}^{(0)}, \; k = 1,2,3,\\
        &\qquad \quad \text{where } \, \mathring{O}_{0} = \begin{pmatrix}
            \kappa & 0 & 0\\
            0 & 1 & 0\\
            0 & 0 & \kappa^{-1}
        \end{pmatrix} \Omega, \quad \Omega = \begin{pmatrix}
            1 & 1 & 1 \\
            1 & \omega & \omega^2\\
            1 & \omega^2 & \omega
        \end{pmatrix}.\\
        & \mathring{\Phi}^{(\infty)} = \left( I + \O(1/\zeta) \right) \zeta^m, \; \zeta \to \infty.
    \end{align*}
    \item Stokes matrices:
    \begin{align*}
        R_k = [\mathring{\Phi}^{(0)}_k(\zeta)]^{-1} \mathring{\Phi}^{(0)}_{k+1}(\zeta), \; k = 1, 2.
    \end{align*}
    \item Connection matrices:
    \begin{align*}
        D_k = [\mathring{\Phi}_k^{(0)}(\zeta)]^{-1}\mathring{\Phi}^{(\infty)}(\zeta).
    \end{align*}
\end{itemize}

%%%%%%%%%%%%%%%%%%%%%%%%%%%%%%%%%%%%%%%%%%

\subsubsection{Connection Formulae}\label{connection formulae sec}

\begin{lem} \label{E_1 = D_1 D_1^* C}
$S_k^{(0)} = R_k$ and $E_1 = D_1 D_1^* C$.
\end{lem}

\begin{proof}
    The proof for the Stokes matrices is exactly the same as Corollary 4.3 in \cite{GIL3}. We will give a proof for the connection matrices.
    
    Since   $(g L_{\R} G)^T$  and $\Psi_k^{(0, \infty)}$ are   solutions of the first equation of \eqref{Iwasawa Lax II}, they differ by constant (in both $\zeta$ and $x$) matrices $Y_k$, $X_k$ respectively:
    \begin{align}
        \Psi_k^{(0)} = (g L_{\R} G)^T \Bigr|_{\lambda = \zeta t} Y_k, \quad \Psi_k^{(\infty)} = (g L_{\R} G)^T \Bigr|_{\lambda = \zeta t} X_k. \label{definition of Y_k and Z_k}
    \end{align}
    Notice that
    \begin{align*}
        E_1 = [\Psi_1^{(0)}(\zeta)]^{-1}\Psi^{(\infty)}_{1}(\zeta) = Y_1^{-1} \left[ (g L_{\R} G)^T \right]^{-1} (g L_{\R} G)^T X_1 = Y_1^{-1} X_1.
    \end{align*}
    By the inversion symmetry relation \eqref{antisymmetry for canonical solution at 0 ver.1} and reality condition \eqref{reality of canonical solutions at infty},
    \begin{align}
        \Delta \overline{\Psi^{(\infty)}_{2} ( \overline{\zeta} )} 3 d_3 = \overline{\Psi^{(0)}_1 \left( \frac{e^{i \pi}}{x^2 \overline{\zeta}} \right)} \Rightarrow \Delta \Psi^{(\infty)}_{2} ( e^{2 \pi i} \zeta ) C 3 d_3 = \overline{\Psi^{(0)}_1 \left( \frac{e^{i \pi}}{x^2 \overline{\zeta}} \right)}. \label{expression I for E_1 decomposition}
    \end{align}
    Using the loop group reality $\left[ L_{\R}^* \left( 1/ \overline{\lambda} \right) \right]^{-1} = L_{\R} (\lambda)$ for \eqref{definition of Y_k and Z_k}, we have
    \begin{align*}
        \overline{\Psi_1^{(0)} \left( \frac{e^{i \pi}}{x^2 \overline{\zeta}} \right)} = \overline{G} L_{\R} (e^{i \pi} \zeta t)^{-1} \overline{g \left( \frac{e^{i \pi}}{\overline{t \zeta} } \right)} \overline{Y_1}.
    \end{align*}
    By \eqref{definition of Y_k and Z_k}, it also holds that
    \begin{align*}
        L_{\R} (e^{i \pi} \zeta t)^{-1} = G \left[ \Psi_1^{(\infty)} (e^{i \pi} \zeta) \right]^{T-1} X_1^T g(e^{i \pi} \zeta t).
    \end{align*}
    Since $\overline{G} G = I$ and $g(e^{i \pi} \zeta t) \overline{g \left( \frac{e^{i \pi}}{\overline{t \zeta} } \right)} = I$,
    it   follows that
    \begin{align}
        \overline{\Psi_1^{(0)} \left( \frac{e^{i \pi}}{x^2 \overline{\zeta}} \right)} = \left[ \Psi_1^{(\infty)} (e^{i \pi} \zeta) \right]^{T-1} X_1^T \overline{Y_1} = \Delta \Psi^{(\infty)}_2 (e^{i 2 \pi}\zeta) 3d_3^{-1} X_1^T \overline{Y_1}, \label{expression II for E_1 decomposition}
    \end{align}
    where we used \eqref{antisymmetry for canonical solution at 0 ver.2} for the last equality. By \eqref{expression I for E_1 decomposition} and \eqref{expression II for E_1 decomposition}, we obtain $E_1 = D_1 D_1^* C$.
\end{proof}

\begin{rem}
Recall the parametrization \eqref{parametrization of E_1 ver.2} of $E_1$. Since $E_1 C = D_1 D_1^{*}$, $E_1C$ is positive definite. Therefore, the upper left $1 \times 1$ corner of $E_1 C$ has a positive determinant, i.e., $A^{\R}$ is positive in this situation, i.e., when the solution arises through the Iwasawa factorization.
\end{rem}

It turns out that equation \eqref{Iwasawa Phi0} can be reduced to $3$rd order scalar differential equations which can be solved by Barnes-type integrals involving Gamma functions. This gives rise to an explicit formula for the connection matrix $D_1$. As this is very similar to the discussion leading to Theorem 3.13 in \cite{GIL3}, we just state the result.  

\begin{lem} The connection matrix $D_1$ admits the decomposition,
\begin{align}
    D_1 = R_1 K^{-1} \Lambda^{-1}, \label{decomposition of D_1}
\end{align}
where
\begin{align*}
    &R_1 = S_1^{(0)} = \begin{pmatrix}
        1 & 0 & - \omega^2 s^{\R}\\
        - \omega^2 s^{\R} & 1 & \omega (s^{\R})^2 + \omega s^{\R}\\
        0 & 0 & 1
    \end{pmatrix}, \;
    K = \begin{pmatrix}
        1 & \omega^{\gamma} & \omega^{-\gamma}\\
        1 & \omega & \omega^{-1} \\
        1 & \omega^{2 - \gamma} & \omega^{\gamma - 2}
    \end{pmatrix},\\
    &\Lambda = - \frac{i \sqrt{3}}{4 \pi^2} \begin{pmatrix}
        \kappa \mathcal{A}_0 & & \\
         & (1 - \gamma) \mathcal{B}_0 & \\
         & & \frac{2(\gamma - 1)^2}{\kappa} \mathcal{C}_0
    \end{pmatrix}, \\
    &\mathcal{A}_0 = 2 \pi i e^{- \pi i \gamma} 3^{-\gamma}  \Gamma \left( \frac{2 - 2 \gamma }{3} \right) \Gamma \left( \frac{1 - \gamma}{3} \right),\\
    &\mathcal{B}_0 = - \frac{2 \pi i}{3}  \Gamma \left( \frac{1 - \gamma }{3} \right) \Gamma \left( \frac{\gamma - 1}{3} \right), \\
    &\mathcal{C}_0 = 2 \pi i e^{\pi i \gamma} 3^{-2 + \gamma}  \Gamma \left( \frac{\gamma - 1}{3} \right) \Gamma \left( \frac{2\gamma - 2}{3} \right).
\end{align*}
\end{lem}

\begin{thm} Let us introduce a real parameter $q^\R$ by
\begin{align*}
    q^\R := \frac{2 (\gamma - 1)^2 }{\kappa^2} 3^{2 (\gamma - 1)} \frac{\Gamma \left( \frac{\gamma - 1}{3} \right) \Gamma \left( \frac{2 \gamma - 2}{3} \right) }{ \Gamma \left( \frac{2 - 2 \gamma }{3} \right) \Gamma \left( \frac{1 - \gamma}{3} \right) }.
\end{align*}
Then,
\begin{align}
    E_1^{-1} = \frac{1}{\lambda_0} C \left[S_1^{(0)} \right]^{*-1} K^* \begin{pmatrix}
        \frac{1}{\lambda_1 q^{\R}} & & \\
        & 1 & \\
        & & \frac{q^\R}{\lambda_1} 
    \end{pmatrix}
    K \left[S_1^{(0)} \right]^{-1} \label{E_1 decomposition near 0}
\end{align}
where $\lambda_0 = \frac{4}{3} \sin^2 \frac{\pi}{3} (1 - \gamma)$ and $\lambda_1 = 2 \cos \frac{\pi}{3} (1 - \gamma)$.
\end{thm}

\begin{proof}
    Substituting the decomposition \eqref{decomposition of D_1} of $D_1$ into $E_1 = D_1 D_1^* C$, one can get \eqref{E_1 decomposition near 0} by direct computation.
\end{proof}

Moreover, \eqref{E_1 decomposition near 0} implies
\begin{align}
    \left[ S_1^{(0)} \right]^* C E_1^{-1} S_1^{(0)} = \frac{1}{\lambda_0} K^* \begin{pmatrix}
        \frac{1}{\lambda_1 q^\R} & & \\
        & 1 & \\
        & & \frac{q^\R}{\lambda_1} 
    \end{pmatrix} K. \label{E_1 decomposition near 0 ver. 2}
\end{align}
Note that the LHS of \eqref{E_1 decomposition near 0 ver. 2} is parametrized   only by the  monodromy data
\begin{align}
\begin{aligned}
    &\left[ S_1^{(0)} \right]^* C E_1^{-1} S_1^{(0)}\\
    &= \begin{pmatrix}
        9 A^{\R} & 9 \omega^2 B & 9 \omega \overline{B} \\
        9 \omega \overline{B} & 9 A^{\R} & 9 \omega s^{\R} A^{\R} + 9 \omega^2 B - 9 s^{\R} \overline{B}\\
        9 \omega^2 B & 9 \omega^2 s^{\R} A^{\R} - 9 s^{\R} B + 9 \omega \overline{B} & 9 A^{\R}
    \end{pmatrix}
\end{aligned} \label{expression for decomposition of E_1 inverse I}
\end{align}
by \eqref{parametrization of S_1,2 at 0} and \eqref{parametrization of E_1 ver.2}, while the RHS of \eqref{E_1 decomposition near 0 ver. 2} is parametrized by the asymptotic data
\begin{align}
    \frac{1}{\lambda_0} K^* \begin{pmatrix}
        \frac{1}{\lambda_1 q^\R} & & \\
        & 1 & \\
        & & \frac{q^\R}{\lambda_1} 
    \end{pmatrix} K &= \frac{1}{\lambda_0 \lambda_1}\begin{pmatrix}
        \Lambda_1 & \Lambda_2 & \overline{\Lambda_2} \\
        \overline{\Lambda_2} & \Lambda_1 & \overline{\Lambda_3} \\
        \Lambda_2 & \Lambda_3 & \Lambda_1 
    \end{pmatrix}, \label{expression for decomposition of E_1 inverse II}
\end{align}
where
\begin{align*}
    &\Lambda_1 = \frac{1}{q^\R} + \lambda_1 + q^\R, \qquad \Lambda_2 = \frac{1}{q^\R} \omega^{\gamma} + \lambda_1 \omega + q^\R \omega^{2 - \gamma}, \\
    &\Lambda_3 = \frac{1}{q^\R} \omega^{2 \gamma} + \lambda_1 \omega^2 + q^\R \omega^{4 - 2 \gamma}.
\end{align*}

Comparing \eqref{expression for decomposition of E_1 inverse I} with \eqref{expression for decomposition of E_1 inverse II}, we have
\begin{align} 
    & 9A^{\R} = \frac{1}{\lambda_0 \lambda_1} \Lambda_1 \label{prep for connection formula 1}\\
    & 9 \omega^2 B = \frac{1}{\lambda_0 \lambda_1} \Lambda_2 \label{prep for connection formula 2}\\
    & 9 \omega^2 s^{\R} A^{\R} - 9 s^{\R} B + 9 \omega \overline{B} = \frac{1}{\lambda_0 \lambda_1} \Lambda_3. \label{prep for connection formula 3}
\end{align} 

This provides us with the one-to-one correspondence between monodromy data $\{s^{\R}, y^{\R} \}$ and asymptotic data $\{ \gamma, \rho \}$.
\begin{thm}One has
\begin{align}
    & s^{\R} = -2 \cos \left[\frac{2 \pi}{3}(1 - \gamma) \right] - 1 \label{connection formula 1},\\
    & y^{\R} = \frac{1}{9 \lambda_0 \lambda_1} \left( q^\R - \frac{1}{q^\R} \right) \sin \left[ \frac{2 \pi}{3} (1 - \gamma) \right]. \label{connection formula 2}
\end{align}
where $\lambda_0 = \frac{4}{3} \sin^2 \frac{\pi}{3} (1 - \gamma)$, $\lambda_1 = 2 \cos \frac{\pi}{3} (1 - \gamma)$, and %\begin{align*}
    $q^\R = \frac{2 (\gamma - 1)^2 }{e^{-2\rho}} 3^{2 (\gamma - 1)} \frac{\Gamma \left( \frac{\gamma - 1}{3} \right) \Gamma \left( \frac{2 \gamma - 2}{3} \right) }{ \Gamma \left( \frac{2 - 2 \gamma }{3} \right) \Gamma \left( \frac{1 - \gamma}{3} \right) }.$
%\end{align*}
\end{thm}

\begin{proof}
    Substitute \eqref{prep for connection formula 1} and \eqref{prep for connection formula 2} into \eqref{prep for connection formula 3} to have \eqref{connection formula 1}.
    Take the imaginary part of \eqref{prep for connection formula 2} multiplied by $\omega^2$ to have \eqref{connection formula 2}.
\end{proof}

\begin{rem}
      Since $\gamma$ satifies $-\frac{1}{2} < \gamma < 1$, it follows that $-3 < s^{\R} < 1$ from \eqref{connection formula 1}.
\end{rem}

%%%%%%%%%%%%%%%%%%%%%%%%%%%%%

\subsection{Connection Formulae in terms of $q^\R, \gamma$}

In this subsection, we will describe different ways to write down the connection formulae in terms of $q^\R,\gamma$. Formulae \eqref{prep for connection formula 1}-\eqref{prep for connection formula 3} imply the following statement. 
\begin{prop}\label{Prop xyA}
    Let $\omega B=x^{\R}+iy^{\R}$ and $A^\R$ be the monodromy data, then one has
    \begin{align}
        &x^{\R} \equiv x^{\R}(\rho,\gamma)=\dfrac{1}{9\lambda_0\lambda_1}\left[\bigl(q^\R+\dfrac{1}{q^\R}\bigr)\cos\Bigl(\dfrac{2\pi}{3}(1-\gamma)\Bigr)+\lambda_1\right], \label{x in terms of Nj}\\
        &y^{\R}\equiv y^{\R}(\rho,\gamma)=\dfrac{1}{9\lambda_0\lambda_1}\left[\bigl(q^\R-\dfrac{1}{q^\R}\bigr)\sin\Bigl(\dfrac{2\pi}{3}(1-\gamma)\Bigr)\right],\label{y in terms of Nj}\\
        &3A^\R\equiv 3A^\R(\rho,\gamma)=\dfrac{1}{3\lambda_0\lambda_1}\left[\bigl(q^\R+\dfrac1{q^\R}\bigr) +\lambda_1\right], \label{A in terms of Nj}
    \end{align}
    where 
    \begin{align*}
    q^\R \equiv q^\R(\rho,\gamma)= \frac{2 (\gamma - 1)^2 }{e^{-2\rho}} 3^{2 (\gamma - 1)} \frac{\Gamma \left( \frac{\gamma - 1}{3} \right) \Gamma \left( \frac{2 \gamma - 2}{3} \right) }{ \Gamma \left( \frac{2 - 2 \gamma }{3} \right) \Gamma \left( \frac{1 - \gamma}{3} \right) }.
\end{align*}
\end{prop}
Direct computation shows that the relations between $A^{\R}, B, s^{\R}$ described by 
\eqref{identity 7 at parametrization of E_1} and \eqref{identity 6 at parametrization of E_1} are automatically satisfied  if we substitute formulae \eqref{x in terms of Nj}, \eqref{y in terms of Nj}, and \eqref{A in terms of Nj}.

In order to describe the asymptotics of $w_0$ at infinity in terms of $q^{\R},\gamma$ we need the following corollary of Proposition \ref{Prop xyA}.
\begin{cor} For $\gamma\in (-1/2,1)$, one has 
\begin{align}
\arg \omega B= \arg\left[\left(q^\R+\dfrac{1}{q^\R}\right)\cos\left(\dfrac{2\pi}{3}(1-\gamma)\right) + \lambda_1+i \left(q^\R-\dfrac{1}{q^\R}\right)\sin\left(\dfrac{2\pi}{3}(1-\gamma)\right) \right]
\end{align}
\end{cor}
 
\begin{proof}
    Follows from Proposition \ref{Prop xyA} and the fact that
    $\lambda_0,\lambda_1 > 0$ for $\gamma \in (-1/2, 1)$.
\end{proof}
Observe that the statement of Theorem \ref{result 1} can be written in the following way
\begin{align*}
    w_0(x) = \frac{\sigma}{\sqrt{x}} \cos \left( 2 \sqrt{3} x +  \dfrac{\ln 3A^\R}{2\pi} \, \ln x + \psi \right) + \O \left( \frac{1}{x} \right), 
\end{align*}
where
\begin{align*}
    \sigma^2 &=  \frac{ \sqrt{3}}{ {4 \pi}} \ln 3A^\R, \; \sigma > 0,\\
    \psi &= \frac{\ln (24\sqrt{3})}{2 \pi } \ln 3A^{\R}  + \frac{3\pi}{4} + \arg \omega B - \arg \Gamma\left(i\,\frac{\ln 3A^{\R}}{2 \pi } \right).\\
\end{align*}

Combining these observations we have the following connection formula between the asymptotics of $w_0(x)$ at zero and infinity.

\begin{thm}
For every $\gamma\in (-1/2,1)$ and every $\rho \in\R$, there exists a unique solution $w_0(x)$ of the radial Toda equation, i.e.,
\begin{align*}
      (w_0)_{xx} + \frac{1}{x} (w_0)_x = 2e^{-2w_0} -2e^{4w_{0}},
\end{align*}
such that
\begin{align*}
    w_0(x) = \gamma \ln x + \rho + o(1), \qquad x\rightarrow 0,
\end{align*}
and 
\begin{align*}
 w_0(x) = \frac{\sigma}{\sqrt{x}} \cos \left( 2 \sqrt{3} x +  \dfrac{\ln 3A^\R}{2\pi} \, \ln x + \psi \right) + \O \left( \frac{1}{x} \right), \quad x\rightarrow \infty,
 \end{align*}
 where 
 \begin{align*}
    \sigma^2 &=  \frac{ \sqrt{3}}{ {4 \pi}} \ln 3A^\R, \; \sigma > 0,\\
    \psi &= \frac{\ln (24\sqrt{3})}{2 \pi } \ln 3A^{\R}  + \frac{3\pi}{4} + \arg \omega B - \arg \Gamma\left(i\,\frac{\ln 3A^{\R}}{2 \pi } \right),\\
    3A^\R  &=\dfrac{1}{3\lambda_0\lambda_1}\left[\left(q^\R+\dfrac{1}{q^\R}\right) +\lambda_1\right],\\
    \arg \omega B &=  \arg \left[ \left(q^\R+\dfrac{1}{q^\R}\right)\cos\left(\dfrac{2\pi}{3}(1-\gamma)\right) + \lambda_1 + i \left(q^\R-\dfrac{1}{q^\R}\right) \sin \left(\dfrac{2\pi}{3}(1-\gamma)\right) \right],\\
    \lambda_1 &= 2 \cos \frac{\pi}{3} (1 - \gamma),\\
    \lambda_0 &= \frac{4}{3} \sin^2 \frac{\pi}{3} (1 - \gamma),\\
    q^\R &= \frac{2 (\gamma - 1)^2 }{e^{-2\rho}} 3^{2 (\gamma - 1)} \frac{\Gamma \left( \frac{\gamma - 1}{3} \right) \Gamma \left( \frac{2 \gamma - 2}{3} \right) }{ \Gamma \left( \frac{2 - 2 \gamma }{3} \right) \Gamma \left( \frac{1 - \gamma}{3} \right) }.
 \end{align*}
 or
\begin{align*}
    \sigma^2 &=  \frac{ \sqrt{3}}{ {2}} X, \; \sigma > 0,\\
    \psi &=  \ln (24\sqrt{3})  X  + \frac{3\pi}{4} + \alpha  - \arg \Gamma\left(i\,X \right),\\
    X &= \dfrac{1}{2\pi}\ln \left[ \dfrac{1}{8 \cos\pi\tilde{\gamma} \sin^2\pi\tilde{\gamma}}    \left(q^\R+\dfrac{1}{q^\R}\right) +\dfrac{1}{4\sin^2\pi\tilde{\gamma}} \right],\\
    \alpha &=  \arg \left( \left(q^\R+\dfrac{1}{q^\R}\right) \cos 2\pi\tilde{\gamma} +2\cos\pi\tilde{\gamma} + i \left(q^\R-\dfrac{1}{q^\R}\right) \sin 2\pi\tilde{\gamma} \right),\\
    q^\R &= \frac{18 \tilde{\gamma}^2 }{e^{-2\rho}} 3^{-6 \tilde{\gamma}} \frac{\Gamma (-\tilde{\gamma}) \Gamma (-2\tilde{\gamma}) }{ \Gamma (\tilde{\gamma}) \Gamma (2\tilde{\gamma}) },\\
    \tilde{\gamma}&=\dfrac{1-\gamma}{3}.
 \end{align*}
\end{thm}

Every real solution of \eqref{negative tt*-Toda with x when n=2} corresponds to a point in the monodromy data set $\mathcal{M}$ (see \eqref{monodormy data}). Given any monodromy data $(s^\R, y^\R)$, with $A^\R > 0$, we can solve equations \eqref{connection formula 1} and \eqref{connection formula 2} for $\gamma$ and $\rho$.
As stated in the remark in section \ref{section titled more about the monodromy data}, we shall show in \cite{gikmo2} that there are no solutions of \eqref{negative tt*-Toda with x when n=2} when $A^\R < 0$.
These facts yield the last statement of Theorem \ref{main theorem}, i.e., the completeness of the description \eqref{zero} of the behavior of solutions of \eqref{negative tt*-Toda with x when n=2} at $x = 0$.

%% file: Appendix.tex
\appendix

%%%%%%%%%%%%%%%%%%%%%%%%%%%%%

\section{Parametrization of the Connection Matrix $E_1$}\label{proof of E}

Recall the cyclic symmetry relation \eqref{cyclic symmetry for E_1 updated}:
  \begin{align*}
        d_3^{-1} E_1 = \omega \left( Q_2^{(\infty)}  Q_{2\frac{1}{3}}^{(\infty)} \Pi \right) d_3^{-1} E_1 \left(Q_1^{(\infty)} Q_{1\frac{1}{3}}^{(\infty)} \Pi \right).
    \end{align*}
It can be written as  
\begin{align}
    d_3^{-1} E_1 \left( Q_1^{(\infty)}  Q_{1\frac{1}{3}}^{(\infty)} \Pi \right)^{-1} = \omega \left( Q_2^{(\infty)}  Q_{2\frac{1}{3}}^{(\infty)} \Pi \right) d_3^{-1} E_1. \label{cyclic symmetry for E_1 to parametrize}
\end{align}
Let us parametrize  $E_1$ by complex numbers $A, B, D, F, G, H, K, L, M \in \C$:
\begin{align*}
   \begin{pmatrix}
        A & B & D\\
        F& G & H\\
        K & L & M
    \end{pmatrix},
\end{align*}
then the LHS of \eqref{cyclic symmetry for E_1 to parametrize} is
\begin{align*}
    d_3^{-1} E_1 \Pi ^{-1} Q_{1\frac{1}{3}}^{(\infty)-1} Q_1^{(\infty)-1}
    = \begin{pmatrix}
        B          & a \omega^2 A - a B + D                 & A\\
        \omega^2 G & a \omega F - a \omega^2 G + \omega^2 D & \omega^2 F\\
        \omega L   & a K - a \omega L + \omega M & \omega K
    \end{pmatrix},
\end{align*}
and the RHS of \eqref{cyclic symmetry for E_1 to parametrize} is
\begin{align*}
   & \omega \left( Q_2^{(\infty)}  Q_{2\frac{1}{3}}^{(\infty)} \Pi \right) d_3^{-1} E_1 =\\
    &\begin{pmatrix}
        F & G & H\\
        - a \omega^2 F + \omega^2 K + a \omega A & -a\omega^2 G + \omega^2 L + a \omega B & -a\omega^2 H + \omega^2 M + a \omega D\\
        \omega A & \omega B & \omega D
    \end{pmatrix}.
\end{align*}
Therefore, we should have 
\begin{align*}
    F &= B, \; G = a \omega^2 A - aB + D, \; H = A\\ 
    L &= A, \; M = a A + B - a \omega^2 D, \; K = D.
\end{align*}
So, the parametrization of $E_1$ can be reduced to $3$ variables, $A, \; B,$ and $D$:
\begin{align}
    E_1 &= \begin{pmatrix}
        A & B & D\\
        B & a\omega^2 A - a B + D & A\\
        D & A & a A + B - a \omega^2 D
    \end{pmatrix} \label{parametrization of E_1 ver.1}
\end{align}
which satisfies the symmetry $E_1=E_1^T$ observed in Proposition \ref{E symmetry 1}. From the relation \eqref{antisymmetry for E updated} with this symmetry, one has
\begin{align}
\begin{aligned}
     \left(S_1^{(\infty)} d_3^{-1} E_1\right)^2   = \frac{1}{9} I,
\end{aligned} \label{antisymmetry for E_1 to parametrize}
\end{align}
 where
\begin{align*}
    & S_1^{(\infty)} d_3^{-1} E_1\\
    &= \begin{pmatrix}
        A + a \omega^2 B & B + a^2 \omega A - a^2 \omega^2 B + a \omega^2 D & D + a \omega^2 A \\
        \omega^2 B & a \omega A - a \omega^2 B + \omega^2 D & \omega^2 A\\
        a A - a \omega B + \omega D & a B - a^2 A + a^2 \omega B - a \omega D + \omega A & \omega B
    \end{pmatrix}.
\end{align*}
Comparing $(S_1^{(\infty)} d_3^{-1} E_1)^2$ with $\frac{1}{9} I$, we have
\begin{align}
   A^2 + \omega^2 B^2 + \omega D^2 + a\omega^2 AB - a \omega BD + a DA = \frac{1}{9},
\label{identity 1 at parametrization of E_1}
\end{align}
and
\begin{align}
\begin{aligned}
    a A^2 + AB + \omega^2 BD + \omega AD = 0.\label{identity 2 at parametrization of E_1}
\end{aligned}
\end{align}
All other relations are equivalent to these two. Subtracting the $a\omega^2$ multiple of \eqref{identity 2 at parametrization of E_1} from \eqref{identity 1 at parametrization of E_1}, we have
\begin{align}
    (1 - a^2 \omega^2) A^2 + \omega^2 B^2 - 2 a \omega BD + \omega D^2 = \frac{1}{9}. \label{identity 3 at parametrization of E_1}
\end{align}
Furthermore, by the reality condition  \eqref{reality of E_1}
 \begin{align*}
        E_1 = C \overline{E_1} C,
    \end{align*}
we should have
\begin{align*}
   &\begin{pmatrix}
        A & B & D\\
        B & a\omega^2 A - a B + D & A\\
        D & A & a A + B - a \omega^2 D
    \end{pmatrix}=\begin{pmatrix}
        \overline{A} & \overline{D} & \overline{B}\\
        \overline{D} & \overline{a} \overline{A} + \overline{B} - \overline{a} \omega \overline{D} & \overline{A}\\
        \overline{B} & \overline{A} & \overline{a} \omega \overline{A} - \overline{a} \overline{B} + \overline{D}
    \end{pmatrix}.
\end{align*}
Thus, we have
\begin{align}
    \begin{aligned}
        A &= \overline{A}, \; B = \overline{D}, \; a\omega^2 A - a B + D = \overline{a} \overline{A} + \overline{B} - \overline{a} \omega \overline{D}.
    \end{aligned}
\end{align}
Under the conditions  $A = \overline{A}, \; B = \overline{D}$ the last equation is equivalent to   \eqref{a = omega^2 s}: $a = \omega^2 s^\R$ . We will write $A$ as $A^{\R}$ to highlight that it is a real number.  Equations \eqref{identity 2 at parametrization of E_1}, \eqref{identity 3 at parametrization of E_1} become 
\begin{align}
   & s^{\R}(A^{\R})^2 + |B|^2 + A^{\R}(\omega B + \omega^2 \bar{B}) = 0,
    \label{identity 4 at parametrization of E_1}\\
    &(1 - (s^{\R})^2)(A^{\R})^2 - 2 s^{\R} |B|^2 + (\omega^2 B^2 + \omega \bar{B}^2 ) = \frac{1}{9}.
    \label{identity 5 at parametrization of E_1}
\end{align}
From \eqref{identity 4 at parametrization of E_1} and \eqref{identity 5 at parametrization of E_1}, it follows that
\begin{align*}
    [(1 + s^{\R}) A^{\R} + \omega B + \omega^2 \overline{B}]^2 &= (1 + (s^{\R})^2 + 2 s^{\R}) (A^{\R})^2 + \omega^2 B^2 + \omega \overline{B}^2\\ 
    & \qquad \qquad + 2[(1 + s^{\R})\omega AB + |B|^2 + (1 + s^{\R})\omega^2 AB]\\
    &= (1 + (s^{\R})^2) (A^{\R})^2 + \omega^2 B^2 + \omega \overline{B}^2 + 2 [s^{\R}\omega AB + s^{\R}\omega^2 A \overline{B}]\\
    &= (1 - (s^{\R})^2) (A^{\R})^2 + \omega^2 B^2 + \omega \overline{B}^2 - 2 s^{\R} |B|^2\\
    &= \frac{1}{9},
\end{align*}
or 
\begin{align*}
    (1 + s^{\R}) A^{\R} + \omega B + \omega^2 \overline{B}  = \pm\frac{1}{3}.
\end{align*}
One can show that only
\begin{align}
    (1 + s^{\R})A^{\R} + \omega B + \omega^2 \overline{B} = \frac{1}{3} \label{identity 6 at parametrization of E_1 ap}
\end{align}
is consistent with  $\det E_1 = - \frac{1}{27}$.
Moreover, by \eqref{identity 4 at parametrization of E_1} and \eqref{identity 6 at parametrization of E_1 ap}, it follows that
\begin{align}
    (A^{\R})^2 - \frac{1}{3} A^{\R} = |B|^2. \label{identity 7 at parametrization of E_1 ap}
\end{align}

%%%%%%%%%%%%%%%%%%%%%%%%%%%%%%%%%%

\section{Proof of Proposition \ref{ABS TH}} \label{Proof of ABS}

Let $\omega B = x^{\R} + iy^{\R}$ where $x^{\R}, y^{\R} \in \R$, \eqref{identity 6 at parametrization of E_1} becomes
\begin{align}
    (1 + s^{\R}) A^{\R} + 2 x^{\R} = \frac{1}{3} \Leftrightarrow x^{\R} = \frac{1 - 3(1 + s^{\R}) A^{\R}}{6}, \label{expression for x real}
\end{align}
and \eqref{identity 7 at parametrization of E_1} becomes
\begin{align}
    (A^{\R})^2 - \frac{1}{3} A^{\R} - (x^{\R})^2 = (y^{\R})^2. \label{identity 8 at parametrization of E_1}
\end{align}
Substituting \eqref{expression for x real} into \eqref{identity 8 at parametrization of E_1}, we have
\begin{align}
    (y^{\R})^2 = \frac{(3 + s^{\R})(1 - s^{\R})}{4}(A^{\R})^2 - \frac{1 - s^{\R}}{6} A^{\R} - \frac{1}{36}  \label{identity 9 at parametrization of E_1}
\end{align}
Since $A^\R \geq \frac13$ , the cases $s^\R=1$ and $s^\R=-3$ contradict $(y^\R)^2>0$. Solving \eqref{identity 9 at parametrization of E_1} for $A^\R$ gives
\begin{align}\label{A pm}
    A^{\R} = \frac{1}{3(3 + s^{\R})} \pm \dfrac{2 \sqrt{ \frac{(1-s^\R)^2}{36} +  (3 + s^{\R})(1 - s^{\R}) \left( \frac{1}{36} + (y^{\R})^2 \right)}}{(3 + s^{\R})(1-s^\R)}.
\end{align}
When  $s^\R>1$ or $s^\R<-3$, formula \eqref{A pm} gives either   complex     or  negative   $A^\R$, thus we should have  $-3<s^\R<1$. Moreover, for $-3<s^\R<1$ only 
\begin{align*}
    A^{\R} = \frac{1}{3(3 + s^{\R})} + \frac{2}{3 + s^{\R}} \sqrt{ \frac{1}{36} + \frac{3 + s^{\R}}{1 - s^{\R}} \left( \frac{1}{36} + (y^{\R})^2 \right)}.
\end{align*}
gives positive $A^\R$.

Finally, by \eqref{expression for x real}, one can also write $B$ in terms of $s^{\R}$ and $y^{\R}$ only.
\begin{align*}
    B = \omega^2 (x^{\R} + i y^{\R}) = \omega^2 \left( \frac{1 - 3(1 + s^{\R}) A^{\R}}{6} + i y^{\R} \right)
\end{align*}
 
%%%%%%%%%%%%%%%%%%%%%%%%%%%%%%%%%%

\section{Proof of Proposition \ref{Prop E decomp}} \label{Appendix E decomp}
 
Consider  the decomposition $\check{E}_1^{-1} =\check{L}_1 D_1 \check{R}_1$ , where by \eqref{L1R1} matrices $\check{L}_1$ and  $  \check{R}_1$ have the following structure 
\begin{align}
   \check{L}_1= \begin{pmatrix}
        1 & * e^{-x \varphi_1 } & * e^{-x\varphi_2}\\
        0  & 1 & 0\\
        0  &  * e^{x \varphi_3} & 1
    \end{pmatrix},\quad 
    \check{R}_1=
    \begin{pmatrix}
        1 & 0 & 0\\
        * e^{x \varphi_1 } & 1 & * e^{-x \varphi_3}\\
        * e^{x \varphi_2} &  0 & 1
    \end{pmatrix}.
\end{align} 

Observe that the desired decomposition can be obtained from the decomposition 
\begin{align}
    &\tilde{E}_1^{-1} = \frac{1}{3} E_1^{-1} C = L_1 D_1 R_1,
\label{R1L1}\\
 &{L}_1= \begin{pmatrix}
        1 & *   & *  \\
        0  & 1 & 0\\
        0  &  *   & 1
    \end{pmatrix},\quad 
    {R}_1=
    \begin{pmatrix}
        1 & 0 & 0\\
        *   & 1 & * \\
        *   &  0 & 1
    \end{pmatrix},
\end{align} 
since by \eqref{check E} and \eqref{LLcheck}
\begin{align}
    \check{E_1}= e^{x \theta(\zeta)} \tilde{E}_1 e^{-x \theta(\zeta)}=e^{x \theta(\zeta)} L_1 D_1 R_1 e^{-x \theta(\zeta)}= \check{L}_1 D_1 \check{R}_1.
\end{align}
Taking the inverse of \eqref{R1L1}
\begin{align}
    3C E_1 = R_1^{-1} D_1^{-1} L_1^{-1} \label{RDL for3CE_1}
\end{align}
we introduce the following parameterization  
\begin{align}\label{L1D1R1}
    \quad
    R_1^{-1} &= \begin{pmatrix}
        1 & 0 & 0\\
        a_1 & 1 & b_1\\
        c_1 & 0 & 1
    \end{pmatrix}, \;
    D_1^{-1} = \begin{pmatrix}
        d_1 & 0 & 0\\
        0 & e_1 & 0\\
        0 & 0 & f_1
    \end{pmatrix}, \;
    L_1^{-1} = \begin{pmatrix}
        1 & g_1 & h_1\\
        0 & 1 & 0\\
        0 & i_1 & 1
    \end{pmatrix}.
\end{align}
Then the  RHS of \eqref{RDL for3CE_1} is
\begin{align}\label{C.1}
    R_1^{-1} D_1^{-1} L_1^{-1} = \begin{pmatrix}
        d_1 & d_1 g_1 & d_1 h_1\\
        a_1 d_1 & a_1 d_1 g_1 + e_1 + b_1 f_1 i_1 & a_1 d_1 h_1 + b_1 f_1\\
        c_1 d_1 & c_1 d_1 g_1 + f_1 i_1 & c_1 d_1 h_1 + f_1
    \end{pmatrix}.
\end{align}
On the other hand, the LHS of \eqref{RDL for3CE_1} is  
\begin{align}\label{C.3}
\begin{aligned}
    \qquad \; \; 3CE_1 &= \begin{pmatrix}
        3A^{\R} & 3B & 3\overline{B}\\
        3\overline{B} & 3A^{\R} & 3\omega^2 s^{\R} A^{\R} + 3B - 3\omega s^{\R} \overline{B}\\
        3B & 3\omega s^{\R} A^{\R} - 3\omega^2 s^{\R} B + 3\overline{B} & 3A^{\R}
    \end{pmatrix},
\end{aligned}
\end{align}
where $A^{\R}, \, B$ were introduced in Appendix \ref{proof of E}.
Using identities \eqref{identity 6 at parametrization of E_1} and \eqref{identity 7 at parametrization of E_1} between $A^{\R}$ and $B$ we get 
\begin{align*}
    &(1 + s^{\R})A^{\R} + \omega B + \omega^2 \overline{B} = \frac{1}{3}, \quad (A^{\R})^2 - \frac{1}{3} A^{\R} = |B|^2.
\end{align*}
Comparing  \eqref{C.1} and \eqref{C.3}  we have
\begin{align} 
\begin{aligned} \quad
    &d_1 = 3A^{\R}, \; a_1 = \overline{g_1} = \frac{\overline{B}}{A^{\R}}, \; c_1 = \overline{h_1} = \frac{B}{A^{\R}}, \; f_1 = \frac{3((A^{\R})^2 - |B|^2)}{A^{\R}} = 1,\\
    & b_1 = \overline{i_1} = \frac{\omega^2 s^{\R} (A^{\R})^2 - \overline{B}^2 + A^{\R} B - \omega s^{\R} A^{\R} \overline{B}}{(A^{\R})^2 - |B|^2} = -\omega \frac{\overline{B}}{A^{\R}},\\
    &e_1 = \frac{3((A^{\R})^2 - |B|^2)}{A^{\R}} (1 - |b_1|^2) = 1 - \frac{|B|^2}{A^{\R}},
\end{aligned}\label{parametrization of L_1 D_1 R_1 ver.1}
\end{align}
since 
\begin{align}
\begin{aligned}
    \omega^2s (A^{\R})^2 - \overline{B}^2 + A^{\R} B - \omega s   A^{\R} \overline{B} &= - \omega^2 |B|^2 - \omega A^{\R} \overline{B} - \overline{B}^2 -   \omega s A^{\R} \overline{B}\\
    &= -\omega A^{\R} \overline{B} - \frac{1}{3} \omega \overline{B} + \omega A^{\R} \overline{B}\\
    &= - \frac{1}{3} \omega \overline{B}.
\end{aligned} \label{identity 13 at parametrization of E_1}
\end{align}

Finally,  substituting \eqref{parametrization of L_1 D_1 R_1 ver.1} into  \eqref{L1D1R1}, \eqref{RDL for3CE_1} and taking inverse, we obtain that the decomposition   $\Tilde{E}_1^{-1}=  L_1 D_1 R_1$ such that 
\begin{align*}
\begin{aligned}
    L_1 &= \begin{pmatrix}
     1 & - \frac{B}{A^{\R}} - \omega^2 \frac{|B|^2}{(A^{\R})^2} & -\frac{\overline{B}}{A^{\R}}\\
     0 & 1 & 0\\
     0 & \omega^2 \frac{B}{A^{\R}} & 1
    \end{pmatrix}, \;
    D_1 = \begin{pmatrix}
        \frac{1}{3A^{\R}} & & \\
        & 3 A^{\R} & \\
        & & 1
    \end{pmatrix},\\
    R_1 &= \begin{pmatrix}
        1 & 0 & 0\\
        - \frac{\overline{B}}{A^{\R}} - \omega \frac{|B|^2}{(A^{\R})^2} & 1 & \omega\frac{\overline{B}}{A^{\R}}\\
        - \frac{B}{A^{\R}} & 0 & 1
    \end{pmatrix}.
\end{aligned} 
\end{align*}
All the other decompositions can be obtained similarly. 

%%%%%%%%%%%%%%%%%%%%%%%%%%%%

\section{Proof of Proposition \ref{Residue prop}}\label{proof of res}

First, we consider the jump matrix $G_R(\zeta)$ on $\partial U_1$. By \eqref{matching up cond for parametrix near 1} and \eqref{asymptotics of Y^D}, it holds that
\begin{align*}
    &G_R(\zeta) = \check{Y}^D \left( P^{(1)}(\zeta) \right)^{-1}\\
    &= \check{Y}^D(\zeta)
    \left(I - \frac{1}{z} 
    \left\{\!\begin{aligned}
        &I \\[1ex]
        &D_1^{-1} \\[1ex]
        &D_6^{-1}
    \end{aligned}\right\}
    \begin{pNiceMatrix}
        1 & 0 & 0 \\
        0 & \Block{2-2}{z^{-\nu \sigma_3}} \\
        0 \\
    \end{pNiceMatrix}
    \begin{pNiceMatrix}
        0 & 0 & 0 \\
        0 & \Block{2-2}{m} \\
        0 \\
    \end{pNiceMatrix}
    \begin{pNiceMatrix}
        1 & 0 & 0 \\
        0 & \Block{2-2}{z^{\nu \sigma_3}} \\
        0 \\
    \end{pNiceMatrix}
    \left\{\!\begin{aligned}
        &I \\[1ex]
        &D_1 \\[1ex]
        &D_6
    \end{aligned}\right\}
    + \cdots \right) \check{Y}^{D-1}\\
    &= \Theta  (\zeta) \begin{pmatrix}
        1 & & \\
        & z^{-\nu}(\zeta - 1)^{\nu} & \\
        & & z^{\nu}(\zeta - 1)^{-\nu}
    \end{pmatrix}\\
    &\qquad \times \left(I - \frac{1}{z} 
    \begin{pNiceMatrix}
        0 & 0 & 0 \\
        0 & \Block{2-2}{m} \\
        0 \\
    \end{pNiceMatrix}
    + \cdots \right )
    \begin{pmatrix}
        1 & & \\
        & z^{\nu}(\zeta - 1)^{-\nu} & \\
        & & z^{-\nu}(\zeta - 1)^{\nu}
    \end{pmatrix}
    \Theta^{-1}(\zeta)\\
    &= I - \frac{1}{z(\zeta)} \widetilde{\Theta }(\zeta) \begin{pmatrix}
        1 & & \\
        & \kappa^{-1} & \\
        & & \kappa
    \end{pmatrix}
    \begin{pNiceMatrix}
        0 & 0 & 0 \\
        0 & \Block{2-2}{m} \\
        0 \\
    \end{pNiceMatrix}
    \begin{pmatrix}
        1 & & \\
        & \kappa & \\
        & & \kappa^{-1}
    \end{pmatrix}
   \widetilde{\Theta }(\zeta)^{-1} + \O\left( \frac{1}{x} \right),
\end{align*}
where we set
\begin{align*}
    &z^{\nu} (\zeta - 1)^{- \nu} = \kappa \left( 1 + \O(\zeta - 1) \right), \\ 
    &\kappa = e^{\frac{3 \pi i}{4} \nu} (2 x)^{\nu/2} 3^{\nu/4},
\end{align*}
and $\widetilde{\Theta}(\zeta)=\left(I+\sum_{l=1}^\infty \widetilde{\Theta}_l(\zeta-1)^l\right) F$ is some Taylor series coming from 
\begin{align}
    \widetilde{\Theta }(\zeta) \begin{pmatrix}
        1 & & \\
        & \kappa^{-1} & \\
        & & \kappa\end{pmatrix} = \Theta  (\zeta) \begin{pmatrix}
        1 & & \\
        & z^{-\nu}(\zeta - 1)^{\nu} & \\
        & & z^{\nu}(\zeta - 1)^{-\nu}
    \end{pmatrix}.
\end{align}

Since
\begin{align*}
    F &=
    \begin{pNiceMatrix}[margin]
         e^{\pi i \nu} & & 0  & & & 0  & \\
         0 & & & \Block{2-2}{e^{\frac{\pi i \nu}{2}} \left( 2\sqrt{3}\right)^{-\nu \sigma_3} } \\
         0 \\
    \end{pNiceMatrix},
\end{align*}
we can compute the residue at $\zeta = 1$ as
\begin{align}
    \res_{\zeta = 1} \frac{G_R(\zeta) - I}{1} = \begin{pNiceMatrix}
        0 & 0 & 0\\
        0 & \Block{2-2}{\widehat{m}}\\
        0 \\ 
    \end{pNiceMatrix} \label{residue at 1 ver 1}
\end{align}
where
\begin{align*}
    \widehat{m} = - \frac{1}{\sqrt{2x}} 3^{-1/4} e^{-\frac{3 \pi i}{4}} (2 \sqrt{3})^{- \nu \sigma_3} \kappa^{-\sigma_3}
    \begin{pmatrix}
        0 &  -\alpha\\
        \nu/\alpha & 0\\
    \end{pmatrix}
    \kappa^{\sigma_3}(2 \sqrt{3})^{\nu \sigma_3}.
\end{align*}
It  follows then from \eqref{residue at 1 ver 1} that 
\begin{align}
   \res_{\zeta = 1} \frac{G_R(\zeta) - I}{1} = \begin{pmatrix}
        0 & 0 & 0 \\
        0 & 0 & \widehat{\alpha}\\
        0 & \widehat{\beta} & 0
    \end{pmatrix} \label{residue at 1 ver 2}
\end{align}
where 
\begin{align*}
    &\widehat{\alpha} = \frac{\alpha}{\sqrt{2x}} (24 \sqrt{3} x)^{- \nu} 3^{-1/4} e^{-\frac{3 \pi i}{4} - \frac{3 \pi i}{2} \nu}\\
    &\widehat{\beta} = -\frac{1}{\sqrt{2x}} \frac{\nu}{\alpha} (24 \sqrt{3} x)^{\nu} 3^{-1/4} e^{-\frac{3 \pi i}{4} + \frac{3 \pi i}{2} \nu}
\end{align*}
with
\begin{align*}
    \alpha &= - \frac{i}{s_1 e^{i 2\sqrt{3} x}} \frac{\sqrt{2 \pi} e^{2\pi i \nu}}{\Gamma(-\nu)}, \quad \nu = \frac{1}{2 \pi i } \ln (1 - |s_1|^2), \quad s_1 = \omega^2 \frac{B}{A^{\R}}.
\end{align*}

Using symmetry relations between local parametrices, \eqref{local parametrix near -1} -- \eqref{local parametrix near minus omega}, one can find all   other residues. For example, by \eqref{local parametrix near -1}  we have
\begin{align*}
    G_R(\zeta) = d_3^{T-1} G_R^{T-1}(- \zeta) d_3
\end{align*}
and thus
\begin{align*}
    \frac{1}{2 \pi i} \int_{\partial U_{-1}} \frac{G_R(\zeta) - I}{\zeta} d\zeta &= \frac{1}{2 \pi i} \int_{\partial U_{-1}} \frac{ d_3^{T-1} (G_R^{T-1}(- \zeta) - I)d_3}{\zeta} d\zeta\\
    &= d_3^{-1} \left[ \frac{1}{2 \pi i} \int_{\partial U_{1}} \frac{ G_R(\zeta) - I}{\zeta} d\zeta \right]^{T-1} d_3.
\end{align*}
This implies that
\begin{align}
\begin{aligned}
    \res_{\zeta = -1} \frac{G_R(\zeta) - I}{\zeta} &= d_3^{-1} \left[ -\res_{\zeta = 1} \frac{G_R(\zeta) - I}{\zeta} \right]^{T} d_3\\
    &= \begin{pmatrix}
        0 & 0 & 0 \\
        0 & 0 & -\widehat{\beta} \omega\\
        0 & -\widehat{\alpha}\omega^2 & 0
    \end{pmatrix}. \label{residue at -1}
\end{aligned}
\end{align}
Similarly, the relation \eqref{local parametrix near omega} implies
\begin{align}
\begin{aligned}
    \res_{\zeta = \omega} \frac{G_R(\zeta) - I}{\zeta} &= \Pi \left[ \res_{\zeta = 1} \frac{G_R(\zeta) - I}{\zeta} \right] \Pi^{-1} = \begin{pmatrix}
        0 & \widehat{\alpha} & 0 \\
        \widehat{\beta}  & 0 & 0\\
        0 & 0 & 0
    \end{pmatrix}. \label{residue at omega}
\end{aligned}
\end{align}
The relation \eqref{local parametrix near - omega bar} implies
\begin{align}
\begin{aligned}
    \res_{\zeta = -\overline{\omega}} \frac{G_R(\zeta) - I}{\zeta} &= \Pi^{-1} \left[ \res_{\zeta = -1} \frac{G_R(\zeta) - I}{\zeta} \right] \Pi \\
    &= \begin{pmatrix}
        0 & 0 & -\widehat{\alpha} \omega^2 \\
        0 & 0 & 0\\
        -\widehat{\beta}\omega & 0 & 0
    \end{pmatrix}. \label{residue at minus omega bar}
\end{aligned}
\end{align}
The relation \eqref{local parametrix near omega bar} implies
\begin{align}
\begin{aligned}
    \res_{\zeta = \overline{\omega}} \frac{G_R(\zeta) - I}{\zeta} &= d_3^{-1} \left[ -\res_{\zeta = -1} \frac{G_R(\zeta) - I}{\zeta} \right]^{T} d_3 = \begin{pmatrix}
        0 & 0 & \widehat{\beta} \\
        0 & 0 & 0\\
        \widehat{\alpha} & 0 & 0
    \end{pmatrix}. \label{residue at omega bar}
\end{aligned}
\end{align}
The relation \eqref{local parametrix near minus omega} implies
\begin{align}
\begin{aligned}
    \res_{\zeta = -\omega} \frac{G_R(\zeta) - I}{\zeta} &= d_3^{-1} \left[ \res_{\zeta = \omega} \frac{G_R(\zeta) - I}{\zeta} \right]^{T} d_3\\
    &= \begin{pmatrix}
        0 & -\widehat{\beta} \omega & 0 \\
        -\widehat{\alpha} \omega^2  & 0 & 0\\
        0 & 0 & 0
    \end{pmatrix}. \label{residue at minus omega}
\end{aligned}
\end{align}
Recall that
\begin{align*}
    R(0) &= I + \sum_{k = 1}^{6} \frac{1}{2 \pi i} \int_{\gamma_k} \frac{G_R (\zeta') - I}{\zeta'} d \zeta' + \O \left( \frac{1}{x} \right). 
\end{align*}
Therefore, combining the residue calculations \eqref{residue at 1 ver 2} -- \eqref{residue at minus omega}, we obtain 
\begin{align}
    R(0) &= \begin{pmatrix}
        1 & \widehat{\alpha} - \omega \widehat{\beta} & \widehat{\beta} - \omega^2 \widehat{\alpha}\\
        \widehat{\beta} - \omega^2 \widehat{\alpha} & 1 & \widehat{\alpha} - \omega \widehat{\beta}\\
        \widehat{\alpha} - \omega \widehat{\beta} & \widehat{\beta} - \omega^2 \widehat{\alpha} & 1
    \end{pmatrix} + \O\left( \frac{1}{x} \right). \label{appendix}
\end{align}

%%%%%%%%%%%%%%%%%%%%%%%%

\section{Proof of Theorem \ref{result 1} }\label{proof of main}

We have seen that 
\begin{align*}
    w_0 (x) = \frac{\sqrt{3}}{2} e^{-\frac{5 \pi i}{6}} (\widehat{\alpha} - \omega \widehat{\beta})  + \O \left( \frac{1}{x} \right),
\end{align*}
where
\begin{align*}
    &\widehat{\alpha} = \frac{\alpha}{\sqrt{2x}} (24 \sqrt{3} x)^{- \nu} 3^{-1/4} e^{-\frac{3 \pi i}{4} - \frac{3 \pi i}{2} \nu}\\
    &\widehat{\beta} = -\frac{1}{\sqrt{2x}} \frac{\nu}{\alpha} (24 \sqrt{3} x)^{\nu} 3^{-1/4} e^{-\frac{3 \pi i}{4} + \frac{3 \pi i}{2} \nu}\\
    &\alpha = - \frac{i}{s_1 e^{i 2\sqrt{3} x}} \frac{\sqrt{2 \pi} e^{2\pi i \nu}}{\Gamma(-\nu)},\quad  \nu  = \frac{1}{2 \pi i } \ln (1 - |s_1|^2).
\end{align*}
Observe that
\begin{align*}
    \frac{\nu}{\alpha} &= i e^{2i \sqrt{3} x} \frac{s_1 \Gamma(-\nu)}{\sqrt{2 \pi}} e^{-2 \pi i \nu} \nu = - i e^{ 2 i \sqrt{3} x} s_1 \sqrt{\frac{\pi}{2}}\frac{1}{\Gamma(\nu)} \frac{e^{-2 \pi i \nu}}{\sin \pi \nu}\\
    &= - i e^{2i \sqrt{3} x} s_1 \sqrt{\frac{\pi}{2}}\frac{1}{\Gamma(\nu)} \frac{e^{-3 \pi i \nu}}{\sin \pi \nu e^{- \pi i \nu}} = - e^{2 i \sqrt{3}x} \frac{\sqrt{2 \pi}}{\overline{s_1} \, \Gamma(\nu)} e^{- \pi i \nu},
\end{align*}
where we used
\begin{equation}
\begin{aligned}
    &\Gamma(-\nu) \cdot \nu = - \frac{\pi}{\sin \pi \nu \cdot \Gamma(\nu)},\\
    &\sin \pi \nu \cdot e^{- \pi i \nu} = -\frac{e^{- 2 \pi i \nu}}{2 i} |s_1|^2 \label{useful relations of gamma function}.
\end{aligned}
\end{equation}
Thus,
\begin{align*}
    \widehat{\alpha} - \omega \widehat{\beta} &= - \frac{i}{s_1 e^{2i \sqrt{3} x}} \frac{\sqrt{2 \pi} e^{2\pi i \nu}}{\Gamma(-\nu)} \frac{1}{\sqrt{2 x}} (24 \sqrt{3} x)^{- \nu} 3^{-1/4} e^{-\frac{3 \pi i}{4} - \frac{3 \pi i}{2} \nu}\\
    & \qquad \qquad + \frac{\omega}{\sqrt{2x}} \left( - e^{-2 i \sqrt{3}x} \frac{\sqrt{2 \pi}}{\overline{s_1} \, \Gamma(\nu)} e^{- \pi i \nu} \right) (24 \sqrt{3} x)^{\nu} 3^{-1/4} e^{-\frac{3 \pi i}{4} + \frac{3 \pi i}{2} \nu}\\
    &= - i e^{-\frac{3\pi i}{4}} \frac{e^{- 2i \sqrt{3} x}}{s_1 \Gamma(-\nu)} \sqrt{\frac{\pi}{x}} (24 \sqrt{3} x)^{- \nu} 3^{-1/4} e^{\frac{\pi i}{2} \nu}\\
    & \qquad \qquad \qquad - \omega e^{- \frac{3 \pi i}{4}} \frac{e^{2i \sqrt{3}x}}{\overline{s_1} \, \Gamma(\nu)} \sqrt{\frac{\pi}{x}} (24 \sqrt{3} x)^{\nu} 3^{-1/4} e^{\frac{\pi i}{2} \nu}.
\end{align*}
Note that we have $e^{-\frac{5\pi i}{6}} (- i) e^{- \frac{3\pi i}{4}} = e^{- \frac{\pi i}{12}}$ and $e^{-\frac{5\pi i}{6}} (- \omega ) e^{- \frac{3\pi i}{4}} = e^{\frac{\pi i}{12}}$. Hence, 
\begin{align}
\begin{aligned}
    w_0(x) &= 3^{1/4} \Re \left( e^{-\frac{\pi i}{12}} \frac{e^{-2 i \sqrt{3}x}}{s_1 \Gamma (- \nu)} \sqrt{\frac{\pi}{x}} \, (24 \sqrt{3} x)^{- \nu} e^{\frac{\pi i }{2} \nu} \right) + \O \left( \frac{1}{x} \right)\\
    &= 3^{1/4} \frac{1}{|s_1| |\Gamma(-\nu)|} \sqrt{\frac{\pi}{x}} e^{\frac{\pi i \nu}{2}}\\
    &\qquad \times \cos \left( 2\sqrt{3}x + \frac{\pi}{12} + \arg s_1 - \arg \Gamma(\nu) - i \nu \ln (24 \sqrt{3}) - i \nu \ln x \right)  + \O \left( \frac{1}{x} \right). 
\end{aligned} \label{w_0(x) expression 2}
\end{align}
If we set
\begin{align*}
    \psi &= - i \nu \ln (24\sqrt{3}) + \frac{\pi}{12} + \arg s_1 - \arg \Gamma(\nu),\\
    \sigma& = 3^{1/4} \frac{\sqrt{\pi}}{|s_1| |\Gamma(- \nu)|} e^{\frac{\pi i \nu}{2}} > 0,
\end{align*}
then \eqref{w_0(x) expression 2} becomes
\begin{align*}
    w_0(x) = \frac{\sigma}{\sqrt{x}} \cos \left( 2 \sqrt{3} x - i \nu \ln x + \psi \right) + \O \left( \frac{1}{x} \right).
\end{align*}
Moreover, using \eqref{useful relations of gamma function}, we have 
\begin{align*}
    \sigma^2 = - \frac{\sqrt{3}}{4 \pi} \ln (1 - |s_1|^2).
\end{align*}

%%%%%%%%%%%%%%%%%%%%%%%%%%%%%

\section{Comparison with Kitaev's Result}\label{Kitaev}
 
In Kitaev's paper \cite{Kitaev}, the following change of variables was introduced:
\begin{align*}
 &\tilde{w}=s^{\frac13}e^u\\
 &\tau=\frac{9}{16}s^{\frac43}=x^2,
 \end{align*}
 so \eqref{Painleve III} is equivalent to
 \begin{align*}
  \epsilon (\tau u_{\tau})_\tau=e^u-e^{-2u}
\end{align*}
with $\epsilon=-1$.

Comparing with our formulae, we observe that 
\begin{align*}
    u(\tau)=u(x^2)=-2w_0(x) \; \text{ and } \; \partial_\tau=\frac{1}{2x}\partial_x.
\end{align*}

In \cite{Kitaev}, the following result was obtained:
\begin{align*}
    &e^{u(\tau)}-1\sim -a\sqrt{6}(3\tau)^{-\frac14} \cos\left(2\sqrt{3\tau}+a^2\ln \sqrt{3\tau}+\phi-\frac{\pi}{4}\right), \quad \tau\rightarrow \infty
    \end{align*}
    or in our terms,
    \begin{align*}
        e^{-2w_0(x)}-1\sim -a\sqrt{6}(3x^2)^{-\frac14} \cos\left(2\sqrt{3 }x+a^2\ln \sqrt{3 }x+\phi-\frac{\pi}{4}\right),\quad x\rightarrow \infty,
    \end{align*}
    where 
    \begin{align*}
    &a=\sqrt{\dfrac{|\ln g_3|}{2\pi}}\operatorname{exp}\Big\{\frac12 \arg \ln g_3\Big\}\\
    &\phi=a^2\ln 24-\frac{i}{2} \ln\left|\dfrac{\Gamma(-ia^2)}{\Gamma(ia^2)}\dfrac{g_2+sg_3}{g_1}\right|+\frac12 \arg \dfrac{\Gamma(-ia^2)}{\Gamma(ia^2)}+\frac12\arg\frac{g_2+sg_3}{g_1}
\end{align*}
and the complex parameters  $g_1,g_2,g_3$ satisfy 
\begin{align}\label{g eq}
    g_1+g_2+g_3=1, \qquad g_1g_2+g_2g_3+g_1g_3(1+s)=0.
\end{align}

It seems not clear how to make a correct correspondence between Kitaev's parameters $(g_1, g_2, g_3)$ and our parameters $(A^\R, B, s^\R)$. But, one can make the following two observations:
\begin{enumerate}
\item If we take
\begin{align}
\begin{aligned}
    &g_1=3\omega^2 \bar{B},\\
    &g_2=1-3A^\R-3\omega^2 \bar{B} =3\omega B+3As^\R,\\
    &g_3=3A^\R,\\
    &s    =  - s^\R,
\end{aligned} \label{idea 1}
\end{align}
then \eqref{g eq} becomes equivalent to \eqref{identity 7 at parametrization of E_1} and \eqref{identity 6 at parametrization of E_1}, and it holds that
\begin{align*}
    w_0(x) = -\frac{\sigma}{\sqrt{x}} \cos \left( 2 \sqrt{3} x - i \nu \ln x + \psi \right) + \O \left( \frac{1}{x} \right)
\end{align*}
which implies our result differs by a minus sign. Although the correspondence \eqref{idea 1} seems natural, one might give an alternative change of parameters as follows:
\item If we take
\begin{align}
\begin{aligned}
    &g_1=-3\omega^2 \bar{B}, \\
    &g_2=1-3\omega^2\bar{B}-3A^\R,\\
    &g_3=3A^\R,\\
    &s =  2+ s^\R-\frac{2}{3A^\R},
\end{aligned} \label{idea 2}
\end{align}
then \eqref{g eq} is satisfied by \eqref{identity 7 at parametrization of E_1} and \eqref{identity 6 at parametrization of E_1}, and we have 
\begin{align*}
    w_0(x) = \frac{\sigma}{\sqrt{x}} \cos \left( 2 \sqrt{3} x - i \nu \ln x + \psi \right) + \O \left( \frac{1}{x} \right).
\end{align*}
This asymptotics is exactly what we have.
%, but it is unclear to suggest the correspondence \eqref{idea 2}.
\end{enumerate}